\documentclass{lmcs}
\pdfoutput=1

\usepackage{lastpage}
\usepackage[utf8]{inputenc}
\lmcsdoi{19}{4}{39}
\lmcsheading{}{\pageref{LastPage}}{}{}%
{Oct.~18,~2022}{Dec.~25,~2023}{}

\pdfoutput=1


\usepackage{hyperref}
\usepackage{booktabs}
\usepackage{amsmath}
\usepackage{amsthm}
\usepackage{xcolor}
\usepackage{adjustbox}
\usepackage{calc}
\usepackage{tikz}
\usetikzlibrary{decorations.pathreplacing,angles,quotes,positioning,shapes.misc,decorations.text}
\usepackage{xspace}
\usepackage{url}

\usepackage{wasysym}

\usepackage[normalem]{ulem} 

\usepackage[ruled,vlined,noend]{algorithm2e}
\usepackage{listings}
\usepackage{caption}

\usepackage{verbatimbox}

\renewcommand{\epsilon}{\varepsilon}

\newcommand{\db}{\textbf{D}}
\newcommand{\type}[1]{\text{type}(#1)}
\newcommand{\schRel}{\textsf{Rels}}
\newcommand{\schFunc}{\textsf{Funcs}}

\newcommand{\objectsRes}[1]{\objects_{#1}}

\newcommand{\Cset}{\ensuremath{{\Gamma}}}
\newcommand{\Attr}[1]{\ensuremath{\text{Attr}(#1)}}
\newcommand{\myvar}[1]{\ensuremath{var}(#1)}

\newcommand{\classC}{\ensuremath{\mathcal{C}}}
\newcommand{\basictemplates}{\textbf{VarTemp}\xspace}
\newcommand{\equaltemplates}{\textbf{EqTemp}\xspace}
\newcommand{\mtbtemplates}{\textbf{MTBTemp}\xspace}
\newcommand{\acyctemplates}{\textbf{AcycTemp}\xspace}

\newcommand{\acycrestemplates}{\textbf{AcycResTemp}\xspace}

\newcommand{\mydef}{\mathrel{:=}}

\newcommand{\eg}{\textit{e.g.}}
\newcommand{\ssize}[1]{|#1|}

\newcommand{\trans}[1][i]{T_{#1}}

\newcommand{\transset}{{\mathcal{T}}}
\newcommand{\schedule}{s}
\newcommand{\objects}{\textbf{Tuples}}

\newcommand{\variables}{\mathbf{Var}}
\newcommand{\templ}[1][i]{\tau_{#1}}
\newcommand{\workload}{\mathcal{P}}

\newcommand{\cg}[1]{CG(#1)}

\newcommand{\multitree}{multi-tree\xspace}

\newcommand{\prefix}[2]{{\normalfont\textsf{prefix}}_{#2}(#1)}

\newcommand{\postfix}[2]{{\normalfont\textsf{postfix}}_{#2}(#1)}

\newcommand{\restricted}{restricted\xspace}

\newcommand{\mvrc}{RC\xspace}
\newcommand{\MVRC}{\mvrc}
\newcommand{\rc}{\mvrc}

\newcommand{\readun}{{\normalfont\textsc{read uncommitted}}\xspace}
\newcommand{\readcom}{{\normalfont\textsc{read committed}}\xspace}

\newcommand{\snapshot}{{\normalfont\textsc{snapshot isolation}}\xspace}

\newcommand{\isolationlevel}{\text{$\mathcal{I}$}}

\newcommand{\NLOGSPACE}{{\sc nlogspace}\xspace}
\newcommand{\LOGSPACE}{{\sc logspace}\xspace}

\newcommand{\PTIME}{{\sc ptime}\xspace}

\newcommand{\coNP}{{\rm co{\sc np}}\xspace}

\newcommand{\EXPSPACE}{{\sc expspace}\xspace}
\newcommand{\NEXPSPACE}{{\sc nexpspace}\xspace}
\newcommand{\PSPACE}{{\sc pspace}\xspace}

\newcommand{\EXPTIME}{{\sc exptime}\xspace}

\newcommand{\RobustnessTemp}[1]{{\sc t-robustness}(#1)\xspace}
\newcommand{\RobustnessTempTwo}[2]{{\sc t-robustness}(#1,#2)\xspace}

\newcounter{conditioncounter}

\newcommand{\new}[1]{{#1}}
\newcommand{\myR}{\ensuremath{\mathtt{R}}}
\newcommand{\myW}{\ensuremath{\mathtt{W}}}
\newcommand{\myUP}{\ensuremath{\mathtt{U}}}
\newcommand{\R}[2][i]{\myR_{#1}\mathtt{[#2]}}
\newcommand{\W}[2][i]{\myW_{#1}\mathtt{[#2]}}
\newcommand{\UP}[2][i]{\myUP_{#1}\mathtt{[#2]}}
\newcommand{\CT}[1][i]{\mathtt{C}_{#1}}

\newcommand{\ReadSet}[1]{\ensuremath{\text{ReadSet}(#1)}}
\newcommand{\WriteSet}[1]{\ensuremath{\text{WriteSet}(#1)}}

\newcommand{\mya}{{\tt a}}
\newcommand{\mys}{{\tt s}}
\newcommand{\myc}{{\tt c}}
\newcommand{\x}{\mathtt{t}}
\newcommand{\y}{\mathtt{v}}
\newcommand{\z}{\mathtt{q}}

\newcommand{\dom}[1]{\text{\it dom}(#1)}
\newcommand{\range}[1]{\text{\it range}(#1)}
\newcommand{\domf}{\text{\it dom}}
\newcommand{\rangef}{\text{\it range}}
\newcommand{\myvs}{S}
\newcommand{\vt}{\mathtt{T}}
\newcommand{\myvv}{\mathtt{V}}

\newcommand{\vw}{\mathtt{W}}
\newcommand{\vx}{\mathtt{X}}
\newcommand{\vy}{\mathtt{Y}}
\newcommand{\vz}{\mathtt{Z}}
\newcommand{\vq}{\mathtt{Q}}

\newcommand{\fc}[1]{f_{#1}}
\newcommand{\fas}{\fc{A\to S}}
\newcommand{\fsa}{\fc{S\to A}}
\newcommand{\fac}{\fc{A\to C}}
\newcommand{\fca}{\fc{C\to A}}
\newcommand{\f}{\fas}
\newcommand{\g}{\fac}

\newcommand{\ListAttr}[1]{\ensuremath{\{\text{#1}\}}}
\newcommand{\Account}{\ensuremath{\text{Account}}\xspace}
\newcommand{\Savings}{\ensuremath{\text{Savings}}\xspace}
\newcommand{\Checking}{\ensuremath{\text{Checking}}\xspace}

\newcommand{\CustID}{\ensuremath{\text{CustomerID}}\xspace}

\newcommand{\Balance}{\ensuremath{\text{Balance}}\xspace}

\newcommand{\GoPremium}{\ensuremath{\text{GoPremium}}\xspace}

\newcommand{\Warehouse}{\ensuremath{\text{Warehouse}}\xspace}
\newcommand{\District}{\ensuremath{\text{District}}\xspace}
\newcommand{\Customer}{\ensuremath{\text{Customer}}\xspace}
\newcommand{\Order}{\ensuremath{\text{Order}}\xspace}
\newcommand{\OrderLine}{\ensuremath{\text{OrderLine}}\xspace}
\newcommand{\Stock}{\ensuremath{\text{Stock}}\xspace}
\newcommand{\fdisttowh}{\fc{D\to W}}
\newcommand{\fcusttodist}{\fc{C\to D}}
\newcommand{\fordtocust}{\fc{O\to C}}
\newcommand{\flinetoord}{\fc{L\to O}}
\newcommand{\flinetostock}{\fc{L\to S}}
\newcommand{\fstocktowh}{\fc{S\to W}}

\newcommand{\sstart}{\textit{op}_0}

\newcommand{\tmap}{\mu}
\newcommand{\barmu}{\bar\mu}
\newcommand{\typemapset}{\mathcal{S}}
\newcommand{\typemapfunc}[1][\typemapset]{\varphi_{#1}}

\newcommand{\sgname}{\textit{SG}}
\newcommand{\sg}[1]{\sgname(#1)}

\newcommand{\asgimplies}[2]{\mathrel{\scalebox{0.8}{$\stackrel{\scalebox{0.5}{$#1$}}{\rightsquigarrow}$}_{#2}}}
\newcommand{\varimplies}[2]{\mathrel{\scalebox{0.8}{$\stackrel{\scalebox{0.5}{$#1$}}{\rightsquigarrow}$}_{#2}}}
\newcommand{\varequiv}[1]{\mathrel{\equiv_{#1}}}
\newcommand{\vardet}[2]{\mathrel{\scalebox{0.8}{$\stackrel{\scalebox{0.5}{$#1$}}{\Rightarrow}$}_{#2}}}
\newcommand{\conflictquadruple}{conflicting quadruple}
\newcommand{\conflictquadruples}{conflicting quadruples}
\newcommand{\pconflictquadruple}{potentially conflicting quadruple}
\newcommand{\pconflictquadruples}{potentially conflicting quadruples}
\newcommand{\transcopies}[1]{\text{Trans}(#1)}
\newcommand{\cqsequence}{sequence of \conflictquadruples{}}
\newcommand{\pcqsequence}{sequence of \pconflictquadruples{}}
\newcommand{\contextset}{\mathcal{A}}
\newcommand{\contextfunc}[1][\contextset]{\varphi_{#1}}
\newcommand{\tcontextdomain}[1][]{\textsf{TupleContexts}(#1)}
\newcommand{\contextdomain}[1][]{\textsf{Contexts}(#1)}
\newcommand{\conn}[3]{\ensuremath{#1\,\aquarius_{#3}\,#2}}
\newcommand{\notconn}[3]{\ensuremath{#1\,\not\hspace{-1ex}{\aquarius}_{#3}\,#2}}

\newcommand{\vs}[1][]{\ensuremath{\mathtt{S_\mathsf{#1}}}}
\newcommand{\vi}{\ensuremath{\mathtt{I}}}

\newcommand{\vc}{\ensuremath{\mathtt{C}}}

\newcommand{\vb}{\ensuremath{\mathtt{B}}}
\newcommand{\dominoset}{\mathcal{D}}

\newcommand{\ufunc}[2][f]{{#1}_{\text{#2}}}

\newcommand{\ttinit}{\ensuremath{\text{Split}}}
\newcommand{\ttinitb}{\ensuremath{\text{First}}}
\newcommand{\ttclose}{\ensuremath{\text{Last}}}
\newcommand{\ttdomino}[1]{\ensuremath{\text{Domino}_{\vec{#1}}}}
\renewcommand{\vec}[1]{\mathbf{#1}}

\newcommand{\ptrans}{transaction template}   % {parameterized transaction} - lowercase singular
\newcommand{\ptranss}{transaction templates} % {parameterized transactions} - lowercase multiple
\newcommand{\Ptranss}{Transaction templates} % {Parameterized transactions} - uppercase multiple
\newcommand{\PTranss}{Transaction Templates} % {Parameterized Transactions} - title multiple
\newcommand{\shortptrans}{template}   % {parameterized transaction} - lowercase singular
\newcommand{\shortptranss}{templates} % {parameterized transactions} - lowercase multiple
\newcommand{\shortPtrans}{Template} % {Parameterized transaction} - uppercase singular
\newcommand{\shortPtranss}{Templates} % {Parameterized transactions} - uppercase multiple
\newcommand{\shortPTranss}{Templates} % {Parameterized Transactions} - title multiple
\def\eg{{\em e.g.}}

\begin{document}

\title[Robustness against RC for Templates with Functional Constraints]{Robustness against Read Committed for Transaction Templates with Functional Constraints}

%OPTIONAL comment concerning the title, e.g., if a variant or an extended abstract of the paper has appeared elsewhere.
\titlecomment{{\lsuper*}The present paper is the full version of~\cite{DBLP:conf/icdt/VandevoortK0N22} and supplies all proofs.}
%\thanks{This work is funded by FWO-grant G019921N.}	%optional

% affiliations are numbered automatically with a, b, c (see below)
% use the optional argument to indicate the affiliation(s) of each author
% omit the argument if there is only one author, or only one affiliation
\author[B.~Vandevoort]{Brecht Vandevoort\lmcsorcid{0000-0001-7212-4625}}[a]
\author[B.~Ketsman]{Bas Ketsman\lmcsorcid{0000-0002-4032-0709}}[b]
\author[C.~Koch]{Christoph Koch\lmcsorcid{0000-0002-9130-7205}}[c]
\author[F.~Neven]{Frank Neven\lmcsorcid{0000-0002-7143-1903}}[a]

% affiliation 1 (automatically numbered a)
\address{UHasselt, Data Science Institute, ACSL, Belgium}	%optional
% write emails for all authors having that affiliation
\email{brecht.vandevoort@uhasselt.be, frank.neven@uhasselt.be}  %optional

% affiliation 2 (automatically numbered b)
\address{Vrije Universiteit Brussel, Belgium}	%optional
\email{bas.ketsman@vub.be}  %optional

% affiliation 3 (automatically numbered c)
\address{\'Ecole Polytechnique F\'ed\'erale de Lausanne, Switzerland}	%optional
\email{christoph.koch@epfl.ch}  %optional

%% etc.

%% required for running head on odd and even pages, use suitable
%% abbreviations in case of long titles and many authors:

%%%%%%%%%%%%%%%%%%%%%%%%%%%%%%%%%%%%%%%%%%%%%%%%%%%%%%%%%%%%%%%%%%%%%%%%%%%

%% the abstract has to PRECEDE the command \maketitle:
%% be sure not to issue the \maketitle command twice!

% mandatory: add short abstract of the document
\begin{abstract}
    The popular isolation level Multiversion Read Committed (RC) trades some of the strong guarantees of serializability for increased transaction throughput.
    Sometimes, transaction workloads can be safely executed under RC obtaining serializability at the lower cost of RC.
    Such workloads are said to be robust against RC. Previous work has yielded a tractable procedure for deciding robustness against RC for workloads generated by transaction programs modeled as transaction templates. An important insight of that work is that, by more accurately modeling transaction programs, we are able to recognize larger sets of workloads as robust. 
    In this work, we increase the modeling power of transaction templates by extending them with functional constraints, which are useful for capturing data dependencies like foreign keys. We show that the incorporation of functional constraints can identify more workloads as robust that otherwise would not be. Even though we establish that the robustness problem becomes undecidable in its most general form, we show that various restrictions on functional constraints lead to decidable and even tractable fragments that can be used to model and test for robustness against RC for realistic scenarios.
\end{abstract}

\maketitle

\section{Introduction}
\label{sec:intro}

Many database systems implement several isolation levels, allowing users to trade isolation guarantees for improved performance. The highest, serializability, projects the appearance of a complete absence of concurrency, and thus perfect isolation. Executing transactions concurrently under weaker isolation levels can introduce certain anomalies. Sometimes, 
a transactional workload can be executed at an isolation level lower than serializability without introducing any anomalies. This is a desirable scenario: a lower isolation level, usually implementable with a cheaper concurrency control algorithm, yields the stronger isolation guarantees of serializability 
for free.
This formal property is called robustness \cite{DBLP:conf/pods/Fekete05,DBLP:conf/concur/0002G16}: a set of transactions $\transset$ is called \emph{robust against a given isolation level} if every possible interleaving of the transactions in $\transset$ that is {allowed} under the specified isolation level is serializable.

Robustness received quite a bit of attention in the literature.
Most existing work focuses on Snapshot Isolation (SI)~\cite{Alomari:2008:CSP:1546682.1547288,DBLP:conf/cav/BeillahiBE19,DBLP:conf/pods/Fekete05,DBLP:journals/tods/FeketeLOOS05} or higher isolation levels~\cite{DBLP:conf/concur/BeillahiBE19,DBLP:conf/concur/0002G16,DBLP:conf/concur/Cerone0G15,cerone_et_al:LIPIcs:2017:7794}. 
It is particularly interesting to consider robustness against lower level isolation levels like multi-version Read Committed (referred to as \MVRC from now on). Indeed,  \MVRC is widely available, often the default in database systems (see, e.g.,~\cite{DBLP:journals/pvldb/BailisDFGHS13}), and is generally expected to have better throughput than stronger isolation levels. 

In previous work~\cite{fullversion}, we provided a tractable decision procedure for robustness against \MVRC for workloads generated by transaction programs modeled as transaction templates. The approach is centered on a novel characterization of robustness against \MVRC in the spirit of \cite{DBLP:conf/pods/Fekete05,DBLP:conf/pods/Ketsman0NV20} that improves over the sufficient condition presented in \cite{DBLP:conf/aiccsa/AlomariF15}, and on a formalization of transaction programs, called {\it transaction templates}, facilitating fine-grained reasoning for robustness against \MVRC. 
Conceptually, {transaction templates} as introduced in \cite{fullversion} are functions with parameters, and can, for instance, be derived from stored procedures inside a database system (c.f.\ Figure~\ref{fig:smallbank-abstract-syntax}  for an example).  The abstraction generalizes transactions as usually studied in concurrency control research -- sequences of read and write operations -- by making the objects worked on variable, determined by input parameters. Such parameters are {\it typed} to add additional power to the analysis. They support {\it atomic updates} (that is, a read followed by a write of the same database object, to make a relative change to its value). Furthermore, database objects read and written are considered at the granularity of fields, rather than just entire tuples, decoupling conflicts further and allowing to recognize additional cases that would not be recognizable as robust on the tuple level. 

An important insight obtained from \cite{fullversion} is that more accurate modeling of the workload allows to recognize larger sets of transaction programs as robust. Processing workloads under RC increases the throughput of the transactional database system compared to when executing the workload under SI or serializable SI, so larger robust sets mean better performance of the database system.
In this work, we increase the modeling power of transaction templates by extending them with {\it functional constraints}, which are useful for capturing data dependencies like foreign keys (inclusion dependencies). This appears to be a sweet spot for strengthening modelling power -- as we show in this paper,
it allows us to remain with abstractions that have been well established within database theory, without having to move to general program analysis,
and it pushes the robustness frontier on popular transaction processing benchmarks. Generally speaking, workloads can profit more from richer modelling the larger and more complex they get, so the fact that adding functional constraints yields larger robust sets already on these simple benchmarks suggests that these techniques are practically useful.
Our contributions can be summarized as follows:
\begin{itemize}
\item We argue in Section~\ref{sec:example} through the SmallBank and TPC-C benchmarks that the incorporation of functional constraints can identify more workloads as robust that otherwise would not be, and that they reduce the extent to which changes need to be made to workloads to make them robust against RC.

\item In Section~\ref{sec:undec}, we establish that robustness in its most general form becomes undecidable. The proof is a reduction from PCP and relies on cyclic dependencies between functions allowing to connect data values through an unbounded application of functions. 

\item We consider a fragment in Section~\ref{sec:MTBij} that only allows a very limited form of cyclic dependencies between functions and assumes additional constraints on templates that, together, imply that functions behave as bijections. Robustness against RC can be decided in \NLOGSPACE and this fragment is general enough to model the SmallBank benchmark.

\item In Section~\ref{sec:acyclic}, we obtain an \EXPSPACE decision procedure when the schema graph is acyclic (so, no cyclic dependencies between functions).
Even for small input sizes, such a result is not practical. We provide various restrictions that lower the complexity to \PSPACE and \EXPTIME, and which allow to model the TPC-C benchmark as discussed.
Notice that, for robustness testing, an exponential time decision procedure is considered to be practical as the size of the input is small and robustness is a static property that can be tested offline. 
\end{itemize}

\noindent These contributions should be contrasted with our earlier work~\cite{fullversion}, where we focused on a characterization for robustness against \mvrc for basic \ptranss{} without functional constraints and performed an experimental study to show how the robustness property can improve transaction throughput.

%\section{Application}
\section{\new{Benchmarks}}
\label{sec:example}

We present a small extension of the SmallBank benchmark~\cite{Alomari:2008:CSP:1546682.1547288} to exemplify the modeling power of transaction templates and discuss how the addition of functional constraints can detect larger sets of transaction templates to be robust. Finally, we discuss in the context of the TPC-C benchmark how the incorporation of functional constraints requires less changes to templates in making them robust.

The SmallBank schema consists of three tables: 
\emph{Account(\underline{Name}, CustomerID, IsPremium)}, 
\emph{Savings(\underline{CustomerID}, Balance, InterestRate)}, and
\emph{Checking(\underline{CustomerID}, Balance)}.
Underlined attributes are primary keys.
The \emph{\Account} table associates customer names with IDs {and keeps track of the premium status (Boolean)}; \emph{CustomerID} is a {\tt UNIQUE} attribute. The other tables contain the balance (numeric value) of the savings and checking accounts of customers identified by their ID. \emph{\Account(CustomerID)} is a foreign key referencing both the columns
\emph{\Savings(CustomerID)} and \emph{\Checking(CustomerID)}.
The interest rate on a savings account is based on a number of parameters, including the account status (premium or not).
The application code can interact with the database through a fixed number of transaction programs:
\begin{itemize}
    \item  Balance($N$): returns the total balance (savings \& checking) for a customer with name $N$.

    \item  DepositChecking($N$,$V$): makes a deposit of amount $V$ 
    on the checking account of the customer with name $N$.

    \item TransactSavings($N$,$V$): makes a deposit or withdrawal $V$ on the savings account of the customer with name $N$.

    \item Amalgamate($N_1$,$N_2$): transfers all the funds from 
    $N_1$ to  
    $N_2$.

    \item WriteCheck($N$,$V$): writes a check $V$ against the account of the customer with name $N$, penalizing if overdrawing.

    \item
    GoPremium($N$): converts the account of the customer with name $N$ to a premium account and updates the interest rate of the corresponding savings account. { This transaction program is an extension w.r.t.~\cite{Alomari:2008:CSP:1546682.1547288}.  }
\end{itemize}
The transaction templates for these programs are presented in Figure~\ref{fig:smallbank-abstract-syntax}.
The corresponding SQL code is given in Appendix~\ref{sec:app:smallbank}.

\begin{figure*}[t]

    \begin{minipage}[c]{0.99\textwidth}
        \centering\small
    
    \begin{minipage}[t]{0.38\textwidth-2ex}
    \Balance:
    \[
    \begin{array}{l}
    \R[]{\vx: \Account\ListAttr{N, C}}\\
    \R[]{\vy: \Savings\ListAttr{C, B}}\\
    \R[]{\vz: \Checking\ListAttr{C, B}}\\
    \vy = \fas(\vx),\ \vx = \fsa(\vy)\\
    \vz = \fac(\vx),\ \vx = \fca(\vz)
    \end{array}
    \]
    \end{minipage}
    \begin{minipage}[t]{0.31\textwidth}
    DepositChecking:
    \[
    \begin{array}{l}
    \R[]{\vx: \Account\ListAttr{N, C}}\\
    \UP[]{\vz: \Checking\ListAttr{C, B}\ListAttr{B}}\\
    \vz = \fac(\vx),\ \vx = \fca(\vz)
    \end{array}
    \]
    \end{minipage}
    \begin{minipage}[t]{0.31\textwidth}
    TransactSavings:
    \[
    \begin{array}{l}
    \R[]{\vx: \Account\ListAttr{N, C}}\\
    \UP[]{\vy: \Savings\ListAttr{C, B}\ListAttr{B}}\\
    \vy = \fas(\vx),\ \vx = \fsa(\vy)
    \end{array}
    \]
    \end{minipage}
    
    \medskip
    
    \begin{minipage}[t]{0.38\textwidth-2ex}
    Amalgamate:
    \[
    \begin{array}{l}
    \R[]{\vx_1: \Account\ListAttr{N, C}}\\
    \R[]{\vx_2: \Account\ListAttr{N, C}}\\
    \UP[]{\vy_1: \Savings\ListAttr{C, B}\ListAttr{B}}\\
    \UP[]{\vz_1: \Checking\ListAttr{C, B}\ListAttr{B}}\\
    \UP[]{\vz_2: \Checking\ListAttr{C, B}\ListAttr{B}}\\
    \vx_1\neq \vx_2,\\
    \vy_1 = \fas(\vx_1),\ \vx_1 = \fsa(\vy_1)\\
    \vy_2 = \fas(\vx_2),\ \vx_2 = \fsa(\vy_2)\\
    \vz_1 = \fac(\vx_1),\ \vx_1 = \fca(\vz_1)\\
    \vz_2 = \fac(\vx_2),\ \vx_2 = \fca(\vz_2)
    \end{array}
    \]
    \end{minipage}
    \begin{minipage}[t]{0.31\textwidth}
    WriteCheck:
    \[
    \begin{array}{l}
    \R[]{\vx: \Account\ListAttr{N, C}}\\
    \R[]{\vy: \Savings\ListAttr{C, B}}\\
    \R[]{\vz: \Checking\ListAttr{C, B}}\\
    \UP[]{\vz: \Checking\ListAttr{C, B}\ListAttr{B}}\\
    \vy = \fas(\vx),\ \vx = \fsa(\vy)\\
    \vz = \fac(\vx),\ \vx = \fca(\vz)
    \end{array}
    \]
    \end{minipage}
    \begin{minipage}[t]{0.31\textwidth}
    GoPremium:
    \[
    \begin{array}{l}
    \UP[]{\vx: \Account\ListAttr{N, C}\ListAttr{I}}\\
    \R[]{\vy: \Savings\ListAttr{C, I}}\\
    \UP[]{\vy:\Savings\ListAttr{C}\ListAttr{I}}\\
    \vy = \fas(\vx),\ \vx = \fsa(\vy)
    \end{array}
    \]
    \end{minipage}

    \caption{Transaction templates for SmallBank.}
    \label{fig:smallbank-abstract-syntax}
\end{minipage}
\end{figure*}

\new{Based on this benchmark, we give an informal description of transaction templates and functional constraints to illustrate their modeling power. More formal definitions can be found in Section~\ref{sec:defs}.}
In short, a transaction template is a sequence of read (\myR), write (\myW) and update (\myUP) statements over typed variables ($\vx$, $\vy$, \ldots) with additional equality and disequality constraints. For instance, $\R[]{\vy :\Savings\{C,I\}}$ in \GoPremium{} indicates that a read operation is performed to a tuple in relation \emph{\Savings} on the attributes \emph{CustomerID} and \emph{InterestRate}. We abbreviate the names of attributes by their first letter to save space. The set $\{C,I\}$ is the read set. 
Write operations have an associated write set while update operations contain a read set followed by a write set: e.g.,
$\UP[]{\vx :\Account\{N,C\}\{I\}}$ in \GoPremium{} first reads the \emph{Name} and \emph{CustomerID} of tuple $\vx$ and then writes to the attribute \emph{InterestRate}. 
To capture the dependencies between tuples induced by the foreign keys, we use two unary functions:  $\f$ maps a tuple of type \emph{\Account} to a tuple of type \emph{\Savings}, while $\g$ maps a tuple of type \emph{\Account} to a tuple of type \emph{\Checking}. As \emph{\Account{(CustomerID)}} is {\tt UNIQUE}, every savings and checking account is associated to a unique \emph{\Account} tuple. This is modelled through the functions $\fca$ and $\fsa$ with an analogous interpretation. Notice that the equality constraints for each template in Figure~\ref{fig:smallbank-abstract-syntax} imply that these functions are bijections and each others inverses.

A transaction $T$ over a database $\bf D$ is an \emph{instantiation} of a transaction template $\tau$ if there is a variable mapping $\mu$ from the variables in $\tau$ to tuples in $\bf D$ that satisfies all the constraints in $\tau$ \new{such that} $\mu(\tau)=T$.  
For instance, consider a database $\bf D$ with tuples $\mya_1,\mya_2,\ldots$ of type \emph{\Account{}}, 
$\mys_1,\mys_2,\ldots$ of type \emph{\Savings{}}, and $\myc_1,\myc_2,\ldots$ of type \emph{\Checking{}} with $\fas^{\bf D}(\mya_i)=\mys_i$, $\fac^{\bf D}(\mya_i)=\myc_i$, $\fsa^{\bf D}(\mys_i)=\mya_i$, $\fca^{\bf D}(\myc_i)=\mya_i$ for each $i$.
Then, for $\mu_1=\{\vx\to \mya_1, \vy\to\mys_1\}$, 
$\mu_1(\GoPremium)=\UP[]{\mya_1}\R[]{\mys_1}\UP[]{\mys_1}$ is an instantiation of GoPremium whereas $\mu_2(\GoPremium)$ with $\mu_2=\{\vx\to \mya_1, \vy\to\mys_2\}$ is not as the functional constraint $\vy=\fas(\vx)$ is not satisfied. Indeed, $\mu_2(\vy)=\mys_2\neq\mys_1=\fas^{\bf D}(\mya_1)=\fas^{\bf D}(\mu_2(\vx))$.
We then say that a set of transactions is \emph{consistent} with a set of templates if every transaction is an instantiation of a transaction template.
\new{More formal definitions are given in Section~\ref{sec:defs}.}

\new{Functional constraints are different from the more usual data consistency constraints}
like key constraints, functional dependencies or denial constraints, \new{etc}. The latter are intended to verify data consistency, whereas the former are intended to verify whether a set of transactions instantiated from \shortptranss{} are indeed consistent with these \shortptranss{}.
The abstraction of functional constraints provides a straightforward mechanism to capture dependencies between tuples implied by e.g.\ foreign key constraints. Consider for example variables $\vx$ and $\vy$ in \GoPremium. Rather than specifying that the value of the attribute \emph{\CustID} in the tuple assigned to $\vx$ should agree with the value of the attribute \emph{\CustID} in the tuple assigned to $\vy$ and combining this information with the defined foreign key from \emph{\Account} to \emph{\Savings} to conclude that two instantiations of \GoPremium that agree on the tuple assigned to $\vx$ should also agree on the tuple assigned to $\vy$, the functional constraint $\vy = \fas(\vx)$ expresses this dependency more directly.
An additional benefit of our abstraction is that this approach is not limited to dependencies implied by foreign keys. For the SmallBank benchmark, for example, we can infer from the fact that \emph{\Account{(CustomerID)}} is {\tt UNIQUE} that each checking and savings account is associated to exactly one \emph{\Account} tuple, even though no foreign key from respectively \emph{\Checking} and \emph{\Savings} to \emph{\Account} is defined in the schema.
\new{Since functional constraints are based on unary functions, they are limited to expressing tuples being implied by a single other tuple (e.g., the \emph{\Savings} tuple being implied by the \emph{\Account} tuple in \GoPremium). More complex relationships where a tuple is implied by the co-occurrence of two or more other tuples cannot be captured by our formalism.}

Our previous work~\cite{fullversion}, which did not consider functional constraints, has shown that \{Am,DC,TS\}, \{Bal,DC\}, and \{Bal,TS\} are maximal robust sets of transaction templates. This means that for any database, for any set of transactions $\transset$ that is consistent with one of the three mentioned sets, any possible interleaving of the transactions in $\transset$ that is allowed under RC is \emph{always} serializable! 
Using the results from Section~\ref{sec:MTBij}, it follows that when functional constraints are taken into account GoPremium can be added to each of these sets as well: \{Am,DC,GP,TS\}, \{Bal,DC,GP\}, \{Bal,TS,GP\} are maximal robust sets.

We argue that incorporating functional constraints is crucial. Indeed, without functional constraints it's easy to show that even the set $\{\GoPremium\}$ is not robust. Consider the schedule  over two instantiations $\trans[1]$ and $\trans[2]$ of \GoPremium, where we use the mappings $\mu_1$
and $\mu_2$ as defined above 
for respectively $\trans[1]$ and $\trans[2]$ (we show the read and write sets to facilitate the discussion):
$$
    \arraycolsep=.1em
    \begin{array}{lcccc}
    \trans[1]:\,  \UP[1]{\mya_1\ListAttr{N,C}\ListAttr{I}} \, \R[1]{\mys_1\ListAttr{C,I}} \, & & \UP[1]{\mys_1\ListAttr{C}\ListAttr{I}} \, \CT[1]\\
    \trans[2]:& \UP[2]{\mya_2\ListAttr{N,C}\ListAttr{I}} \, \R[2]{\mys_1\ListAttr{C,I}} \, \UP[2]{\mys_1\ListAttr{C}\ListAttr{I}} \, \CT[2]&
    \end{array}
$$
The above schedule is allowed under RC as there is no dirty write, but it is not conflict serializable. Indeed, there is a rw-conflict between $\R[1]{\mys_1\ListAttr{C,I}}$ and $\UP[2]{\mys_1\ListAttr{C}\ListAttr{I}}$ as the former reads the attribute $I$ that is written to by the latter, which implies that $T_1$ should occur before $T_2$ in an equivalent serial schedule. But, there is a ww-conflict between $\UP[2]{\mys_1\ListAttr{C}\ListAttr{I}}$ and $\UP[1]{\mys_1\ListAttr{C}\ListAttr{I}}$ as both write to 
the common attribute $I$ implying that $T_2$ should occur before $T_1$ in an equivalent serial schedule. Consequently, the schedule is not serializable.
However, taking functional constraints into account, $\{T_1,T_2\}$ is not consistent with $\{\GoPremium\}$ as $\mu_2(\vy)=\mys_1\neq \mys_2 = \fas(a_2)=\fas(\mu_2(\vx))$ implying that the above schedule is \emph{not} a counterexample for robustness.

The second benchmark is based on the TPC-C benchmark~\cite{TPCC}. We modified the schema and \shortptranss{} to turn all predicate reads into key-based accesses.
The schema consists of six relations:
\begin{itemize}
    \item \emph{Warehouse(\underline{WarehouseID}, Info, YTD)},
    \item \emph{District(\underline{WarehouseID}, \underline{DistrictID}, Info, YTD, NextOrderID)},
    \item \emph{Customer(\underline{WarehouseID}, \underline{DistrictID}, \underline{CustID}, Info, Balance)},
    \item \emph{Order(\underline{WarehouseID}, \underline{DistrictID}, \underline{OrderID}, CustID, Status)},
    \item \emph{OrderLine(\underline{WarehouseID}, \underline{DistrictID}, \underline{OrderID}, \underline{OrderLineID}, ItemID, DeliveryInfo,\phantom{m} Quantity)}, and
    \item \emph{Stock(\underline{WarehouseID}, \underline{ItemID}, Quantity)}.
\end{itemize}
The function names belonging to this schema are given in Table~\ref{table:functionstpcc}.

\begin{table}[t]
\centering
\begin{tabular}{ c c  c } 
    \toprule
    $f$ & $\dom{f}$ & $\range{f}$ \\
    \midrule
    $\fdisttowh$ & $\District$ & $\Warehouse$ \\ 
    $\fcusttodist$ & $\Customer$ & $\District$ \\ 
    $\fordtocust$ & $\Order$ & $\Customer$ \\ 
    $\flinetoord$ & $\OrderLine$ & $\Order$ \\ 
    $\flinetostock$ & $\OrderLine$ & $\Stock$ \\ 
    $\fstocktowh$ & $\Stock$ & $\Warehouse$ \\ 
    \bottomrule
\end{tabular}
\caption{Function names for the TPC-C benchmark schema.}
\label{table:functionstpcc}
\end{table}

We focus on five different \ptranss{}:
\begin{itemize}
    \item NewOrder($W$, $D$, $C$, $I_1$, $Q_1$, $I_2$, $Q_2$, \ldots): creates a new order for the customer identified by $(W,D,C)$. The id for this order is obtained by increasing the \emph{NextOrderID} attribute of the \emph{District} tuple identified by $(W,D)$ by one. Each order consists of a number of items $I_1, I_2, \ldots$ with respectively quantities $Q_1, Q_2, \ldots$. For each of these items, a new \emph{OrderLine} tuple is created and the related stock quantity is decreased.
    \item Payment($W$, $D$, $C$, $A$): represents a customer identified by $(W,D,C)$ paying an amount $A$. This payment is reflected in the database by increasing the balance of this customer by $A$. This amount is furthermore added to the YearToDate (\emph{YTD}) income of both the related \emph{Warehouse} and \emph{District} tuples.
    \item OrderStatus($W$, $D$, $C$, $O$): requests information about the current status of the order identified by $(W,D,O)$. This \ptrans{} collects information of the customer identified by $(W,D,C)$ who created the order, the \emph{Order} tuple itself, and the different \emph{OrderLine} tuples related to this order.
    \item Delivery($W$, $D$, $C$, $O$): delivers the order represented by $(W,D,O)$. The status of the order is updated, as well as the \emph{DeliveryInfo} attribute of each \emph{OrderLine} tuple related to this order. The total price of the order is deduced from the balance of the customer who made this order, identified by $(W,D,C)$.
    \item StockLevel($W$, $I$): returns the current stock level of item $I$ in warehouse $W$.
\end{itemize}
A detailed abstraction of each \ptrans{} is given in Figure~\ref{fig:tpc-c}. To shorten the presentation, we only show two orderlines per order.

\begin{figure}[t]
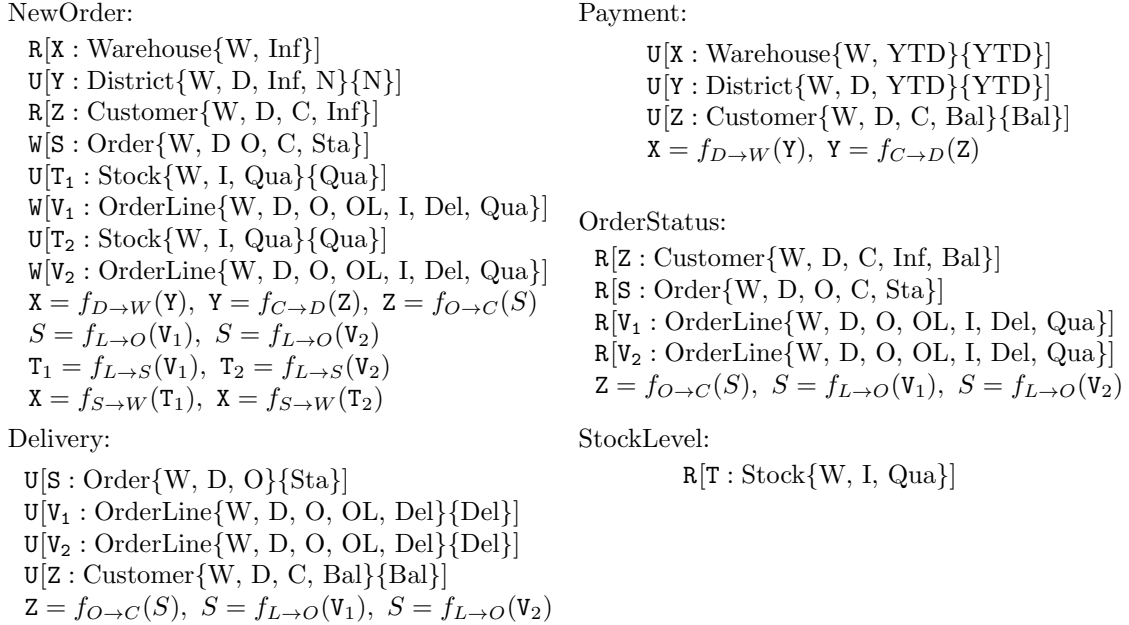

\centering \small

\begin{minipage}[t]{0.49\textwidth}
NewOrder:
\[
\begin{array}{l}
\R[]{\vx: \Warehouse\ListAttr{W, Inf}}\\
\UP[]{\vy: \District\ListAttr{W, D, Inf, N}\ListAttr{N}}\\
\R[]{\vz: \Customer\ListAttr{W, D, C, Inf}}\\
\W[]{\myvs: \Order\ListAttr{W, D O, C, Sta}}\\
\UP[]{\vt_1: \Stock\ListAttr{W, I, Qua}\ListAttr{Qua}}\\
\W[]{\myvv_1: \OrderLine\ListAttr{W, D, O, OL, I, Del, Qua}}\\
\UP[]{\vt_2: \Stock\ListAttr{W, I, Qua}\ListAttr{Qua}}\\
\W[]{\myvv_2: \OrderLine\ListAttr{W, D, O, OL, I, Del, Qua}}\\
\vx = \fdisttowh(\vy),\ \vy = \fcusttodist(\vz),\ \vz = \fordtocust(\myvs)\\
\myvs = \flinetoord(\myvv_1),\ \myvs = \flinetoord(\myvv_2)\\
\vt_1 = \flinetostock(\myvv_1),\ \vt_2 = \flinetostock(\myvv_2)\\
\vx = \fstocktowh(\vt_1),\ \vx = \fstocktowh(\vt_2)
\end{array}
\]
\end{minipage}
\begin{minipage}[t]{0.49\textwidth}
\begin{minipage}[t]{\textwidth}    
Payment:
\[
\begin{array}{l}
\UP[]{\vx: \Warehouse\ListAttr{W, YTD}\ListAttr{YTD}}\\
\UP[]{\vy: \District\ListAttr{W, D, YTD}\ListAttr{YTD}}\\
\UP[]{\vz: \Customer\ListAttr{W, D, C, Bal}\ListAttr{Bal}}\\
\vx = \fdisttowh(\vy),\ \vy = \fcusttodist(\vz)\\
\phantom{\myvs = \flinetoord(\myvv_1),\ \myvs = \flinetoord(\myvv_2)}
\end{array}
\]
\end{minipage}

\medskip

\begin{minipage}[t]{\textwidth}
OrderStatus:
\[
\begin{array}{l}
\R[]{\vz: \Customer\ListAttr{W, D, C, Inf, Bal}}\\
\R[]{\myvs: \Order\ListAttr{W, D, O, C, Sta}}\\
\R[]{\myvv_1: \OrderLine\ListAttr{W, D, O, OL, I, Del, Qua}}\\
\R[]{\myvv_2: \OrderLine\ListAttr{W, D, O, OL, I, Del, Qua}}\\
\vz = \fordtocust(\myvs),\ \myvs = \flinetoord(\myvv_1),\ \myvs = \flinetoord(\myvv_2)
\end{array}
\]
\end{minipage}
\end{minipage}

\medskip

\begin{minipage}[t]{0.49\textwidth}
Delivery:
\[
\begin{array}{l}
\UP[]{\myvs: \Order\ListAttr{W, D, O}\ListAttr{Sta}}\\
\UP[]{\myvv_1: \OrderLine\ListAttr{W, D, O, OL, Del}\ListAttr{Del}}\\
\UP[]{\myvv_2: \OrderLine\ListAttr{W, D, O, OL, Del}\ListAttr{Del}}\\
\UP[]{\vz: \Customer\ListAttr{W, D, C, Bal}\ListAttr{Bal}}\\
\vz = \fordtocust(\myvs),\ \myvs = \flinetoord(\myvv_1),\ \myvs = \flinetoord(\myvv_2)
\end{array}
\]
\end{minipage}
\begin{minipage}[t]{0.49\textwidth}
StockLevel:
\[
\begin{array}{l}
\R[]{\vt: \Stock\ListAttr{W, I, Qua}}\\
\phantom{\myvs = \flinetoord(\myvv_1),\ \myvs = \flinetoord(\myvv_2)}
\end{array}
\]
\end{minipage}

\caption{Abstraction for the TPC-C \ptranss{}. Attribute names are abbreviated.}
\label{fig:tpc-c}
\end{figure}

Incorporating functional constraints for TPC-C can not identify larger sets of templates to be robust. However, 
when a set of \ptranss{} $\workload$ is not robust against \mvrc, an equivalent set of \shortptranss{} $\workload'$ can be constructed from $\workload$ by \emph{promoting} certain $\myR$-operations to $\myUP$-operations~\cite{fullversion}. By incorporating functional constraints it can be shown that fewer $\myR$-operations need to be promoted leading to an increase in throughput as $\myR$-operations do not take locks whereas 
$\myUP$-operations do. 
Consider for example the subset $\workload = \{\text{Delivery}, \text{OrderStatus}\}$ of the TPC-C benchmark, given in Figure~\ref{fig:tpc-c}, where functional constraints are added to express the fact that a tuple of type \emph{OrderLine} implies the tuple of type \emph{Order} (function $\flinetoord$), which in turn implies the tuple of type \emph{Customer} (function $\fordtocust$). 
This set $\workload$ is not robust against \mvrc, but robustness can be achieved by promoting the $\myR$-operation over \emph{Customer} in \emph{OrderStatus} to a $\myUP$-operation.
However, without functional constraints, this single promoted operation no longer guarantees robustness, as witnessed by the following schedule:
$$
    \arraycolsep=.1em
    \begin{array}{lcccc}
    \trans[1] \text{(Orderstatus)}:\,  \UP[1]{c} \, \R[1]{a} \, & & \R[1]{b_1} \, \R[1]{b_2} \, \CT[1]\\
    \trans[2] \text{(Delivery)}:& \UP[2]{a} \, \UP[2]{b_1} \, \UP[2]{b_2} \, \UP[2]{c'} \, \CT[2]&
    \end{array}
$$
Notice in particular how this schedule implicitly assumes in $\trans[2]$ that \emph{Order} $\mathtt{a}$ belongs to \emph{Customer} $\mathtt{c'}$ instead of \emph{Customer} $\mathtt{c}$ to avoid a dirty write on $\mathtt{c}$. Without functional constraints,
$\workload$ is only robust against \mvrc if \emph{all} $\myR$-operations in \emph{OrderStatus} are promoted to $\myUP$-operations.

\section{Definitions}
\label{sec:defs}

We recall the necessary definitions from \cite{fullversion} and extend them with functional constraints.

\subsection{Databases}
A \emph{relational schema} is a pair $(\schRel,\schFunc)$ where $\schRel$ is a set of relation names and $\schFunc$ is a set of function names. 
A finite set of attribute names $\Attr{R}$ is associated to every relation $R\in\schRel$.
Relations will be instantiated by abstract objects that serve as an abstraction of relational tuples. To this end, for every relation $R\in\schRel$, we fix an infinite set of tuples 
$\objectsRes{R}$. 
Furthermore, we assume that $\objectsRes{R} \cap \objectsRes{S} = \emptyset$ for all  $R,S\in\schRel$ with $R\neq S$. We then denote by $\objects$ the set $\bigcup_{R\in\schRel} \objectsRes{R}$ of all possible tuples. Notice that, by definition, for every $\x\in\objects$ there is a unique relation $R \in \schRel$ such that $\x\in\objectsRes{R}$. 
In that case, we say that $\x$ is of \emph{type} $R$ and denote the latter by $\type{\x}=R$. Each function name $f\in \schFunc$ has a domain $\dom{f}\in\schRel$ and a range $\range{f}\in\schRel$. Functions are used to encode relationships between tuples like for instance those implied by foreign-keys constraints. For instance, in the SmallBank example $\schFunc=\{f_{A\to S},f_{A\to C}\}$, $\dom{f_{A\to S}}=\dom{f_{A\to C}}=A$, 
$\range{f_{A\to S}}=S$, and $\range{f_{A\to C}}=C$. 
A \emph{database} $\db$ over schema $(\schRel,\schFunc)$ assigns to every relation name $R\in \schRel$ a finite set $R^{\db} \subset\objectsRes{R}$ and to every function name $f \in \schFunc$ a function $f^{\db}$ from $\dom{f}^{\db}$ to $\range{f}^{\db}$.

\subsection{Transactions and Schedules}
\label{sec:def:trans}

For a tuple $\x \in \objects$, we distinguish three operations $\R[]{\x}$, $\W[]{\x}$, and $\UP[]{\x}$ on $\x$, denoting that tuple $\x$ is read, written, or updated, respectively. {We say that the operation is on the tuple $\x$.}
The operation $\UP[]{\x}$ is an atomic update and should be viewed as an atomic sequence of a read of $\x$ followed by a write to $\x$.
We will use the following terminology: a \emph{read operation} is an $\R[]{\x}$ or a $\UP[]{\x}$, and a \emph{write operation} is a $\W[]{\x}$ or a $\UP[]{\x}$. Furthermore, an \myR-operation is an $\R[]{\x}$, a \myW-operation is a $\W[]{\x}$, and a \myUP-operation is a $\UP[]{\x}$.
We also assume a special \emph{commit} operation denoted $\CT[]$.
To every operation $o$ on a tuple of type $R$, we associate the set of attributes
$\ReadSet{o}\subseteq\Attr{R}$ and $\WriteSet{o}\subseteq\Attr{R}$ containing, respectively, the set of attributes that $o$ reads from and writes to. When $o$ is a $\myR$-operation then $\WriteSet{o}=\emptyset$. Similarly, when $o$ is a $\myW$-operation then $\ReadSet{o}=\emptyset$.

A \emph{transaction} $\trans[]$
is a sequence of read and write operations 
followed by a commit.  
We assume that a transactions starts when its first operation is executed, but no earlier.
Formally, we model a transaction as a linear order $(\trans[],\leq_{\trans[]})$, where $\trans[]$ is the set of (read, write and commit) operations occurring in the transaction and $\leq_{\trans[]}$ encodes the ordering of the operations. As usual, we use $<_{\trans[]}$ to denote the strict ordering.

When considering a set $\transset$ of transactions, we assume that every transaction in the set has a unique id $i$ and write $\trans$ to make this id explicit. Similarly, to distinguish the operations of different transactions, we add this id as a subscript to the operation. {That is, we write $\W{\x}$, $\R{\x}$, and $\UP{\x}$ to denote a $\W[]{\x}$, $\R[]{\x}$, and $\UP[]{\x}$ occurring in transaction $\trans$; similarly $\CT[i]$ denotes the commit operation in transaction $T_i$. }
This convention is consistent with the literature (see, \eg\
\cite{DBLP:conf/sigmod/BerensonBGMOO95,DBLP:conf/pods/Fekete05}). 
To avoid ambiguity of notation, we assume that a transaction performs at most one write, one read, and one update per tuple.
The latter is a common assumption (see, \eg~\cite{DBLP:conf/pods/Fekete05}). All our results carry over to the more general setting in which multiple writes and reads per tuple are allowed.

A \emph{(multiversion) schedule} $\schedule$ over a set $\transset$ of transactions is a tuple $(O_\schedule, \leq_\schedule, {\ll_\schedule,} v_\schedule)$ where $O_\schedule$ is the set 
 containing all operations of transactions in $\transset$ as well as a special operation $\sstart$ conceptually writing the initial versions of all existing tuples, $\leq_\schedule$ encodes the ordering of these operations, {$\ll_\schedule$ is a \emph{version order} providing for each tuple $\x$ a total order over all write operations on $\x$ occurring in $\schedule$,} and $v_\schedule$ is a \emph{version function} mapping each read operation $a$ in $\schedule$ to either $\sstart$ or to a write\footnote{Recall that a write operation is either a $\W[]{x}$ or a $\UP[]{x}$.} operation different from $a$ in $\schedule$. 
We require that $\sstart \leq_\schedule a$ for every operation $a \in {O_\schedule}$, {$\sstart \ll_\schedule a$ for every write operation $a \in {O_\schedule}$}, and that $a <_{\trans[]} b$ implies $a <_\schedule b$ for every $\trans[] \in \transset$ and every $a,b \in \trans[]$.\footnote{{Recall that $<_{\trans[]}$ denotes the order of operations in transaction $\trans[]$.}} 
We furthermore require that for every read operation $a$, $v_\schedule(a) <_\schedule a$ and, if $v_\schedule(a) \neq \sstart$, then the operation $v_\schedule(a)$ is on the same tuple as $a$.
Intuitively, $\sstart$ indicates the start of the schedule, the order of operations in $s$ is consistent with the order of operations in every transaction $\trans[]\in\transset$, and the version function maps each read operation $a$ to the operation that wrote the version observed by $a$.
If $v_\schedule(a)$ is $\sstart$, then $a$ observes the initial version of this tuple.
{The version order $\ll_\schedule$ represents the order in which different versions of a tuple are installed in the database. For a pair of write operations on the same tuple, this version order does not necessarily coincide with $\leq_\schedule$. For example, under \mvrc the version order is based on the commit order instead.}

We say that a schedule $\schedule$ is a \emph{single version schedule} if {$\ll_\schedule$ coincides with $\leq_\schedule$ and} every read operation always reads the last written version of the tuple. Formally, {for each pair of write operations $a$ and $b$ on the same tuple, $a \ll_\schedule b$ iff $a <_\schedule b$, and} for every read operation $a$ there is no write operation $c$ on 
the same tuple as $a$
with $v_\schedule(a) <_\schedule c  <_\schedule a$. 
A single version schedule over a set of transactions $\transset$ is \emph{single version 
serial} if its transactions are not interleaved with operations from other transactions. That is, for every $a,b,c \in {O_\schedule}$ with $a <_{\schedule}
b<_{\schedule} c$ and $a,c \in \trans[]$ implies $b \in \trans[]$ for every
$\trans[] \in \transset$.

{The absence of aborts in our definition of schedule is consistent with the common assumption~\cite{DBLP:conf/pods/Fekete05,DBLP:conf/concur/0002G16} that an underlying recovery mechanism will rollback aborted transactions. We only consider isolation levels that only read committed versions. Therefore there will never be cascading aborts.}

\subsection{Conflict Serializability}
\label{sec:ser}

{Let $a_j$ and $b_i$ be two operations on the same tuple from different transactions $\trans[j]$ and $\trans[i]$ in a {set of transactions $\transset$}. We then say that {$a_j$ is \emph{conflicting} with $b_i$} if:
\begin{itemize}
	\item \emph{(ww-conflict)} $\WriteSet{a_j} \cap \WriteSet{b_i} \neq \emptyset$; or,
	\item \emph{(wr-conflict)} $\WriteSet{a_j} \cap \ReadSet{b_i} \neq \emptyset$; or, 
    \item \emph{(rw-conflict)} $\ReadSet{a_j} \cap \WriteSet{b_i} \neq \emptyset$.
\end{itemize}}
\noindent In this case, we also say that $a_j$ and $b_i$ are conflicting operations.
Furthermore, commit operations and the special operation $\sstart$ never conflict with any other operation.
When $a_j$ and $b_i$ are conflicting operations in $\transset$, we say that $a_j$ \emph{depends on} $b_i$ in a schedule $\schedule$ over $\transset$, denoted $b_i \rightarrow_\schedule a_j$ if:\footnote{Throughout the paper, we adopt the following convention:  a $b$ operation can be understood as a `before' while an $a$ can be interpreted as an `after'.}
\begin{itemize}
    \item \emph{(ww-dependency)} {{$b_i$ is ww-conflicting with $a_j$} and $b_i \ll_{\schedule} a_j$}; or,
    \item \emph{(wr-dependency)} {{$b_i$ is wr-conflicting with $a_j$} and $b_i = v_\schedule(a_j)$ or $b_i \ll_{\schedule} v_\schedule(a_j)$}; or, 
    \item \emph{(rw-antidependency)} {{$b_i$ is rw-conflicting with $a_j$} and $v_\schedule(b_i) \ll_{\schedule} a_j$.}
\end{itemize}

\noindent Intuitively, a ww-dependency from $b_i$ to $a_j$ implies that $a_j$ writes a version of a tuple {that is installed} after the version written by $b_i$.
A wr-dependency from $b_i$ to $a_j$ implies that $b_i$ either writes the version observed by $a_j$, or it writes a version that is {installed} before the version observed by $a_j$.
A rw-antidependency from $b_i$ to $a_j$ implies that $b_i$ observes a version {installed} before the version written by $a_j$.

Two schedules $\schedule$ and $\schedule'$ are \emph{conflict equivalent} if they are over the same set $\transset$ of transactions and for every pair of conflicting operations $a_j$ and $b_i$, $b_i \rightarrow_\schedule a_j$ iff $b_i \rightarrow_{\schedule'} a_j$.

\begin{defi}
    A schedule $\schedule$ is \emph{conflict serializable} if it is conflict equivalent to a single version serial schedule.
\end{defi}

A \emph{conflict graph} $\cg{\schedule}$ for schedule $\schedule$ over a set of transactions $\transset$ is the graph whose nodes are the transactions in $\transset$ and where there is an edge from $T_i$ to $T_j$ if $T_i$ has an operation $b_i$ that conflicts with an operation $a_j$ in $T_j$ and $b_i \rightarrow_\schedule a_j$.
\begin{thmC}[\cite{DBLP:books/cs/Papadimitriou86}]\label{theo:not-conflict-serializable}
    A schedule $\schedule$ is conflict serializable iff the conflict graph for
    $\schedule$ is acyclic.
\end{thmC}

\subsection{Multiversion Read Committed}

Let $\schedule$ be a schedule for a set $\transset$ of transactions.
Then, $\schedule$ \emph{exhibits a dirty write}
iff there are two {ww-conflicting} operations $a_j$ and $b_i$ in $\schedule$ on the same tuple $\x$
with $a_j \in \trans[j]$, $b_i \in \trans[i]$ and $\trans[j] \neq \trans[i]$
such that $b_i <_\schedule a_j <_\schedule \CT[i].$
That is, transaction $T_j$ writes to {an attribute of} a tuple that has
been modified earlier by $T_i$, but $T_i$ has not yet issued a commit.

{For a schedule $\schedule$, the version order $\ll_\schedule$ corresponds to the commit order in $\schedule$ if for every pair of write operations $a_j \in \trans[j]$ and $b_i \in \trans[i]$, $b_i \ll_\schedule a_j$ iff $\CT[i] <_\schedule a_j$.}
We say that a schedule $\schedule$ is \emph{read-last-committed (RLC)} if {$\ll_\schedule$ corresponds to the commit order and} for every read operation $a_j$ in $\schedule$ on some tuple $\x$ the following holds:
\begin{itemize}
    \item $v_\schedule(a_j) = \sstart$ or $\CT[i] <_\schedule a_j$ with $v_\schedule(a_j) \in \trans[i]$; and
    \item there is no write\footnote{Recall that a write operation is either a $\myW$ or a $\myUP$-operation.} operation $c_k \in \trans[k]$ on $\x$ with $\CT[k] <_\schedule a_j$ and $v_\schedule(a_j) {\ll_\schedule} c_k$.
\end{itemize}
So, $a_j$ observes the most recent version of $\x$ {(according to the order of commits)} that is committed before $a_j$. Note in particular that a schedule cannot exhibit dirty reads, defined in the traditional way~\cite{DBLP:conf/sigmod/BerensonBGMOO95}, if it is read-last-committed.

\begin{defi} \label{def:isolationlevels}
A schedule is \emph{allowed under isolation level} read committed (RC) if
it is read-last-committed and does not exhibit dirty writes. 
\end{defi}

\noindent Since a read operation in a schedule allowed under \rc can access the most recently committed version immediately instead of waiting for an uncommitted version to be committed, our definition of read committed allows more schedules than the more restrictive lock-based implementation of read committed~\cite{DBLP:conf/sigmod/BerensonBGMOO95}. Furthermore, our definition of \rc should be contrasted with more abstract specifications of Read Committed~\cite{DBLP:conf/icde/AdyaLO00} where read operations are only required to read a committed version, rather than the most recent one. We emphasize that our definition of \rc is in line with practical implementations of read committed found in e.g.\ PostgreSQL.\footnote{\url{https://www.postgresql.org/docs/15/transaction-iso.html}}

\subsection{\PTranss{}}
\label{sec:templates}

\Ptranss{} are transactions where operations are defined over typed variables together with functional constraints on these variables. Types of variables are relation names in $\schRel$ and indicate that variables can only be instantiated by tuples from the respective type.
We fix an infinite set of variables $\variables$ that is disjoint from $\objects$. Every variable $\vx\in\variables$ has an associated relation name in $\schRel$ as type that we denote by $\type{\vx}$.
{For an operation $o_i$ in a \shortptrans{}, $\myvar{o_i}$ denotes the variable in $o_i$.}
An \emph{equality constraint} is an expression of the form $\vx = f(\vy)$ where $\vx,\vy\in\variables$, $\dom{f}=\type{\vy}$ and $\range{f}=\type{\vx}$. A \emph{disequality constraint} is an expression of the form $\vx\ne\vy$ where $\type{\vx}=\type{\vy}$. 

\begin{defi}\label{def:template}
A \emph{\ptrans{}} is a transaction $\tau$
over $\variables$ together with a set 
$\Cset(\tau)$ of equality and disequality constraints.
{In addition, for every operation $o$ in $\tau$ over a variable $\vx$,
$\ReadSet{o}\subseteq \Attr{\type{\vx}}$ and $\WriteSet{o}\subseteq \Attr{\type{\vx}}$}.
\end{defi}

\noindent Recall that we denote variables by capital letters $\vx,\vy,\vz$ and tuples by small letters $\x,\y$. 
\new{The transaction templates derived from the SmallBank and TPC-C benchmarks are shown in Figure~\ref{fig:smallbank-abstract-syntax} and Figure~\ref{fig:tpc-c}, respectively.}
A variable assignment $\mu$ is a mapping from $\variables$ to $\objects$
such that $\mu(\vx)\in \objects_{\type{\vx}}$.
Furthermore, $\mu$ \emph{satisfies} a constraint $\vx = f(\vy)$ (resp., $\vx\ne\vy$) over a database $\db$ when $\mu(\vx) = f^{\db}(\mu(\vy))$ (resp., $\mu(\vx)\ne \mu(\vy)$). 
A variable assignment $\tmap$ for a \ptrans{}
$\tau$ is \emph{admissible} for 
$\db$ if 
it satisfies all constraints in $\Cset(\tau)$ over $\db$. 
By $\tmap (\templ[])$, we denote the transaction obtained by replacing each variable $\vx$ in $\templ[]$ with $\tmap (\vx)$. 

A set of transactions $\transset$ is \emph{consistent} with a set of \ptranss{} $\workload$ {and 
database $\db$}, if for every transaction $\trans[]$ in $\transset$ there is a \ptrans{}
$\tau\in \workload$ and a variable mapping $\mu_{\trans[]}$ that is admissible for $\db$ such that $\mu_{\trans[]}(\tau) = \trans[]$.
\new{We refer to Section~\ref{sec:example} for concrete examples based on transaction templates derived from the SmallBank and TPC-C benchmarks.}

\subsection{Robustness}

We define the robustness property~\cite{DBLP:conf/concur/0002G16} (also called \emph{acceptability} in~\cite{DBLP:conf/pods/Fekete05,DBLP:journals/tods/FeketeLOOS05}), which guarantees serializability for all schedules of a given set of transactions for a given isolation level.

\begin{defi}[Transaction Robustness]
\label{def:robustness}
    A set\/ $\transset$ of transactions is \emph{robust} against RC
    if every schedule for\/ $\transset$ that is allowed under RC is
    conflict serializable.
\end{defi}

In the next definition, we represent conflicting operations from transactions in a set  $\transset$ as quadruples $(T_i, b_i, a_j, T_j)$ with $b_i$ and $a_j$ conflicting operations, and $T_i$ and $T_j$ their respective transactions in $\transset$. We call these quadruples \emph{\conflictquadruples{}} for $\transset$. 
Further, for an operation $b\in \trans[]$, we denote by $\prefix{\trans[]}{b}$ the restriction of $\trans[]$ to all operations that are before or equal to $b$ according to $\leq_{\trans[]}$. Similarly, we denote by $\postfix{\trans[]}{b}$ the restriction of $\trans[]$ to all operations that are strictly after $b$ according to $\leq_{\trans[]}$. Throughout the paper, we interchangeably consider transactions both as linear orders as well as sequences.
Therefore, $\trans[]$ is then equal to the sequence $\prefix{\trans[]}{b}$ followed by $\postfix{\trans[]}{b}$ which we denote by $\prefix{\trans[]}{b}\cdot \postfix{\trans[]}{b}$ for every $b\in T$.

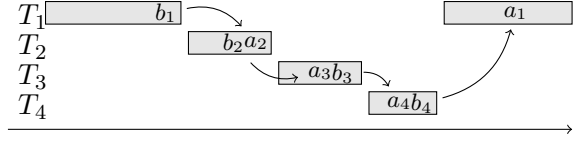
\begin{figure*}[t]
    \centering
    \begin{tikzpicture}
        [
            t0/.style={color=black,fill=black!10},
            t1/.style={color=red,fill=red!30},
            t2/.style={color=blue,fill=blue!30},
            t3/.style={color=orange,fill=orange!30},
            t4/.style={color=pink,fill=pink!30}
        ]

        \draw[t0](.5,1) rectangle (2.3,.7);         \draw[t0](5.8,1) rectangle (7.5,.7);
        \draw[t0](2.4,.6) rectangle (3.5,.3);
        \draw[t0](3.6,.2) rectangle (4.7,-.1);
        \draw[t0](4.8,-.2) rectangle (5.7,-.5);
        \node[anchor=west] at (0,.8) {$T_1$};
        \node[anchor=west] at (0,.4) {$T_2$};
        \node[anchor=west] at (0,.0) {$T_3$};
        \node[anchor=west] at (0,-.4) {$T_4$};
    
        \node(Aa) at (2.1,.85) {\footnotesize{$b_1$}};          \node(Ab) at (6.75,.85) {\footnotesize{$a_1$}};
        \node(Ba) at (3,.45) {\footnotesize{$b_2$}};
        \node(Bb) at (3.3,.45) {\footnotesize{$a_2$}};
        \node(Ca) at (4.15,.05) {\footnotesize{$a_3$}};
        \node(Cb) at (4.45,.05) {\footnotesize{$b_3$}};
        \node(Da) at (5.2,-.35) {\footnotesize{$a_4$}};
        \node(Db) at (5.5,-.35) {\footnotesize{$b_4$}};
       
       \path[->](Aa) edge[bend left] node{} (Bb);
       \path[->](Ba) edge[bend right] node{} (Ca);
       \path[->](Cb) edge[bend left] node{} (Da);
       \path[->](Db) edge[bend right] node{} (Ab);
       
        \draw[->](0,-.7) -- (7.5,-.7);
    \end{tikzpicture}
    \caption{\label{fig:mvschedule} Multiversion split schedule.}
\end{figure*}

\begin{defi}[Multiversion split schedule]\label{def:mvsplitschedule}
Let $\transset$ be a set of transactions and $C = (T_1, b_1, a_2, T_2), (T_2, b_2, a_3, T_3), \ldots, (T_m,\allowbreak b_m,\allowbreak a_1,\allowbreak T_1)$ a \cqsequence{} for $\transset$ such that each transaction in $\transset$ occurs in at most two different quadruples. A \emph{multiversion split schedule} for $\transset$ based on $C$ is a multiversion
schedule that has the following form:
$$ \prefix{\trans[1]}{b_1}\cdot \trans[2]\cdot \ldots \cdot \trans[m] \cdot \postfix{\trans[1]}{b_1}\cdot \trans[m+1] \cdot \ldots \cdot \trans[n],$$
where
\begin{enumerate}
    \item \label{c:1} there is no write operation in $\prefix{\trans[1]}{b_1}$ ww-conflicting with a write operation in any of the transactions $T_2, \ldots, T_m$; 
    \item \label{c:2} $b_1 <_{\trans[1]} a_1$ or $b_{m}$ is rw-conflicting with $a_{1}$; and,
    \item \label{c:3} $b_{1}$ is rw-conflicting with $a_{2}$.
\end{enumerate}
Furthermore, $\trans[m+1],\ldots,\trans[n]$ are the remaining transactions in $\transset$ (those not mentioned in $C$) in an arbitrary order.
\end{defi}

Figure~\ref{fig:mvschedule} depicts a schematic multiversion split schedule. The name stems from the fact that the schedule is obtained by splitting one transaction in two ($T_1$ at operation $b_1$ in Figure~\ref{fig:mvschedule}) and placing all other transactions in $C$ in between. The figure does not display the trailing transactions $\trans[m+1], \trans[m+2], \ldots$ and assumes $b_1<_{T_1} a_1$.

The following theorem characterizes non-robustness in terms of the existence of a multiversion split schedule.
\begin{thmC}[\cite{fullversion}] \label{theo:characterization:split-shedules}
    For a set of transactions $\transset$, the following are equivalent:
    \begin{enumerate}
        \item \label{theo:char:split1} $\transset$ is not robust against \MVRC;
        \item \label{theo:char:split3} there is a multiversion split schedule $\schedule$ for $\transset$ {based on some $C$}.
    \end{enumerate}
\end{thmC}

Let $\workload$ be a set of \ptranss{} and $\db$ be a database.
Then, $\workload$ is \emph{robust against RC over $\db$} if for every set of transactions $\transset$ that is consistent with $\workload$ and $\db$, it holds that $\transset$ is robust against RC.

\begin{defi}[\shortPtrans{} Robustness]\label{def:template_robustness}
 A set of \ptranss{} 
$\workload$ is \emph{robust} against RC if\/ $\workload$ is robust against RC for every database $\db$.
\end{defi}

We say that a \ptrans{} $(\tau,\Cset)$ is a \emph{variable \ptrans{}} when $\Cset=\emptyset$ and
an \emph{equality \ptrans{}} when all constraints in $\Cset$ are equalities.
We denote these sets by \basictemplates and \equaltemplates, respectively.
For an isolation level $\isolationlevel$ and a class of \ptranss{} $\classC$, \RobustnessTempTwo{\classC}{\isolationlevel} is the problem to decide if a given set of \ptranss{} $\workload\in \classC$ is robust against \isolationlevel. When $\classC$ is the class of all \ptranss{}, we simply write \RobustnessTemp{\isolationlevel}.
\begin{thmC}[\cite{fullversion}]
    \RobustnessTempTwo{\basictemplates}{\mvrc} is decidable in \PTIME.
\end{thmC}

In Section~\ref{sec:undec} we start out with a negative result and argue that the addition of functional constraints in its most general form is undecidable by proving undecidability for \RobustnessTempTwo{\equaltemplates}{\mvrc}. Notice in particular that the undecidability result does not even require disequalities.
To obtain decidable fragments, we introduce restrictions on the structure of functional constraints. The \emph{schema graph} $\sg{\schRel, \schFunc}$ of a schema $(\schRel,\schFunc)$ is a directed multigraph having the relations in $\schRel$ as nodes, and in which there are as many edges from a node $R\in \schRel$ to node $S \in \schRel$ as there are functions $f \in \schFunc$ with $\dom{f} = R$ and $\range{f} = S$.
We say that a schema $(\schRel, \schFunc)$ is \emph{acyclic} if the multigraph $\sg{\schRel, \schFunc}$ is acyclic and that it is a \emph{\multitree} if there is at most one directed path between any two nodes in $\sg{\schRel, \schFunc}$.

\begin{exa}
    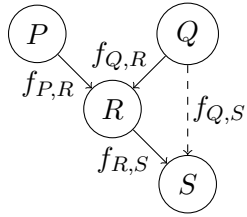
\begin{figure}[t]
        \centering
        \begin{tikzpicture}[scale=1]
        \node[draw, circle,text width=.8em,align=center] (P) at (0,0) {$P$};
        \node[draw, circle,text width=0.8em,align=center] (Q) at (2,0) {$Q$};
        \node[draw, circle,text width=0.8em,align=center] (R) at (1,-1) {$R$};
        \node[draw, circle,text width=0.8em,align=center] (S) at (2,-2) {$S$};
        \draw[->] (P) -- (R) node[midway,below left=-.5em] {$f_{P,R}$};
        \draw[->] (Q) -- (R) node[midway,above left=-.4em] {$f_{Q,R}$};
        \draw[->] (R) -- (S) node[midway,below left=-.5em] {$f_{R,S}$};
        \draw[->, dashed] (Q) -- (S) node[midway,right=-.2em] {$f_{Q,S}$};
        \end{tikzpicture}
        \caption{Acyclic schema graph over the schema $(\{P,Q,R,S\},\allowbreak \{f_{P,R},f_{Q,R},\allowbreak f_{R,S},\allowbreak f_{Q,S}\})$. If we remove function name $f_{Q,S}$ (dashed edge), the resulting schema graph is a multi-tree.}
        \label{fig:ex:sg2}
    \end{figure}
    Consider the schema $(\{P,Q,R,S\}, \{f_{P,R},f_{Q,R},f_{R,S}\})$ with $\dom{f_{i,j}} = i$ and $\range{f_{i,j}} = j$ for each function $f_{i,j}$. The corresponding schema graph with solid lines is given in Figure~\ref{fig:ex:sg2}. This schema is a multi-tree, as there is at most one path between any pair of nodes. Notice that the definition of a multi-tree is more general than a forest, as a node can still have multiple parents (e.g., node $R$ in our example).  Adding the function name $f_{Q,S}$ with $\dom{f_{Q,S}} = Q$ and $\range{f_{Q,S}} = S$ results in the schema graph given in Figure~\ref{fig:ex:sg2}
    that is still acyclic, but no longer a multi-tree as there are now two paths from $Q$ to $S$. \hfill$\Box$
\end{exa}

The schema graph constructed in the proof of Theorem~\ref{theo:equaltemplates} contains several cycles (cf., Figure~\ref{fig:beautifulfigure}).
We consider in Section~\ref{sec:MTBij} robustness for a fragment where a restricted form of cycles in the schema graph is allowed but where
 additional constraints on the \shortptranss{} are assumed.
We consider robustness for acyclic schema graphs in Section~\ref{sec:acyclic}.

\section{Robustness for \shortPTranss{}}
\label{sec:robustness-templates}
\label{sec:undec}

We start out with a negative result and show that the robustness problem in its most general form is undecidable (even when disequalities are not allowed).
The proof is a reduction from \emph{Post's Correspondence Problem (PCP)}~\cite{PCP} and relies on cyclic dependencies between functional constraints. 
{The proof can be found in
the remainder of this section
and is quite elaborate but the basic intuition is simple:
the counterexample split schedule will build up the two strings that need to be generated by the PCP instance by repeated application of functional constraints.

\begin{restatable}{thm}{theoundec}\label{theo:equaltemplates}
     \RobustnessTempTwo{\equaltemplates}{\mvrc} is undecidable.
\end{restatable}

\noindent It might be tempting to relate the above result to the undecidability of the implication problem for functional and inclusion dependencies~\cite{DBLP:journals/siamcomp/ChandraV85}. Functional constraints indeed allow to define inclusion dependencies (as in the SmallBank example) but they always relate complete tuples and are not suited to define functional dependencies. Furthermore, the proof of Theorem~\ref{theo:equaltemplates} makes use of only unary relations, for which the implication problem for functional dependencies and inclusion dependencies is known to be decidable.

The remainder of this section is devoted to proving the correctness of Theorem~\ref{theo:equaltemplates}. We first present the reduction from the PCP problem in Section~\ref{sec:undec:reduction}. Afterwards, we show that this reduction is indeed correct by proving both directions in respectively Section~\ref{prop:pcp-schedule-creation} and Section~\ref{sec:undec:ifdirection}.

\subsection{Reduction}
\label{sec:undec:reduction}

\begin{figure}[tp]
    \centering
    \begin{minipage}[t]{\textwidth/3-1ex}    
    \ttinit($\vi$):
    \[
    \begin{array}{l}
    \W[]{\vx_1: Boolean} \\
    \R[]{\vi : InitialConflict} \\
    \R[]{\vs_1 : String} \\
    \R[]{\vs_e : String} \\
    \W[]{\vc:PCPSolution}\\
    \vx_1 = \ufunc{is-non-empty}(\vs_1)\\
    \vx_1 = \ufunc{is-error}(\vs_e)\\
    \vc = \ufunc{final-domino-sequence}(\vi) \\
    \vs_1 = \ufunc{final-dominoes-string}(\vi)\\
    \vs_e = \ufunc{error-string}(\vi)\\
    \vs_1 = \ufunc{solution-string}(\vc)  \\
    \vi = \ufunc{defines}(\vx_1)
    \end{array}
    \]
    \end{minipage}
    \begin{minipage}[t]{\textwidth/3}    
    \ttinitb($\vs_1$):
    \[
    \begin{array}{l}
    \W[]{\vx_2: Boolean} \\
     \W[]{\vi: InitialConflict}\\
     \R[]{\vs_0: String}\\
     \R[]{\vs_e: String}\\
      \R[]{\vs_1: String}\\
     \W[]{\vb:DominoSequence}\\
    \vx_2  = \ufunc{is-non-empty}(\vs_0)\\
    \vx_2  = \ufunc{is-error}(\vs_0)\\
    \vs_1 = \ufunc{final-dominoes-string}(\vi)\\
    \vs_0  = \ufunc{top-string}(\vb)\\
    \vs_0  = \ufunc{bottom-string}(\vb)\\
    \vs_0 = \ufunc{empty-string}(\vb)\\
    \vs_1 = \ufunc{future-solution-string}(\vb)\\
    \vs_e = \ufunc{detach}(\vs_0) \\
    \vs_e = \ufunc{detach}(\vs_e) \\
    \vi = \ufunc{defines}(\vx_2) \\
    \vb = \ufunc{empty-domino-sequence}(\vs_1)
    \end{array}
    \]
    \end{minipage}
    \begin{minipage}[t]{\textwidth/3-1ex}    
    \ttclose($\vb$):
    \[
    \begin{array}{l}
    \W[]{\vb:DominoSequence}\\
    \R[]{\vs_t : String}\\
    \R[]{\vs_b : String}\\
    \R[]{\vs_1 : String}\\
    \W[]{\vc:PCPSolution}\\
    \vs_t  = \ufunc{top-string}(\vb)\\
    \vs_b  = \ufunc{bottom-string}(\vb)\\
    \vs_1 = \ufunc{future-solution-string}(\vb) \\
    \vc = \ufunc{DS$\rightarrow$PCP}(\vb) \\
    \vs_t = \ufunc{solution-string}(\vc) \\
    \vs_b = \ufunc{solution-string}(\vc) \\
    \vs_1 = \ufunc{solution-string}(\vc) \\
    \vb = \ufunc{PCP$\rightarrow$DS}(\vc) \\
    \end{array}
    \]
    \end{minipage}
    
    \vspace{2em}
    \begin{minipage}[t]{\textwidth-2ex}
    For every domino $d_i = (a_1a_2\ldots a_h, b_1b_2\ldots b_k) \in \dominoset$ a \ptrans{} \ttdomino{i}($\vb$):
    \end{minipage}
    \vspace{-.5em}
    
    \begin{minipage}[t]{\textwidth/3-2ex}    
    \[
    \begin{array}{l}
        \W[]{\vb:DominoSequence}\\
    \R[]{\vs_0:String}\\
    \R[]{\vs_1:String}\\
    \R[]{\vs_t:String}\\
    \R[]{\vs_{t\mathit{a_1}}:String}\\
    \R[]{\vs_{t\mathit{a_1a_2}}:String}\\
    \ldots\\
    \R[]{\vs_{t\mathit{a_1a_2\ldots a_h}}:String}\\
    \R[]{\vs_b:String}\\
    \R[]{\vs_{b\mathit{b_1}}:String}\\
    \R[]{\vs_{b\mathit{b_1b_2}}:String}\\
    \ldots\\
    \R[]{\vs_{b\mathit{b_1b_2\ldots b_k}}:String}\\
    \W[]{\vb_{next}:DominoSequence}\\
    \end{array}
    \]
    \end{minipage}
    \begin{minipage}[t]{\textwidth/3-.7ex}    
    \[
    \begin{array}{l}
    \vs[t]  = \ufunc{top-string}(\vb)\\
    \vs[t\mathit{a_1}] = \ufunc{append-$a_1$}(\vs[t])\\
    \vs[t\mathit{a_1a_2}] = \ufunc{append-$a_2$}(\vs[t\mathit{a_1}])\\
    \ldots\\
    \vs[t\mathit{a_1a_2\ldots a_h}] = \\
    \qquad\quad \ufunc{append-$a_h$}(\vs[t\mathit{a_1\ldots a_{h-1}}])\\
    \vs[t]  = \ufunc{detach}(\vs[t\mathit{a_1}])\\
    \vs[t\mathit{a_1}]  = \ufunc{detach}(\vs[t\mathit{a_1a_2}])\\
    \ldots\\
    \vs[t\mathit{a_1a_2\ldots a_{h-1}}]  =\\
    \qquad\quad \ufunc{detach}(\vs[t\mathit{a_1a_2\ldots a_h}])\\
    \vs[t\mathit{a_1a_2\ldots a_h}]  = \ufunc{top-string}(\vb_\text{next}) \\
    \vs[\mathit{a_1}] = \ufunc{top}(\vs[t\mathit{a_1}]) \\
    \vs[\mathit{a_2}] = \ufunc{top}(\vs[t\mathit{a_1a_2}]) \\
    \ldots \\
    \vs[\mathit{a_h}] = \ufunc{top}(\vs[t\mathit{a_1a_2\ldots a_h}]) \\
    \vs_1 = \ufunc{future-solution-string}(\vb) \\
    \vs_0 = \ufunc{empty-string}(\vb)
    \end{array}
    \]
    \end{minipage}
    \begin{minipage}[t]{\textwidth/3+1.3ex}    
    \[
    \begin{array}{l}
    \vs[b]  = \ufunc{bottom-string}(\vb)\\
    \vs[b\mathit{b_1}] = \ufunc{append-$b_1$}(\vs[b])\\
    \vs[b\mathit{b_1b_2}] = \ufunc{append-$b_2$}(\vs[b\mathit{b_1}])\\
    \ldots\\
    \vs[b\mathit{b_1b_2\ldots b_k}] =\\
    \qquad\quad \ufunc{append-$b_k$}(\vs[b\mathit{b_1\ldots b_{k-1}}])\\
    \vs[b]  = \ufunc{detach}(\vs[t\mathit{b_1}])\\
    \vs[b\mathit{b_1}]  = \ufunc{detach}(\vs[b\mathit{b_1b_2}])\\
    \ldots\\
    \vs[b\mathit{b_1b_2\ldots b_{k-1}}]  =\\
    \qquad\quad \ufunc{detach}(\vs[b\mathit{b_1b_2\ldots b_k}])\\
    \vs[b\mathit{b_1b_2\ldots b_k}]  =\\
    \qquad\quad \ufunc{bottom-string}(\vb_\text{next}) \\
    \vs[\mathit{b_1}] = \ufunc{top}(\vs[b\mathit{b_1}]) \\
    \vs[\mathit{b_2}] = \ufunc{top}(\vs[b\mathit{b_1b_2}]) \\
    \ldots\\
    \vs[\mathit{b_k}] = \ufunc{top}(\vs[b\mathit{b_1b_2\ldots b_k}]) \\
    \vs_1 = \ufunc{future-solution-string}(\vb_\text{next}) \\
    \vs_0 = \ufunc{empty-string}(\vb_\text{next})\\
    \vb_\text{next} = \ufunc{next-sequence}(\vb)\\
    \vb = \ufunc{previous-sequence}(\vb_\text{next})
    \end{array}
    \]
    \end{minipage}
    \caption{\label{fig:pcptemplates}\Ptranss{} for the proof of Theorem~\ref{theo:equaltemplates}.}
\end{figure}

{The proof is based on a reduction from the \emph{Post's Correspondence Problem (PCP)}, which is known to be undecidable~\cite{PCP}.}
A domino is a pair $(\vec{a}, \vec{b})$ of two non-empty strings over $\Sigma$.
Henceforth we call $\vec{a}$ its \emph{top value} and $\vec{b}$ its \emph{bottom value}. 
Given a set of dominoes $\mathcal{D}$, the PCP asks if a non-empty sequence $d_1, d_2, \ldots, d_r$ of dominoes in $\mathcal{D}$ exists such that, with $d_i = (\vec{a_i}, \vec{b_i})$, the strings $\vec{a_1}\vec{a_2}\ldots\vec{a_r}$ and $\vec{b_1}\vec{b_2}\ldots\vec{b_r}$ are identical. 

For the reduction to non-robustness against \mvrc, we construct a set $\workload$ of \ptranss{} consisting of the \ptranss{} in Figure~\ref{fig:pcptemplates} for $\dominoset$.
There are the transactions $\ttinit$, $\ttinitb$ and $\ttclose$ (whose meaning will be explained next) and for every domino in $\mathcal{D}$ there is a template in Figure~\ref{fig:pcptemplates} representing that domino and the action of appending that domino to a sequence of dominoes.  The schema consists of the relations $\{ 
\text{\tt Boolean}, \text{\tt InitialConflict}, \text{\tt String},
\text{\tt PCPSolution}, \text{\tt DominoSequence}\}$ whose meaning will be explained below together with a discussion of all the functions.
The schema graph is presented in Figure~\ref{fig:beautifulfigure} and contains various cycles.

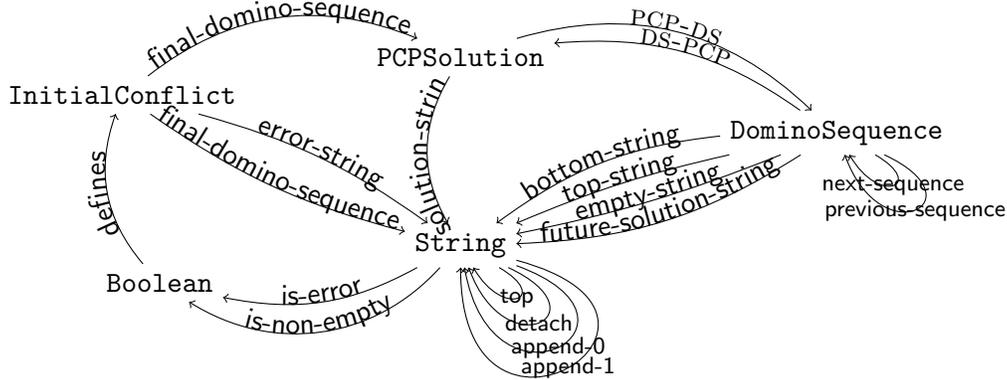
\begin{figure}[ht]
\begin{tikzpicture}

    \node(IC)   at (.5,-.5)  {\tt InitialConflict};
    \node(B)    at (1,-3) {\tt Boolean};
    \node(PCP)  at (5,0) {\tt PCPSolution};
    \node(S)    at (5,-2.5) {\tt String};
    \node(D)    at (10,-1) {\tt DominoSequence};

    \draw[->,postaction={decorate,decoration={raise=0ex,text along path,text align=center,text={|\sffamily|defines}}}] (B) to[bend left] (IC);
    \draw[<-,postaction={decorate,decoration={raise=0ex,text along path,text align=center,text={|\sffamily|is-non-empty}}}] (B) to[bend right=40] (S);
    \draw[<-,postaction={decorate,decoration={raise=0ex,text along path,text align=center,text={|\sffamily|is-error}}}] (B) to[bend right=20] (S);
    \draw[->,postaction={decorate,decoration={raise=0ex,text along path,text align=center,text={|\sffamily|final-domino-sequence}}}] (IC) to[bend right=10] (S);
    \draw[->,postaction={decorate,decoration={raise=0ex,text along path,text align=center,text={|\sffamily|error-string}}}] (IC) to[bend left=10] (S);
    \draw[->] (S) edge[out=330,in=300,looseness=6] node{\sffamily\footnotesize top} (S);
    \draw[->] (S) edge[out=335,in=290,looseness=7] node{\sffamily\footnotesize detach} (S);
    \draw[->] (S) edge[out=340,in=280,looseness=8] node{\sffamily\footnotesize append-0} (S);
    \draw[->] (S) edge[out=345,in=270,looseness=10] node{\sffamily\footnotesize append-1} (S);
    \draw[->] (D) edge[out=330,in=300,looseness=6] node{\sffamily\footnotesize next-sequence} (D);
    \draw[->] (D) edge[out=335,in=290,looseness=7] node{\sffamily\footnotesize previous-sequence} (D);
    \draw[<-,postaction={decorate,decoration={raise=0ex,text along path,text align=center,text={|\sffamily|solution-string}}}] (S) to[bend left] (PCP);
    \draw[->,postaction={decorate,decoration={raise=0ex,text along path,text align=center,text={|\sffamily|final-domino-sequence}}}] (IC) to[bend left] (PCP);
    \draw[<-,postaction={decorate,decoration={raise=0ex,text along path,text align=center,text={|\sffamily|top-string}}}] (S) to[bend left=5] (D);
    \draw[<-,postaction={decorate,decoration={raise=0ex,text along path,text align=center,text={|\sffamily|bottom-string}}}] (S) to[bend left=15] (D);
    \draw[<-,postaction={decorate,decoration={raise=0ex,text along path,text align=center,text={|\sffamily|empty-string}}}] (S) to[bend right=5] (D);
    \draw[<-,postaction={decorate,decoration={raise=0ex,text along path,text align=center,text={|\sffamily|future-solution-string}}}]  (S) to[bend right=15] (D);
    \draw[->,postaction={decorate,decoration={raise=0ex,text along path,text align=center,text={|\sc|pcp-ds}}}] (PCP) to[bend left] (D);
    \draw[<-,postaction={decorate,decoration={raise=0ex,text along path,text align=center,text={|\sc|ds-pcp}}}] (PCP) to[bend left=20] (D);
\end{tikzpicture}
    \caption{\label{fig:beautifulfigure}{Schema graph for the \ptranss{} in Figure~\ref{fig:pcptemplates} (for any set of dominoes).}}
\end{figure}
To prove Theorem~\ref{theo:equaltemplates}, we will show that there is a solution for PCP if and only if $\workload$ is not robust against \mvrc. For the only-if direction, we show that, if there is a solution $\vec{d} =d_1,d_2,\ldots,d_r$ for the PCP problem over $\mathcal{D}$, then there is a multiversion split schedule that encodes this solution in a particular way: in this schedule the split transaction is an instantiation of \ptrans{} $\ttinit$, the next transaction is an instantiation of $\ttinitb$, then followed by instantiations of \ptranss{} $\ttdomino{d_1}, \ldots, \ttdomino{d_r}$ representing the sequence of dominoes in solution $\vec{d}$, and finally an instantiation of \ptrans{} $\ttclose$. Henceforth, we call a schedule that encodes a sequence of dominoes $\vec{d}$ in this way a \emph{schedule-encoding of $\vec{d}$}. For the if-direction, we first show that every multiversion split schedule consistent with the \ptranss{} in Figure~\ref{fig:pcptemplates} for some set $\dominoset$ of dominoes is a schedule-encoding for some sequence $\vec{d}$ of dominoes from $\dominoset$, and then that for every schedule-encoding of a sequence $\vec{d}$ of dominoes, $\vec{d}$ is always a solution for the PCP problem over a set of dominoes containing those in $\vec{d}$.

\subsection{Only-if direction} We first prove the only-if direction of Theorem~\ref{theo:equaltemplates}.

\newcommand{\stringa}{\textsf{top-string}}
\newcommand{\stringb}{\textsf{bottom-string}}

\begin{prop}[Only-if part of Theorem~\ref{theo:equaltemplates}]\label{prop:pcp-schedule-creation}
    Let $\dominoset$ be a set of dominoes with a solution $\vec{d}$ for the PCP problem for $\dominoset$. Then there exists a schedule-encoding of $\vec{d}$ that is consistent with the \ptranss{} in Figure~\ref{fig:pcptemplates} and some database $\db$. 
\end{prop}
\begin{proof}

\newcommand{\charc}{c}
    Let $\vec{d} = d_1,d_2,\ldots,d_r$ be a solution to the PCP problem for $\dominoset$. Let $\vec{a_1}\vec{a_2}\ldots\vec{a_r}$ be the 
    read of top values and $\vec{b_1}\vec{b_2}\ldots\vec{b_r}$ be the 
    read of bottom values, which thus represent an identical string $\vec{\charc} = \charc_1 \cdots \charc_n$, with $\charc_i \in \Sigma$. 
    We now construct a schedule $\schedule$ and database $\db$ as in Definition~\ref{def:mvsplitschedule} with transactions based on the \ptranss{} $\workload$ in Figure~\ref{fig:pcptemplates}.

    Relation \textsf{PCPSolution} contains a tuple that we interpret as the PCP solution $\vec{d} = d_1,d_2,\ldots,d_r$. Relation \textsf{DominoSequence} contains $r+1$ tuples, one for every prefix of $\vec{d}$, including the empty sequence $()$ and the PCP solution $\vec{d}$ itself. For convenience of notation, we will henceforth often represent tuples by their interpretation, which is justified by the fact that every tuple in a particular relation will have a different interpretation, and the relation itself can always be derived from the context (e.g., the function signature).

    Since the PCP solution has an interpretation in both the relations \textsf{PCPSolution} and \textsf{DominoSequence}, we assume two functions, $\ufunc{PCP$\rightarrow$DS}:\textsf{PCPSolution}\to\textsf{DominoSequence}$ and $\ufunc{DS$\rightarrow$PCP}:\textsf{DominoSequence}\to\textsf{PCPSolution}$ that
    %map these interpretations on one another.
    \new{relate these interpretations to each other.}
    That is,  
     $\ufunc{PCP$\rightarrow$DS}^{\db}(\vec{d}) = \vec{d}$ and $\ufunc{typecase-to-C}^{\db}(\vec{d}) = \vec{d}$.

    Further, we have functions $\ufunc{next-sequence}:\textsf{DominoSequence}\to\textsf{DominoSequence}$ and $\ufunc{previous-sequence}:\textsf{DominoSequence}\to\textsf{DominoSequence}$ with \new{the} following interpretation:
    \begin{align*}
        \ufunc{next-sequence}^{\db}(\vec{d'}) &= \vec{d'}d & \text{with $d'$ a strict prefix of $\vec{d}$ followed by domino $d$ in $\vec{d}$}, \\
        \ufunc{next-sequence}^{\db}(\vec{d}) &= \vec{d}, \\
        \ufunc{previous-sequence}^{\db}(\vec{d'}d) &= \vec{d'} &\text{with $d'$ a strict prefix of $\vec{d}$ followed by domino $d$ in $\vec{d}$}, \\
        \ufunc{previous-sequence}^{\db}(()) &= ().
    \end{align*}
    \new{Intuitively, these functions relate each tuple in $\textsf{DominoSequence}$ representing a prefix of $\vec{d}$ to the prefixes obtained by adding or removing one domino in the sequence. That is, given a tuple in $\textsf{DominoSequence}$ representing a strict prefix $d_1,d_2,\ldots,d_{i-1}, d_i$ of $\vec{d}$, $\ufunc{next-sequence}$ returns the tuple representing $d_1,d_2,\ldots,d_{i-1}, d_i, d_{i+1}$ (i.e., the prefix of $\vec{d}$ obtained by adding one domino), and $\ufunc{previous-sequence}$ returns the tuple representing $d_1,d_2,\ldots,d_{i-1}$ (i.e., the prefix of $\vec{d}$ obtained by removing the last domino). We furthermore distinguish two special cases to guarantee that both functions are defined for all tuples in $\textsf{DominoSequence}$: if the tuple represents the solution $\vec{d}$ itself, then $\ufunc{next-sequence}$ returns $\vec{d}$, and if the tuple represents the empty sequence $()$, then $\ufunc{previous-sequence}$ returns $()$.}

    Relation $\textsf{String}^{\db}$ contains a tuple representing the read $\vec{\charc}$ of PCP-solution sequence $\vec{d}$, a tuple representing an error $\langle\text{error}\rangle$, and a tuple for every substring of $\vec{\charc}$, 
     including the empty string $\langle\rangle$. We assume that all these tuples are different. {We use notation $\langle\rangle$ to denote the empty string to distinguish it from $()$, which denotes the empty sequence of dominoes.}

    \newcommand{\chare}{e}

    Functions $\ufunc{append-0}:\textsf{String}\to\textsf{String}$, $\ufunc{append-1}:\textsf{String}\to\textsf{String}$, $\ufunc{detach}:\textsf{String}\to\textsf{String}$, and $\ufunc{top}:\textsf{String}\to\textsf{String}$ simulate standard string operations for the interpretations of tuples in relation \textsf{String}. Thus, tuples representing a (possibly empty) string $\vec{\chare}$: 
    \begin{align*}
        \ufunc{append-\textit{c}}^{\db}(\langle\vec{\chare}\rangle) &= \left\{
            \begin{array}{ll}
                \langle\vec{\chare}c\rangle &\text{ with $\vec{\chare}$ a (possibly empty) string over $\Sigma$, $c \in \Sigma$, and} \\
                & \text{ $\langle\vec{\chare}c\rangle$ a substring of $\vec{\charc}$,} \\
                \langle\text{error}\rangle & \text{ otherwise},
            \end{array}\right. \\
        \ufunc{detach}^{\db}(\langle\vec{\chare}c\rangle) &= \langle\vec{\chare}\rangle  \text{ with $\vec{\chare}$ a (possibly empty) string over $\Sigma$, and $c \in \Sigma$}, \\
        \ufunc{detach}^{\db}(\langle\rangle) &= \ufunc{detach}(\langle\text{error}\rangle) = \langle\text{error}\rangle, \\
        \ufunc{top}^{\db}(\langle\vec{\chare}c\rangle) &= \langle c\rangle \text{ with $\vec{\chare}$ a (possibly empty) string over $\Sigma$, and $c \in \Sigma$},
        \\
        \ufunc{top}^{\db}(\langle\rangle) &= \ufunc{top}(\langle\text{error}\rangle) = \langle\text{error}\rangle.
    \end{align*}
    Notice that these function interpretations are closed under $\db$, that is, every tuple from relation \textsf{String}$^\db$ maps onto a tuple that is in relation \textsf{String}$^\db$.

    Every tuple in \textsf{DominoSequence} is associated with three tuples in \textsf{String} representing, respectively, the read of top values, the read of bottom values, and the empty string. The association is made via functions $\ufunc{top-string}:\textsf{DominoSequence}\to\textsf{String}$, $\ufunc{bottom-string}:\textsf{DominoSequence}\to\textsf{String}$, and $\ufunc{empty-string}:\textsf{DominoSequence}\to\textsf{String}$ with following interpretations in $\db$:
    \begin{align*}
        \ufunc{top-string}^{\db}(\vec{d'}) &= \vec{\chare} \text{, with $\vec{\chare}$ the read of top values on dominoes in $\vec{d'}$}, \\
        \ufunc{bottom-string}^{\db}(\vec{d'}) &= \vec{\chare} \text{, with $\vec{\chare}$ the read of bottom values on dominoes in $\vec{d'}$, and} \\
        \ufunc{empty-string}^{\db}(\vec{d'}) &= \langle\rangle.
    \end{align*}
    We emphasize that in the expressions above the read $\vec{\chare}$ of top and bottom values on dominoes in $\vec{d'}$ might be empty.

    Finally, for function $\ufunc{future-solution-string}:\textsf{DominoSequence}\to\textsf{String}$ we consider the interpretation that associates every domino sequence $\vec{d'}$ represented by a tuple in relation \textsf{DominoSequence} in $\db$ to the final read $\ufunc{future-solution-string}^{\db}(\vec{d'}) = \vec{\charc}$. Function $\ufunc{solution-string}:\textsf{PCPSolution}\to\textsf{String}$ does the same for the single tuple representing $\vec{d}$ in \textsf{PCPSolution}, thus with $\ufunc{solution-string}^{\db}(\vec{d}) = \vec{\charc}$. 
  Function $\ufunc{empty-domino-sequence}:\textsf{String}\to\textsf{DominoSequence}$ is interpreted to map every tuple in \textsf{String} onto the tuple from \textsf{DominoSequence} representing the empty sequence $()$.

All other relations and functions have as purpose to pass tuples from one transaction to another in a schedule and to enforce that certain tuples do not collide, which is useful for the (if)-part of the proof.

    Relation $\textsf{Boolean}^{\db}$ contains two tuples, which we interpret as Boolean values $0$ and $1$. Function $\ufunc{is-non-empty}:\textsf{String}\to\textsf{Boolean}$ and $\ufunc{is-error}:\textsf{String}\to\textsf{Boolean}$ are interpreted as follows:
    \begin{align*}
        \ufunc{is-non-empty}^{\db}(s) &= \left\{ 
        \begin{array}{ll}
            1 & \text{if }s \ne \langle \rangle, \\
            0 & \text{otherwise},
        \end{array}\right.\text{, and} \\
        \ufunc{is-error}^{\db}(s) &= \left\{ 
        \begin{array}{ll}
            1 & \text{if } s = \langle \text{error}\rangle, \\
            0 & \text{otherwise},
        \end{array}\right.\text{, and}. 
    \end{align*}

Finally, relation $\textsf{InitialConflict}^{\db}$ contains a single tuple, which we refer to by $\langle \text{init} \rangle$. The interpretation of $\ufunc{defines}:\textsf{Boolean}\to\textsf{InitialConflict}$ maps $1$ and $0$ onto 
$\langle \text{init} \rangle$. Function $\ufunc{error-string}:\textsf{InitialConflict}\to\textsf{String}$ maps $\langle \text{init} \rangle$ onto $\langle\text{error}\rangle$. Functions $\ufunc{final-domino-string}:\textsf{InitialConflict}\to\textsf{DominoSequence}$ and $\ufunc{final-domino-sequence}:\textsf{InitialConflict}\to\textsf{PCPSolution}$ map $\langle \text{init} \rangle$ onto the solution domino sequence $\vec{d}$, respectively on the final read $\vec{\charc}$ of $\vec{d}$.

Now the schedule $\prefix{\trans[1]}{b_1}\cdot \trans[2]\cdot \ldots \cdot \trans[m] \cdot \postfix{\trans[1]}{b_1}$, taking $\trans[1] = \ttinit(\langle \text{init} \rangle)$, $\trans[2] = \ttinitb(\langle \text{init} \rangle)$, for $i: 1 \le i \le r$, transaction $\trans[i+2] =  \ttdomino{i}(( d_1, \ldots, d_i))$, $\trans[m] = \ttclose(( d_1, \ldots, d_r))$ and $b_1 = \langle \text{init} \rangle$ has the conditions of Definition~\ref{def:mvsplitschedule}. Indeed, it is based on sequence of conflict quadruples $(T_1, 
\R[1]{\langle\text{init}\rangle}, \W[2]{\langle \text{init}\rangle}, T_2),\allowbreak (T_2, \W[2]{()}, \W[3]{()}, T_3),\allowbreak (T_3, \W[3]{(d_1)},\allowbreak \W[4]{(d_2)}, T_4),\allowbreak \ldots, (T_{r+2}, \W[r+2]{(d_1,\ldots,d_r)},\allowbreak \W[r+3]{(d_1,\ldots,d_r)}, T_{r+3}),\allowbreak (T_{r+3}, \W[r+3]{\vec{d}}, \W[1]{\vec{d}}, T_1)$. 

Condition (\ref{c:1}) is true because there is no ww-conflict between a write operation in $\prefix{\trans[1]}{b_1}$ and a write operation in any of the transactions $T_2, \ldots, T_m$, since the first write operation, respectively second write operation, in $\ttinit(\langle \text{init} \rangle)$ has a type that only occurs before the conflict with $\ttinitb(\langle \text{init} \rangle)$,
    and is the conflict with $\ttclose((d_1, \ldots, d_r))$, respectively. Furthermore (\ref{c:2}) is true because $b_1 <_{\trans[1]} a_1$ and Condition (\ref{c:3}) is true because $b_{1}$ and $a_{2}$ are rw-conflicting.
\end{proof}

\subsection{Helpful lemma} Before proving the opposite direction of Theorem~\ref{theo:equaltemplates}, we first establish the following Lemma.

\begin{lem}\label{lem:robustminimalcondition}
    If a set $\workload$ of \ptranss{} is not robust against \mvrc then there is a \emph{multiversion split schedule} 
$\prefix{\trans[1]}{b_1}\cdot \trans[2]\cdot \ldots \cdot \trans[m] \cdot \postfix{\trans[1]}{b_1}$ for a set $\transset = \{\trans[1], \ldots, \trans[m]\}$ of transactions consistent with $\workload$ in which an operation from a transaction $T_j$ depends on an operation from transaction $T_i$ only if $j=i+1$ or $i=m$ and $j=1$.
\end{lem}
\begin{proof}
    If $\workload$ is not robust against \mvrc, then there is a database $\db$ and a 
    multiversion split schedule $\schedule = \prefix{\trans[1]}{b_1}\cdot \trans[2]\cdot \ldots \cdot \trans[m] \cdot \postfix{\trans[1]}{b_1}\cdot \trans[m+1] \cdot \ldots \cdot \trans[n]$
    based on a sequence of conflict quadruples $C$
    for a set of transactions $\transset$ that is consistent with $\workload$ and $\db$ having the properties of Definition~\ref{def:mvsplitschedule}.

    We can assume that $n=m$. Otherwise removing the transactions $T_{m+1}, \ldots, T_{n}$ from $\transset$, $\schedule$, and $C$. We can also assume that $\schedule$ is read-last-committed. Otherwise, choosing an appropriate version order $\ll_\schedule$ and version function $v_\schedule$.

    Now suppose that there is a transaction $T_j$ with an operation $a'_j$ that depends on an operation $b'_i$ from transaction $T_i$ and with $j \ne i + 1$ or $i=m$ and $j\ne 1$. Clearly, by definition of dependency and the structure of a multiversion split schedule, $i < j$ or $j=1$.

    We proceed the proof by a construction showing that \new{under these assumptions} there is an alternative schedule $\schedule'$ that is also a multiversion split schedule, but for a strict subset of transactions in $\transset$ (thus also still consistent with $\workload$ and $\db$). The result of the lemma then follows from the observation that repeated application of this construction must lead to a schedule with the properties of the lemma, without existence of such a dependency. 
    
    For the construction, we proceed by case distinction.
    
    \smallskip
    \noindent
    \textsc{If $i\ne 1$ and $j\ne 1$}, we construct a schedule $\schedule'$ from $\schedule$ by removing all operations from transactions $\trans[h]$ with $i<h<j$. Notice that we remove at least one transaction, since $i < i+1 < j$. We can derive a sequence of conflict quadruples $C'$ from $C$ by removing all occurrences of these transactions $\trans[h]$ and adding the conflict quadruple $(\trans[i], b'_i, a'_j, \trans[j])$ instead.
    By construction, $\schedule'$ is a multiversion split schedule based on $C'$ over a set of transactions consistent with $\workload$ and $\db$. It remains to show that the newly constructed schedule $\schedule'$ has the properties of Definition~\ref{def:mvsplitschedule}. The latter is straightforward since $C$ and $C'$ agree on their first and last quadruple, due to assumption $i\ne 1$ and $j\ne 1$.

    \smallskip
    \noindent
    \textsc{If $i = 1$}, it follows that $i<j$ and thus $j\ne 1$. Then, we 
    construct a schedule $\schedule'$ from $\schedule$ by removing all operations from transactions $\trans[h]$ with $i<h<j$ and updating the prefix and postfix of $\trans[1]$, now based on $b'_i$.
    Notice that we again remove at least one transaction, since $i < i+1 < j$ and that
    we can derive a sequence of conflict quadruples $C'$ from $C$ in the same way as before, by removing all occurrences of these transactions $\trans[h]$ and adding the conflict quadruple $(\trans[i], b'_i, a'_j, \trans[j])$ instead. By construction, $\schedule'$ is a multiversion split schedule based on $C'$ over a set of transactions consistent with $\workload$ and $\db$. It remains to show that the newly constructed schedule $\schedule'$ has the properties of Definition~\ref{def:mvsplitschedule}.

    First, we observe that $b'_1$ and $a'_j$ are rw-conflicting, which immediately implies that Condition~(3) is true for $\schedule'$. The argument is by exclusion. Indeed, if $b'_1$ and $a'_j$ would be ww-conflicting, then $b'_1 \ll_\schedule a'_j$ implying $b'_1 <_\schedule a_j'$ (due to the assumed read-last committed) and thus $b'_1 \le_\schedule b_1$, which is not allowed by condition~(1) on $\schedule$. It follows from a similar argument that $b'_1$ and $a'_j$ are not wr-conflicting: both $b'_1 = v_\schedule(a'_j)$ and $b'_1 \ll_\schedule v_\schedule(a'_j)$ imply $b'_1 <_\schedule C_1 <_\schedule a'_j$, which contradicts with $C_1$ being the last operation in $\schedule$. 

    Since $b'_1$ is rw-conflicting with $a'_j$, we have $v_\schedule(b'_1) \ll_\schedule a_j'$, implying $b'_1 <_\schedule a_j'$ (due to read-last-committed and the structure of a multiversion split schedule),
    thus $b'_1 \le_\schedule b_1$. Therefore, Condition (1) again transfers from $\schedule$ to $\schedule'$. For similar reasons Condition (2) applies on $\schedule'$: If $b_1 <_{\trans[1]} a_1$ then $b'_1 \le_{\trans[1]}b_1<_{\trans[1]} a_1$.

    \smallskip
    \noindent
    \textsc{Otherwise, if $j=1$}, it follows that $1 < i$. Then, we 
    construct a schedule $\schedule'$ from $\schedule$ by removing all operations from transactions $\trans[h]$ with $i<h$. 
   Notice that we remove at least one transaction, since $i < m$. We can derive a \cqsequence{} $C'$ from $C$ by removing all occurrences of these transactions $\trans[h]$ and adding the \conflictquadruple{} $(\trans[i], b'_i, a'_j, \trans[j])$ instead.

    In this schedule $\schedule'$, Condition (1) and (3) transfer from $\schedule$ by its construction. To see that Condition (2) is true on $\schedule'$, simply notice that if $b_i'$ and $a_1'$ are ww or wr-conflicting, then either $b_i' \ll_\schedule a_j'$ or $b_i' = v_\schedule(a_j')$ or $b_i' \ll_\schedule v_\schedule(a_j')$, which all imply $b_i <_\schedule C_i <_\schedule a_1'$ and thus that $b_1 <_\schedule a_1'$, implying $b_1 <_{\schedule'} a_1'$. 
\end{proof}

\subsection{If direction}
\label{sec:undec:ifdirection}
It remains to argue that the if direction of Theorem~\ref{theo:equaltemplates} is indeed correct.

Next, we show that, if there exists a multiversion split schedule for the set of \ptranss{} in Figure~\ref{fig:pcptemplates} for some set $\dominoset$ of dominoes, then this schedule is always a schedule-encoding of a sequence of dominoes in $\dominoset$.

\begin{prop}\label{pro:pcp-schedule-has-desired-form}
    Let $\dominoset$ be a set of dominoes. If there is a multiversion split schedule $\schedule$ for a set of transactions consistent with the \ptrans{} in Figure~\ref{fig:pcptemplates} for $\dominoset$ and some database $\db$, then this schedule $\schedule$ is a schedule-encoding of some sequence $\vec{d}$ of dominoes in $\dominoset$. 
\end{prop}

\noindent For the proof, let $\db$ be a database and $\schedule = \prefix{\trans[1]}{b_1}\cdot \trans[2]\cdot \ldots \cdot \trans[m] \cdot \postfix{\trans[1]}{b_1}$ a multiversion split schedule for a set of transactions $\transset$ consistent with $\workload$ and $\db$, with the conditions of Lemma~\ref{lem:robustminimalcondition} and based on some sequence of conflict quadruples $C = (\trans[1],b_1,a_2,\trans[2]), (\trans[2],b_2,a_3,\trans[3]) \ldots, (\trans[m],b_m, a_1,\trans[1])$. 
    We show through a sequence of properties (Lemmas~\ref{lem:pcp-lem-1}, \ref{lem:pcp-lem-2}, \ref{lem:pcp-lem-3}, and \ref{lem:pcp-lem-4}), that $\schedule$ is a schedule-encoding of a sequence $\vec{d}$ of dominoes in $\dominoset$.

    As a first property (Lemma~\ref{lem:pcp-dagger}), we observe that \ptranss{} in $\workload$ heavily constrain the possible variable instantiations. For \ptrans{}  $\ttinit$, for example, a variable mapping depends entirely on the choice of the value for variable $\vi$. Since Lemma~\ref{lem:robustminimalcondition} forbids the presence of duplicate transactions in $\transset$, two transactions $T_i$ and $T_j$ (with $i\ne j$) based on \ptrans{} $\ttinit$ cannot agree on their choice for variable $\vi$ in $\schedule$. By applying this argument to other \ptranss{}, we obtain the following corollary of Lemma~\ref{lem:robustminimalcondition}.
        Here, for each transaction $\trans[i]$ in $\schedule$, we write $\tau_i$ to denote the \ptrans{} in $\workload$ that it is based on, and by $\mu_i$ the associated variable mapping for $\tau_i$, with $\mu_i(\tau_i) = \trans[i]$. 

    \begin{lem}\label{lem:pcp-dagger}
        for two transactions $T_i$ and $T_j$ in $\schedule$, with $i\ne j$:
        \begin{itemize}
            \item if $T_i$ and $T_j$ are based on $\ttinit$, then $\mu_i(\vi) \ne \mu_j(\vi)$;
            \item if $T_i$ and $T_j$ are based on $\ttinitb$ then $\mu_i(\vs_1) \ne \mu_j(\vs_1)$; 
            \item if $T_i$ and $T_j$ are based on $\ttclose$ then $\mu_i(\vb) \ne \mu_j(\vb)$ and $\mu_i(\vc) \ne \mu_j(\vc)$;
            \item if $T_i$ and $T_j$ are based on domino \ptranss{} then $\mu_i(\vb) \ne \mu_j(\vb)$ and $\mu_i(\vb_{\text{next}}) \ne \mu_j(\vb_\text{next})$. 
        \end{itemize}
    \end{lem}

    We conclude the proof of Proposition~\ref{pro:pcp-schedule-has-desired-form} with the necessary arguments (Lemmas~\ref{lem:pcp-lem-1}, \ref{lem:pcp-lem-2}, \ref{lem:pcp-lem-3} and \ref{lem:pcp-lem-4}) that $\schedule$ is indeed a schedule-encoding for some sequence of dominoes.

    \begin{lem}\label{lem:pcp-lem-1}
        Transaction $T_1$ is based on $\ttinit$,  $T_2$ is based on $\ttinitb$, and $\mu_1(\vi) = \mu_2(\vi)$, $\mu_1(\vx_1) \ne \mu_2(\vx_2)$, and $\mu_1(\vs_1) = \mu_2(\vs_1) \ne \mu_2(\vs_0)$.
    \end{lem}
    \begin{proof}
Since $b_{1}$ and $a_{2}$ are rw-conflicting (cf, Definition~\ref{def:mvsplitschedule}), and there are no updates in the considered \ptranss{}, operation $b_1$ must be a read. Since \textsf{InitialConflict} is the only type allowing for conflicts involving a read, it is immediate that $T_1$ must be based on $\ttinit$ and $T_2$ based on $\ttinitb$, with $\mu_1(\vi) = \mu_2(\vi)$.  From this equality and function $\ufunc{final-dominoes-string}$ it follows that $\mu_1(\vs_1) =  \mu_2(\vs_1)$. 
From Definition~\ref{def:mvsplitschedule}, particularly that there is no ww-conflict between a write operation in $\prefix{\trans[1]}{b_1}$ and a write operation in any of the transactions $T_2, \ldots, T_m$, it follows that $\mu_1(\vx_1) \ne \mu_2(\vx_2)$.
Finally, function $\ufunc{is-non-empty}$, which maps $\vs_1$ onto $\vx_1$ in \ptrans{} $\ttinit$ and $\vs_0$ onto $\vx_2$ in \ptrans{} $\ttinitb$, implies $\mu_1(\vs_1) \ne \mu_2(\vs_0)$. 
    \end{proof}

        \begin{lem}\label{lem:pcp-lem-2}
            There is a transaction $T_3$ in $\schedule$ and it is based on a domino \ptrans{}, with $\mu_3(\vs_1) = \mu_2(\vs_1)$.
        \end{lem}
        \begin{proof}
            First, suppose towards a contradiction that $m=2$. We already know from Lemma~\ref{lem:pcp-lem-1} that 
$\mu_1(\vi) = \mu_2(\vi)$ and $\mu_1(\vx_1) \ne \mu_2(\vx_2)$, thus $a_1 = b_1 = \R[]{\mu_1(\vi)}$ and $b_2 = a_2 = \W[]{\mu_2(\vi)}$, indicating $b_1 \rightarrow_s a_2$, particularly, $v_\schedule(b_1) \ll_s a_2$, thus implying that $a_1$ cannot depend on $b_2$, which is the desired contradiction.

            The remainder of the proof is by exclusion. Transaction $T_3$ is not based on \ptrans{} $\ttinitb$, because all possible conflicts between $T_2$ and an instantiation of \ptrans{} $\ttinitb$ (implying either $\mu_2(\vx_2) = \mu_3(\vx_2)$,  $\mu_2(\vb) =  \mu_3(\vb)$, or $\mu_2(\vi) =  \mu_3(\vi)$) would imply the equality $\mu_2(\vs_1) = \mu_3(\vs_1)$ (through functional constraints $\vi = \ufunc{defines}(\vx_2)$, $\vs_1 = \ufunc{future-solution-string}(\vb)$, and $\vs_1 = \ufunc{final-dominoes-string}(\vi)$), which is forbidden by Lemma~\ref{lem:pcp-dagger}.
The argument that transaction $T_3$ cannot based on \ptrans{} $\ttinit$ is similar: every possible conflict between $T_2$ and an instantiation of $\ttinit$ implies $\mu_1(\vi) = \mu_2(\vi) = \mu_3(\vi)$ either directly (taking $\mu_2(\vi_2) = \mu_3(\vi_1)$ as conflict) or, when taking 
            $\mu_2(\vx_2) = \mu_3(\vx_1)$ as conflict, through constraints $\vi = \ufunc{defines}(\vx_1)$ and $\vi = \ufunc{defines}(\vx_2)$ in $\trans[2]$ and $\trans[3]$, respectively. Either way, $\mu_1(\vi) = \mu_3(\vi)$ is forbidden by Lemma~\ref{lem:pcp-dagger}.
Finally, to see that $\trans[3]$ is not based on \ptrans{} $\ttclose$, we observe that a conflict between $\trans[2]$ and an instantiation of \ptrans{} $\ttclose$ must be ww-conflicting involving variables $\vb$, thus with $\mu_2(\vb) = \mu_3(\vb)$.
Then, $\mu_2(\vs_0) = \mu_3(\vs_t)$, due to functional constraint $\vs_0 = \ufunc{top-string}(\vb)$ in $\trans[2]$ and $\vs_t = \ufunc{top-string}(\vb)$ in $\trans[3]$, and $\mu_2(\vs_1) = \mu_3(\vs_1)$, due to functional constraint $\vs_1 = \ufunc{future-solution-string}(\vb)$ in $\trans[2]$ and $\trans[3]$.
            However, we also have $\mu_3(\vs_1) = \mu_3(\vs_t)$, due to constraints $\vs_1 = \ufunc{solution-string}(\vb)$ and $\vs_t=\ufunc{solution-string}(\vb)$, thus implying $\mu_2(\vs_0) = \mu_3(\vs_t) = \mu_3(\vs_1) = \mu_2(\vs_1)$, which contradicts with earlier proven Lemma~\ref{lem:pcp-lem-1}.
We conclude that $T_3$ is indeed based on a domino \ptrans{}. 
Therefore, the conflict quadruple $(T_2, b_2, a_3, T_3)$ must admit ww-conflicting operations over variable $\vb$ in $\trans[2]$ and either variable $\vb$ or $\vb_{next}$ in $\trans[3]$. We notice that $\mu_2(\vs_1) = \ufunc{future-solution-string}(\mu_2(\vb))$, $\mu_3(\vs_1) = \ufunc{future-solution-string}(\mu_3(\vb))$, and $\mu_3(\vs_1) = \ufunc{future-solution-string}(\mu_3(\vb_{next}))$, thus independent of the variable $\vb_{next}$ or $\vb$ in $\trans[3]$, we have $\mu_2(\vs_1) = \mu_3(\vs_1)$.
    \end{proof}

    \begin{lem}\label{lem:pcp-lem-3}
    For a transaction $T_i$, with $i\ge 4$, for which all $T_j$'s, with $j \in \{3,\ldots,i-1\}$, are based on domino \ptranss{}, transaction $T_i$ is based on a domino \ptrans{} or on \ptrans{} $\ttclose$. Furthermore $\mu_2(\vs_1) = \mu_i(\vs_1)$.
    \end{lem}
    \begin{proof}
Since domino \ptranss{} do not mention variables of type \textsf{InitialConflict} and write only to variables of type \textsf{DominoSequence}, it remains to show that $\trans[i+1]$ is not based on \ptrans{} $\ttinitb$. 

For this,  observe that $\mu_2(\vs_1) = \mu_{i-1}(\vs_1)$. Indeed, every conflict quadruple $(T_i, b_i, a_{i+1},\allowbreak T_{i+1})$, with $i \in \{3,\ldots, i-1\}$, admits ww-conflicting operations with variables of type \textsf{DominoSequence}. No matter if the conflict is via a variable $\vb$ or $\vb_{next}$, the constraints $\vs_1 = \ufunc{future-solution-string}(\vb)$ and $\vs_1 = \ufunc{future-solution-string}(\vb_{next})$ ensure $\mu_2(\vs_1) = \mu_{i-1}(\vs_1)$.

Now, assume towards a contradiction that $\trans[i+1]$ is based on $\ttinitb$, thus admitting a conflict quadruple $(T_{i}, b_i, a_{i+1}, T_{i+1})$ in $C$.
Then either $b_i=\mu_i(\vb)$ and $a_{i+1} = \mu_{i+1}(\vb)$ \emph{or} $b_i=\mu_i(\vb_\text{next})$ and $a_{i+1} = \mu_{i+1}(\vb)$. Both of these equalities imply $\mu_i(\vs_1) = \mu_{i+1}(\vs_1)$ due to constraints
$\vs_1 = \ufunc{future-solution-string}(\vb)$ and $\vs_1 = \ufunc{future-solution-string}(\vb_{next})$, thus implying $\mu_i(\vs_1) = \mu_2(\vs_1)$, this contradict with Lemma~\ref{lem:pcp-dagger}. We conclude that $\trans[i]$ is indeed based on a domino \ptrans{} or on \ptrans{} $\ttclose$. That $\mu_{i-1}(\vs_1) = \mu_i(\vs_1)$ follows again from the constraints using function $\ufunc{future-solution-string}$.
\end{proof}

\begin{lem}\label{lem:pcp-lem-4}
    If $T_i$ is based on \ptrans{} $\ttclose$, then $i=m$.
\end{lem}
\begin{proof}
Let $T_j$ be the transaction following $T_i$. We already know about $T_j$ that either $j=1$ or must be a transaction that is different to all foregoing transactions  $T_1, \ldots, T_i$ (due to Lemma~\ref{lem:robustminimalcondition}).

    We first show, by exclusion, that transaction $T_j$ is based on $\ttinit$: Transaction $T_{j}$ cannot be based on $\ttclose$, as then either $\mu_i(\vb) = \mu_{j}(\vb)$ or $\mu_i(\vc) = \mu_{j}(\vc)$, which directly contradicts Lemma~\ref{lem:pcp-dagger}.
Similarly, transaction $T_{j}$ cannot be based on $\ttinitb$,  as then $\mu_i(\vb) = \mu_{j}(\vb)$ implying $\mu_2(\vs_1) = \mu_i(\vs_1) = \mu_{j}(\vs_1)$, due to the constraints involving function $\ufunc{future-solution-string}$.
Finally, transaction $T_{j}$ cannot be based on a domino \ptrans{}, because then $\mu_{j}(\vb_\text{next}) = \mu_{i}(\vb) = \mu_{j}(\vb)$
    or $\mu_{i-1}(\vb_\text{next}) = \mu_{i}(\vb) = \mu_{j}(\vb_\text{next})$, thus with $T_{i}$ and $T_{j}$ contradicting Lemma~\ref{lem:pcp-dagger}. We can thus indeed conclude that transaction $T_j$ is based on $\ttinit$. 

To see that $j=1$, recall that $\mu_1(\vs_1) = \mu_{i-1}(\vs_1)$ and the only possible conflict between $T_i$ and $T_j$ implies $\mu_i(\vc) = \mu_{j}(\vc)$. 
From the latter we obtain  $\mu_{i-1}(\vs_1) = \mu_{i}(\vs_1)$, due to $\mu_{i-1}(\vb_\text{next}) = \mu_{i}(\vb)$ and function $\ufunc{future-solution-string}$.
    From this it follows that $\mu_1(\vi) = \mu_j(\vi)$ through $\vi = (\ufunc{defines} \circ \ufunc{is-non-empty})(\vs_1)$ in \ptrans{} $\ttinit$. That $j=1$ then follows from Lemma~\ref{lem:pcp-dagger}. 
\end{proof}

Finally, we show that if there is a multiversion split schedule with the properties of Lemma~\ref{lem:robustminimalcondition} that is a schedule-encoding for a sequence of dominoes $\vec{d}$, then this sequence $\vec{d}$ is also a solution to the respective PCP problem. The next Proposition thus finalizes the proof for the if-direction of Theorem~\ref{theo:equaltemplates}.

\begin{prop}\label{pro:pcp-result}
    Let $\dominoset$ be a set of dominoes. Let $\schedule$ be a multiversion split schedule with the properties of Lemma~\ref{lem:robustminimalcondition} that is consistent with the \ptranss{} in Figure~\ref{fig:pcptemplates} for $\dominoset$ and with some database $\db$. If $\schedule$ is a schedule-encoding of a sequence $\vec{d}$ of dominoes in $\dominoset$, then $\vec{d}$ is a solution for the PCP problem on input $\dominoset$. 
\end{prop}

\begin{proof}
Let $a_1a_2\ldots a_h$ and $b_1b_2\ldots b_k$ be the two strings (with $a_i, b_i \in \Sigma$) obtained by reading from left to right, symbol by symbol, the values on the top, respectively, the bottom of dominoes $d_1, \ldots, d_r$. Let us say that $a_1a_2\ldots a_h = \vec{a_1}\vec{a_2}\ldots\vec{a_r}$ and $b_1b_2\ldots b_k = \vec{b_1}\vec{b_2}\ldots\vec{b_r}$. Notice that $h$ and $k$ are not necessarily equal to $r$ as the top and bottom value of an individual domino can be of different length.

For convenience of notation, we introduce for every $i \in \{0, \ldots, h\}$ and $j \in \{1, \ldots, k\}$ the following notation:
\begin{align*}
\alpha_i \mydef & (\ufunc{append-$a_i$}\circ  \ufunc{append-$a_{i-1}$} \circ \cdots \circ \ufunc{append-$a_1$})(\mu_2(\vs_0)). \\
\beta_j \mydef & (\ufunc{append-$b_j$}\circ  \ufunc{append-$b_{j-1}$} \circ \cdots \circ \ufunc{append-$b_1$})(\mu_2(\vs_0)). 
\end{align*}

    First, we show that 
    \begin{align}
        \alpha_h = \mu_2(\vs_1) = \beta_k.\label{align:a}
    \end{align}
    This result follows from the assumed structure of schedule $\schedule$. More precisely, since an instantiation of $\ttinitb$ with an instantiation of $\ttdomino{d_1}$ can only have conflicts on instantiations of $\W[]{\vb:DominoSequence}$, we have $\mu_2(\vb) = \mu_3(\vb)$, from which it follows that $\mu_2(\vs_1) = \mu_3(\vs_1)$.

    For every individual instantiation of $\ttdomino{d_i}$ in $\schedule$, we have that $\ufunc{append-$a_i^{\ell_a}$}\circ\cdots \circ \ufunc{append-$a_i^1$}(\mu_i(\vs_t)) = \mu_i(\vs_{t\vec{a_i}})$ and $\ufunc{append-$b_i^{\ell_b}$}\circ\cdots \circ \ufunc{append-$b_i^1$}(\mu_i(\vs_b)) = \mu_i(\vs_{b\vec{b_i}})$, with $\vec{a_i} = a_1a_2\ldots a_{\ell_a}$ and $\vec{b_i} = b_1b_2\ldots b_{\ell_b}$. 

    For transactions $\trans[i]$, with $i \in \{3,\ldots,m+1\}$, (thus representing an instantiation of $\ttdomino{d_{i-2}}$ which is followed in $\schedule$ by an instantiation of $\ttdomino{d_{i-1}}$), the only possible conflict is between the instantiation of $\W[]{\vb_{next}:DominoSequence}$ ($\trans[i]$) and the instantiation of $\W[]{\vb:DominoSequence}$ in ($\trans[i+1]$) -- notice that this is indeed the only option due to Lemma~\ref{lem:pcp-dagger} -- thus with $\mu_i(\vb_{next}) = \mu_{i+1}(\vb)$, implying $\mu_i(\vs_{t\vec{a_i}}) = \mu_{i+1}(\vs_t)$, $\mu_i(\vs_{b\vec{b_i}}) = \mu_{i+1}(\vs_b)$, and $\mu_{i}(\vs_1) = \mu_{i+1}(\vs_1)$.

    Finally, transaction $\trans[m-1]$ (an instantiation of $\ttdomino{d_m}$) can only conflict with transaction $\trans[m]$ (an instantiation of $\ttclose$) on instantiations of $\vb_{next}$ (in $\ttdomino{d_m}$, and $\vb$ (in $\ttclose$), thus with $\mu_{m-1}(\vb_{next}) = \mu_{m}(\vb)$, implying $\mu_{m-1}(\vs_{t\vec{a_r}}) = \mu_m(\vs_t) = \mu_{m}(\vs_b) = \mu_{m-1}(\vs_{b\vec{b_r}})$. 

    Combining the above equalities indeed proves Condition~(\ref{align:a}).

    From Condition~(\ref{align:a}) we can now derive that,
    \begin{align}
        \alpha_i = \mu_2(\vs_1) = \beta_i, \text{ for every $i \in \{ 1,\ldots \min\{h,k\}\}$,}\label{align:b}
    \end{align}
    by following an analogous approach. Indeed, in every instantiation of $\ttdomino{i}$, there is a functional constraint for every application of the append function that requires its input to be the result of the detach function applied over its output, which indeed implies Condition~(\ref{align:b}).

    To see that $k=h$, we observe that $k \ne h$ implies an application of the detach function over the instantiation of $\vs_1$ (for which we already argued it has the same tuple assigned for every domino instantiation) for the shortest string, which contradicts with Condition~(\ref{align:a}) because such an application results in the same instantiation as $\vs_e$,  which can never equal the instantiation for $\vs_1$.

    The desired result that the individual symbols in the top and bottom reads of dominoes in sequence $\vec{d}$ are the same now follows from the functional constraint that every interpretation of a string sequence mapped via function $\ufunc{top}$ onto either the interpretation for $\vs_1$ (representing symbol $1 \in \Sigma$) or $\vs_0$ (representing symbol $0 \in \Sigma$).
\end{proof}

\section{Robustness for \shortPtranss{} admitting Multi-Tree Bijectivity}
\label{sec:MTBij}

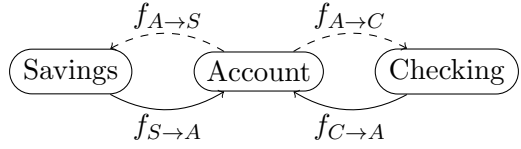
\begin{figure}[t]
        \centering
        \begin{tikzpicture}[scale=1]
        \node[draw, rounded rectangle,align=center] (A) at (0,0) {\Account};
        \node[draw, rounded rectangle,align=center] (S) at (-2.5,0) {\Savings};
        \node[draw, rounded rectangle,align=center] (C) at (2.5,0) {\Checking};
        \draw[->,dashed] (A) to [bend right] node[midway,above=-.2em] {$\fas$} (S) ;
        \draw[->] (S) to [bend right] node[midway,below=-.2em] {$\fsa$} (A) ;
        \draw[->,dashed] (A) to [bend left] node[midway,above=-.2em] {$\fac$} (C) ;
        \draw[->] (C) to [bend left] node[midway,below=-.2em] {$\fca$} (A) ;
        \end{tikzpicture}
        \caption{Schema graph for the SmallBank benchmark. The dashed edges correspond to the multi-tree schema graph for the schema restricted to $\fas$ and $\fac$.}
        \label{fig:sb:sg}
    \end{figure}

We say that a set of \ptranss{} $\workload$ over a schema $(\schRel, \schFunc)$ admits \emph{multi-tree bijectivity} if a disjoint partitioning of $\schFunc$ in pairs $(f_1, g_1),(f_2, g_2), \ldots, (f_n, g_n)$ exists such that $\domf(f_i) = \rangef(g_i)$ and $\domf(g_i) = \rangef(f_i)$ for every pair of function names $(f_i, g_i)$;
every schema graph $\sg{\schRel, \{h_1, h_2, \ldots, h_n\}}$ over the schema restricted to function names $\{h_1, h_2, \ldots, h_n\}$ (with $h_i = f_i$ or $h_i= g_i$) is a multi-tree; and,
for every pair of function names $(f_i, g_i)$ and for every pair of variables $\vx, \vy$ occurring in a \shortptrans{} $\tau_j \in \workload$, we have $f_i(\vx) = \vy \in \Cset_j$ iff $g_i(\vy) = \vx \in \Cset_j$.
Intuitively, we can think of $f_i$ as a bijective function, with $g_i$ its inverse. We denote the class of all sets of \shortptranss{} admitting multi-tree bijectivity by $\mtbtemplates$. The SmallBank benchmark {given in Figure~\ref{fig:smallbank-abstract-syntax}} is in $\mtbtemplates$, witnessed by the partitioning $\{(\fac, \fca), (\fas, \fsa)\}$. {For example}, the schema graph restricted to $\fac$ and $\fas$ is a tree and therefore also a \multitree, as illustrated in Figure~\ref{fig:sb:sg}. 

The next theorem allows disequalities whereas Theorem~\ref{theo:equaltemplates} does not require them.

\begin{thm}\label{theo:1dettemplates}\label{theo:mtbtemplates}
    \RobustnessTempTwo{\mtbtemplates}{\mvrc} is decidable in \NLOGSPACE.
\end{thm}

The approach followed in the proof of Theorem~\ref{theo:1dettemplates} is to repeatedly pick a transaction template while maintaining an overall consistent variable mapping in search for a counterexample multiversion split schedule 
that by Theorem~\ref{theo:characterization:split-shedules} suffices to show that robustness does not hold. 
The main challenge is to show that a variable mapping consistent with all functional constraints can be maintained in logarithmic space and that all requirements for a multiversion split schedule can be verified in \NLOGSPACE.

Central to our approach is a generalization of conflicting operations.
Let $\workload$ be a set of \ptranss{}. 
 For
  $\tau_i$ and $\tau_j$ in $\workload$, we say that an operation $o_i \in \tau_i$ is \emph{potentially conflicting} with an operation $o_j \in \tau_j$ if {$o_i$ and $o_j$ are operations over a variable of the same type, and at least one of the following holds:}
}
\begin{itemize}
    \item $\WriteSet{o_i}\cap \WriteSet{o_j}\neq \emptyset$ (potentially ww-conflicting);
    \item $\WriteSet{o_i}\cap \ReadSet{o_j}\neq \emptyset$ (potentially wr-conflicting); or
    \item $\ReadSet{o_i}\cap \WriteSet{o_j}\neq \emptyset$ (potentially rw-conflicting).
\end{itemize}
Intuitively, potentially conflicting operations lead to conflicting operations when the variables of these operations are mapped to the same tuple by a variable assignment. In analogy to \conflictquadruples{} over a set of transactions {as in Definition~\ref{def:mvsplitschedule}}, we consider \emph{\pconflictquadruples{}} $(\tau_i, o_i, p_j,  \tau_j)$ over 
$\workload$ with $\tau_i, \tau_j \in \workload$, and $o_i \in \tau_i$ an operation that is potentially conflicting with an operation $p_j \in \tau_j$.
For a \pcqsequence{} $D = (\tau_1, o_1, \allowbreak p_2,\allowbreak \tau_2), \ldots, (\tau_m, o_m, p_1, \tau_1)$ over  
$\workload$, we write $\transcopies{D}$ to denote the set $\{\tau_1,\ldots,\tau_m\}$ of \ptranss{} mentioned in $D$. For ease of exposition, we assume a variable renaming such that any pair of templates in $\transcopies{D}$ uses a disjoint set of variables.\footnote{To be formally correct, the latter would require to add every such variable-renamed template to $\workload$ creating a larger set $\workload'$. This does not influence the complexity of Theorem~\ref{theo:mtbtemplates} as $\transcopies{D}$ nor $\workload'$ are used in the algorithm. Their only purpose is to reason about properties of $\bar \mu$.}
The sequence $D$
induces a \cqsequence{} $C = (\trans[1], b_1, a_2, \trans[2]),\allowbreak \ldots,\allowbreak (\trans[m], b_m, a_1, \trans[1])$
by applying a variable assignment $\mu_i$ to each $\tau_i$ in $\transcopies{D}$.
We call such a set of variable assignments simply a \emph{variable mapping} for $D$, denoted $\barmu$, and write $\barmu(D) = C$. For a variable $\vx$ occurring in a \shortptrans{} $\tau_i$, we write $\barmu(\vx)$ as a shorthand notation for $\mu_i(\vx)$, with $\mu_i$ the variable assignment over $\tau_i$ in $\barmu$. This is well-defined as all templates in $\transcopies{D}$ are variable-disjoint.
Furthermore, $\barmu(\myvar{o_i}) = \barmu(\myvar{p_j})$ for each \pconflictquadruple{} $(\tau_i, o_i, p_j, \tau_j)$ in $D$ 
as otherwise the induced quadruple $(\trans[i], b_i, a_j, \trans[j])$ is not a valid \conflictquadruple{} in $C$.
We {say} that a variable mapping $\barmu$ is admissible for a database $\db$ if every variable assignment $\mu_i$ in $\barmu$ is admissible for $\db$.

A basic insight is that if there is a multiversion split schedule $\schedule$ for some $C$ over a set of transactions $\transset$ consistent with $\workload$ and a database $\db$,
then there is a \pcqsequence{} $D$ such that $\barmu(D) = C$ for some $\barmu$. We will verify the existence of such a $C$, 
satisfying the properties of Definition~\ref{def:mvsplitschedule}, 
by nondeterministically constructing $D$ on-the-fly together with a mapping $\bar\mu$. We show in Lemma~\ref{lemma:mtbschedule:types} that when $\workload\in \mtbtemplates$,
$\bar \mu$ is a collection of disjoint type mappings  (that map variables of the same type to the same tuple) such that variables that are ``connected'' in $D$ (in {a} way that we will make precise next) are mapped using the same type mapping. Lemma~\ref{lemma:mtbchar} then shows that already a constant number of those type mappings suffice. 

We introduce the necessary notions to capture when two variables are connected in $D$. We can think of equality constraints $\vy = f(\vx)$ in a template $\tau$ as constraints on the possible variable assignments $\mu$ for $\tau$ when a database $\db$ is given. Indeed, if we fix $\mu(\vx)$ to a tuple in $\db$, then $\mu(\vy) = f^\db(\mu(\vx))$ is immediately implied.
These constraints can cause a chain reaction of implications. If for example $\vz = g(\vy)$ is a constraint in $\tau$ as well, then $\mu(\vx)$ immediately implies $\mu(\vz) = g^\db(f^\db(\mu(\vx)))$. We formalize this notion of implication next. We use sequences of function names 
$F = f_1\cdots f_n$, 
denoting the empty sequence as $\epsilon$ and the concatenation of two sequences $F$ and $G$ by $F \cdot G$.
For two variables $\vx, \vy$ occurring in a \shortptrans{} $\tau$ and a (possibly empty) sequence of function names $F$, we say that \emph{$\vx$ implies $\vy$ by $F$ in $\tau$}, denoted $\vx \varimplies{F}{\tau} \vy$, if $\vx = \vy$ and $F = \epsilon$ or if
there is a variable $\vz$ such that $\vy = f(\vz)$ is a constraint in $\tau$, $\vx \varimplies{F'}{\tau} \vz$ and $F = F' \cdot f$.
We next extend the notions of implication 
to sequences of \pconflictquadruples{}.
Let $D = (\tau_1, o_1, \allowbreak p_2,\allowbreak \tau_2), \ldots, (\tau_m, o_m, p_1, \tau_1)$ be a \pcqsequence{},
and let $\vx$ and $\vy$ be two variables occurring in \shortptranss{} $\tau_i$ and $\tau_j$ in $\transcopies{D}$, respectively.
Then \emph{$\vx$ implies $\vy$ by a sequence of function names $F$ in $D$}, denoted $\vx \varimplies{F}{D} \vy$ if
\begin{itemize}
    \item $i=j$
    and $\vx \varimplies{F}{\tau_i} \vy$ (implication within the same template);
    \item $F = \epsilon$ and $(\tau_i, o_i, p_j, \tau_j)$ or $(\tau_j, o_j, p_i, \tau_i)$ is a \pconflictquadruple{} in $D$ with $o_i$ (respectively $p_i$) an operation over $\vx$ and $p_j$ (respectively $o_j$) an operation over $\vy$ (implication between templates, notice that $\vx \varimplies{\epsilon}{D} \vy$ iff $\vy \varimplies{\epsilon}{D} \vx$); or
    \item there exists a variable $\vz$ such that $\vx \varimplies{F_1}{D} \vz$ and $\vz \varimplies{F_2}{D} \vy$ with $F = F_1 \cdot F_2$.
\end{itemize}
 
\noindent Two variables $\vx$ and $\vy$ occurring in $\transcopies{D}$ are \emph{connected in $D$}, denoted \conn{\vx}{\vy}{D}, if $\vx \varimplies{F}{D} \vy$ or $\vy \varimplies{F}{D} \vx$, 
or if there is a variable $\vz$ with $\conn{\vx}{\vz}{D}$ 
and either $\vz \varimplies{F}{D} \vy$ or $\vy \varimplies{F}{D} \vz$ for some sequence 
$F$. Furthermore, two variables $\vx$ and $\vy$ occurring in a  \shortptrans{} $\tau$ are \emph{connected in $\tau$}, denoted
\conn{\vx}{\vy}{\tau}, if $\vx \varimplies{F}{\tau} \vy$ or $\vy \varimplies{F}{\tau} \vx$, 
or if there is a variable $\vz$ with $\conn{\vx}{\vz}{\tau}$ 
and either $\vz \varimplies{F}{\tau} \vy$ or $\vy \varimplies{F}{\tau} \vz$ for some sequence 
$F$. These definitions of connectedness can be trivially extended to operations over variables: two operations in $D$ (respectively $\tau$) are connected in $D$ (respectively $\tau$) if they are over variables that are connected in $D$ (respectively $\tau$). When $F$ is not important we drop it from the notation. For instance, we denote by 
$\vx \varimplies{}{D} \vy$ that there is an $F$ with $\vx \varimplies{F}{D} \vy$. 

\begin{lem}\label{lemma:mtbconnected}
    Let $D$ be a \pcqsequence{} over
    $\workload\in \mtbtemplates$. Then \conn{\vx}{\vy}{D}
    implies $\vx \varimplies{}{D} \vy$ and $\vy \varimplies{}{D} \vx$.
    Furthermore, if $\type{\vx} = \type{\vy}$ 
        then $\barmu(\vx) = \barmu(\vy)$ for every variable mapping $\barmu$ for $D$ that is admissible for some database $\db$. 
\end{lem}

\noindent Before proving the correctness of Lemma~\ref{lemma:mtbconnected}, we first present two additional lemmas that will be used in the correctness proof.

\begin{lem}\label{lemma:mtbschemaprops}
    Let $(\schRel, \schFunc)$ be a schema for which a disjoint partitioning of $\schFunc$ in pairs $P = (f_1, g_1),(f_2, g_2), \ldots, (f_n, g_n)$ exists such that $\domf(f_i) = \rangef(g_i)$ and $\domf(g_i) = \rangef(f_i)$ for every $(f_i, g_i) \in P$ and
    every schema graph $\sg{\schRel, \{h_1, h_2, \ldots, h_n\}}$ over the schema restricted to function names $\{h_1, h_2, \ldots, h_n\}$ with $h_i \in (f_i, g_i)$ is a multi-tree. Then:
    \begin{enumerate}
        \item\label{lemma:mtbschemaprops:1} there is no function name $f \in \schFunc$ with $\domf(f) = \rangef(f)$; and
        \item\label{lemma:mtbschemaprops:2} for every path in $\sg{\schRel, \schFunc}$ %, say visiting the nodes
        \new{visiting a sequence of nodes} $R_1, R_2, \ldots, R_{m-1}, R_m$, if $R_1 = R_m$ and $R_1 \neq R_i$ for every $i \in [2, m-1]$,
        then $R_2 = R_{m-1}$ and $(f,g)$ is a pair in $P$ with $f$ the edge from $R_1$ to $R_2$ and $g$ the edge from $R_{m-1}$ to $R_m$.
    \end{enumerate}
\end{lem}

\begin{proof}
    Towards a contradiction, assume \emph{(\ref{lemma:mtbschemaprops:1})} does not hold. That is, there is a function name $f_i$ with $\domf(f_i) = \rangef(f_i) = R$ for some type $R$. Let $g_i$ be the function name such that $(f_i,g_i)$ is a pair in $P$. By definition, $\domf(g_i) = \rangef(g_i) = R$. But then we cannot pick a $h_i \in (f_i, g_i)$ such that the resulting schema graph is a multi-tree. Indeed, in both cases, there is a self-loop on $R$, leading to the desired contradiction.
   
    For \emph{(\ref{lemma:mtbschemaprops:2})}, assume towards a contradiction that $R_2 \neq R_{m-1}$. Without loss of generality, we can assume that each node is visited only once in $R_2, \ldots, R_{m-1}$.
    Otherwise, $R_2, \ldots, R_{m-1}$ contains a loop that can be removed from this sequence without altering $R_2$ and $R_{m-1}$.
    Since $R_1, R_2, \ldots, R_{m-1}, R_m$ is a path in $\sg{\schRel, \schFunc}$, there is a sequence of function names $F = e_1 \cdots e_{m-1}$ such that each $e_i$ is an edge from $R_i$ to $R_{i+1}$ in $\sg{\schRel, \schFunc}$, implying $\domf(e_i) = R_i$ and $\rangef(e_i) = R_{i+1}$.
    By assumption that each type $R_i$ occurs only once in $R_2, \ldots, R_{m-1}$ (notice that, for $i=1$, this follows from Condition~(2) of the lemma) and type $R_1 = R_m$ does not appear in $R_2, \ldots, R_{m-1}$, there is no pair of function names $e_i$ and $e_j$ in $F$ with $i\ne j$, $\domf(e_i) = \rangef(e_j)$ and $\rangef(e_i) = \domf(e_j)$.
    Therefore, at most one function name of each pair in $P$ appears in $F$.
    But then we can choose $h_i = e_i$ for each such pair in $P$, with $e_i$ the function name appearing in $F$.
    Since $F$ describes a cycle in $\sg{\schRel, \schFunc}$, the resulting schema graph restricted to these $h_i$ cannot be a multi-tree, as it contains a cycle.
    
    It remains to argue that if $f$ is the edge from $R_1$ to $R_2$ and $g$ is the edge from $R_{m-1}$ to $R_m$ on this path, then $(f, g)$ is a pair in $P$.
    To this end, note that $\domf(f) = \rangef(g)$ and $\rangef(f) = \domf(g)$, as $R_1 = R_m$ and $R_2 = R_{m-1}$.
    If $(f, g)$ is not a pair in $P$, then there are two pairs $(f, f')$ and $(g, g')$ in $P$ with $\domf(f) = \domf(g') = \rangef(f') = \rangef(g) = R_1$ and $\rangef(f) = \rangef(g') = \domf(f') = \domf(g) = R_2$.
    Then we can choose $f$ in $(f, f')$ and $g$ in $(g, g')$.
    Since the resulting schema graph cannot be a multi-tree, as there is a cycle between $R_1$ and $R_2$, this choice leads to a contradiction.
\end{proof}

\begin{lem}\label{lemma:varimpliessymmetry}
    Let $D$ be a \pcqsequence{} over $\workload\in \mtbtemplates$. Then
    \begin{enumerate}
        \item $\vx \varimplies{}{\tau} \vy$ iff $\vy \varimplies{}{\tau} \vx$ for every pair of variables $\vx$ and $\vy$ occurring in a template $\tau$; and
        \item $\vx \varimplies{}{D} \vy$ iff $\vy \varimplies{}{D} \vx$ for every pair of variables $\vx$ and $\vy$ occurring in $D$.
    \end{enumerate}
\end{lem}

\begin{proof}
    \emph{(1)} We argue by induction on the definition of $\vx \varimplies{}{\tau} \vy$ that $\vx \varimplies{}{\tau} \vy$ implies $\vy \varimplies{}{\tau} \vx$. The other direction is analogous. The base case is immediate, as $\vx = \vy$ implies $\vy \varimplies{}{\tau} \vx$ by definition. For the inductive case, assume a variable $\vz$ such that $\vy = f(\vz)$ is a constraint in $\tau$ and $\vx \varimplies{}{\tau} \vz$.
    By the induction hypothesis, $\vz \varimplies{F}{\tau} \vx$ for some sequence of function names $F$. Since $\workload \in \mtbtemplates$, there is a constraint $\vz = f'(\vy)$ in $\tau$ as well. It follows that $\vy \varimplies{F'}{\tau} \vx$ with $F' = f' \cdot F$.
    
    \emph{(2)} We argue by induction on the definition of $\vx \varimplies{}{D} \vy$ that $\vx \varimplies{}{D} \vy$ implies $\vy \varimplies{}{D} \vx$. The other direction is again analogous. The first base case is now immediate, as we already argued that $\vx \varimplies{}{\tau} \vy$ implies $\vy \varimplies{}{\tau} \vx$.
    For the second base case, assume $\vx \varimplies{\epsilon}{D} \vy$ and $(\tau_i, o_i, p_j, \tau_j)$ is a \pconflictquadruple{} in $D$ with $\myvar{o_i} = \vx$ and $\myvar{p_j} = \vy$ (the case for $(\tau_j, o_j, p_i, \tau_i)$ is analogous). $\vy \varimplies{\epsilon}{D} \vx$ then follows by definition.
    For the inductive case, let $\vz$ be a variable such that $\vx \varimplies{F_1}{D} \vz$ and $\vz \varimplies{F_2}{D} \vy$. Then by induction hypothesis $\vz \varimplies{F_1'}{D} \vx$ and $\vy \varimplies{F_2'}{D} \vz$ for some sequence of function names $F_1'$ and $F_2'$.
    By definition, $\vy \varimplies{F'}{D} \vx$ with $F' = F_2' \cdot F_1'$.
\end{proof}

We are now ready to prove the correctness of Lemma~\ref{lemma:mtbconnected}.

\begin{proof}[Proof of Lemma~\ref{lemma:mtbconnected}]
Assuming $\conn{\vx}{\vy}{D}$, we first show by induction on the definition of connectedness that $\vx \varimplies{}{D} \vy$.
    By Lemma~\ref{lemma:varimpliessymmetry}, $\vy \varimplies{}{D} \vx$ then follows.
    For the base case, both $\vx \varimplies{}{D} \vy$ and $\vy \varimplies{}{D} \vx$ imply $\vx \varimplies{}{D} \vy$, where the former is immediate and the latter is by Lemma~\ref{lemma:varimpliessymmetry}.
    For the inductive case, let $\vz$ be a variable with $\conn{\vx}{\vz}{D}$ and either $\vz \varimplies{}{D} \vy$ or $\vy \varimplies{}{D} \vz$. Again, $\vz \varimplies{F_2}{D} \vy$ for some sequence of function names $F_2$ is implied in both cases. By induction hypothesis, $\vx \varimplies{F_1}{D} \vz$ for some sequence of function names $F_1$.
    As a result, $\vx \varimplies{F}{D} \vy$ with $F = F_1 \cdot F_2$.
    
  Next, let $\vx$ and $\vy$ be two variables occurring in $\transcopies{D}$ with $\conn{\vx}{\vy}{D}$ and $\type{\vx} = \type{\vy}$ and let $\barmu$ be a variable mapping for $D$ that is admissible for a database $\db$. We prove that $\barmu(\vx) = \barmu(\vy)$.
    
    We already argued that $\conn{\vx}{\vy}{D}$ implies $\vx \varimplies{}{D} \vy$.
    By definition of $\vx \varimplies{}{D} \vy$, there is a sequence of variables $\vx_1, \vx_2\, \ldots, \vx_n$ with $\vx_1 = \vx$ and $\vx_n = \vy$ such that for each pair of adjacent variables $\vx_i$ and $\vx_{i+1}$:
    \begin{itemize}
        \item[$(\dagger)$] $\vx_i$ and $\vx_{i+1}$ both occur in the same \shortptrans{} $\tau \in \transcopies{D}$ and $\vx_{i+1} = f(\vx_i) \in \Cset(\tau)$ for some function name $f$; or 
        \item[$(\ddagger)$] $\type{\vx_i} = \type{\vx_{i+1}}$ and there is a \pconflictquadruple{} $(\tau_j, o_j, p_k, \tau_k)$ in $D$ with either $\myvar{o_j} = \vx_i$ and $\myvar{p_k} = \vx_{i+1}$ or $\myvar{p_k} = \vx_i$ and $\myvar{o_j} = \vx_{i+1}$.
    \end{itemize}
    In the remainder of this proof, we show that for each pair of variables $\vx_i$ and $\vx_j$ in this sequence with $\type{\vx_i} = \type{\vx_j}$ that $\barmu(\vx_i) = \barmu(\vx_j)$. The desired $\barmu(\vx) = \barmu(\vy)$ then follows immediately as $\vx=\vx_1$ and $\vy=\vx_n$.
    Note that it suffices to show this property only for pairs of variables $\vx_i$ and $\vx_j$ for which no variable $\vx_k$ exists with $i<k<j$ and $\type{\vx_i} = \type{\vx_j} = \type{\vx_k}$. Indeed, if such an $\vx_k$ exists, we can recursively argue that $\barmu(\vx_i) = \barmu(\vx_k)$ and $\barmu(\vx_k) = \barmu(\vx_j)$.
        The argument is by induction on the number of variables between $\vx_i$ and $\vx_j$. 

    If $j = i+1$ \emph{(base case)}, then $(\ddagger)$ applies to $\vx_i$ and $\vx_j$.
    Indeed, if $(\dagger)$ would apply instead, then there would be a function name $f$ with $\domf(f) = \type{\vx_i} = \type{\vx_j} = \rangef(f)$, contradicting Condition~\emph{(\ref{lemma:mtbschemaprops:1})} of Lemma~\ref{lemma:mtbschemaprops}.
    By definition of $\barmu$, we have $\barmu(\vx_i) = \barmu(\vx_j)$.
    
    Next, let $i+1 < j$ \emph{(inductive case)}, and assume that $\barmu(\vx_k) = \barmu(\vx_\ell)$ for all $\vx_k$ and $\vx_\ell$ with $i < k \leq \ell < j$ and $\type{\vx_k} = \type{\vx_\ell}$ \emph{(induction hypothesis)}.
    From this sequence $\vx_i, \ldots \vx_j$, we derive a sequence of function names $F = f_1 \cdots f_{m-1}$, where each function name $f_i$ is based on an application of $(\dagger)$ on adjacent variables (notice that applications of $(\ddagger)$ do not result in a function name being added to $F$).
    By assumption on the types of variables $\vx_k$ with $i < k < j$, we have in particular $\type{\vx_{i+1}} \neq \type{\vx_i}$ and $\type{\vx_{j-1}} \neq \type{\vx_j}$.
    This implies that $(\dagger)$ is applicable for $\vx_{i}$ and $\vx_{i+1}$ (respectively $\vx_{j-1}$ and $\vx_{j}$).
    Furthermore, $\vx_{i}$ and $\vx_{i+1}$ appear in the same \shortptrans{}, say $\tau_i$ (respectively $\tau_j$ for $\vx_{j-1}$ and $\vx_{j}$), and $\vx_{i+1} = f_1(\vx_i) \in \Cset(\tau_i)$ (respectively $\vx_{j} = f_{m-1}(\vx_{j-1}) \in \Cset(\tau_j)$).
    By construction, $F$ then describes a path in $\sg{\schRel, \schFunc}$ visiting the nodes $R_1, R_2, \ldots, R_{m-1}, R_m$ with $\type{\vx_i} = R_1$, $\type{\vx_{i+1}} = R_2$, $\type{\vx_{j-1}} = R_{m-1}$ and $\type{\vx_{j}} = R_{m}$.
    Since this path satisfies Condition~\ref{lemma:mtbschemaprops:2} in Lemma~\ref{lemma:mtbschemaprops}, it follows that $\type{\vx_{i+1}} = \type{\vx_{j-1}}$ and $(f_1, f_{m-1})$ is a pair in the pairwise partitioning of $\schFunc$ witnessing $\workload \in \mtbtemplates$.
    By definition of $\mtbtemplates$, $\vx_{i+1} = f_1(\vx_i) \in \Cset(\tau_i)$ then implies $\vx_{i} = f_{m-1}(\vx_{i+1}) \in \Cset(\tau_i)$.
    According to the induction hypothesis, $\barmu(\vx_{i+1}) = \barmu(\vx_{j-1})$. Since $\barmu$ is admissible for $\db$, we conclude that $\barmu(\vx_i) = f_{m-1}^\db(\barmu(\vx_{i+1})) = f_{m-1}^\db(\barmu(\vx_{j-1})) = \barmu(\vx_j)$.
\end{proof}

It follows from Lemma~\ref{lemma:mtbconnected} that, if we group connected variables, then the same tuple is assigned to all variables of the same type in this group.
We encode this choice of tuples for variables through (total) functions $c:\schRel\to\objects$ that we call \emph{type mappings} and which map a relation onto a particular tuple of that relation's type. For instance, in SmallBank, a type mapping $c$ is determined by an \emph{Account} tuple $\mya$, a \emph{Savings} tuple $\mys$, and a \emph{Checking} tuple $\myc$. 
The following Lemma makes explicit how $\bar\mu$ can be decomposed into type mappings such that connected variables use the same type mapping and disequalities enforce the use of different type mappings.
\begin{restatable}{lem}{lemmamtbscheduletypes}\label{lemma:mtbschedule:types}
    For a multiversion split schedule $\schedule$ based on a \cqsequence{} $C$ over a set of transactions $\transset$ consistent with 
    a $\workload\in \mtbtemplates$ and a database $\db$, let $\barmu$ be the variable mapping for a \pcqsequence{} $D$ over $\workload$ with $\barmu(D) = C$.
    Then, a set $\typemapset$ of type mappings over disjoint ranges and a function $\typemapfunc: \variables \to \typemapset$
    exist with:
    \begin{itemize}
        \item $\barmu(\vx) = c(\type{\vx})$ for every variable $\vx$, with $c = \typemapfunc(\vx)$;
        \item $\typemapfunc(\vx) = \typemapfunc(\vy)$ whenever 
        \conn{\vx}{\vy}{D}; and,
    \item $\typemapfunc(\vx) \neq \typemapfunc(\vy)$
    for every constraint $\vx \neq \vy$ occurring in a template $\tau \in \transcopies{D}$. 
    \end{itemize}
\end{restatable}

\begin{proof}
    \newcommand{\colorf}{\lambda}
    To aid the construction of $\typemapset$ and $\typemapfunc$, we first define a coloring function $\colorf$
    that assigns a color to each tuple occurring in the schedule $\schedule$
    such that the following holds:
    for every pair of tuples $\x$ and $\y$ occurring in $\schedule$:
    \begin{itemize}
        \item  connected tuples are mapped to the same color: if $\barmu(\vx) = \x$, $\barmu(\vy) = \y$ and $\conn{\vx}{\vy}{D}$ for some variables $\vx, \vy$ occurring in $\transcopies{D}$, then $\colorf(\x) = \colorf(\y)$; and
        \item different tuples of the same type are mapped to different colors: if $\type{\x} = \type{\y}$ and $\x \neq \y$, then $\colorf(\x) \neq \colorf(\y)$.
    \end{itemize}
    Note that we can always construct such a function $\colorf$ 
    as by Lemma~\ref{lemma:mtbconnected}, it cannot be the case that $\type{\x} = \type{\y}$, $\x \neq \y$ and there is a pair of variables $\vx, \vy$ with $\barmu(\vx) = \x$, $\barmu(\vy) = \y$, and $\conn{\vx}{\vy}{D}$ . 
    
    For $\alpha\in\range{\colorf}$, define the type mapping $c_{\alpha}$ as follows: for every type $R \in \schRel$:
        $$
        c_{\alpha}(R) =
        \begin{cases}
            \x & \text{if $\colorf(\x) = \alpha$ and $\type{\x} = R$,}\\
            \y_{c,R} & \text{otherwise,}
        \end{cases}
        $$
        where $\y_{c,R}$ is an arbitrary tuple of type $R$ not occurring in $\schedule$ or any other type mapping $c_{\beta}$ for $\beta\in\range{\colorf}$. Define $\typemapset=\{c_\alpha\mid \alpha\in\range{\colorf}\}$.
        By construction, every type mapping in $\typemapset$ is well defined and all type mappings are over disjoint ranges. 
            Furthermore, $c_\alpha \neq c_\beta$ whenever $\alpha\neq\beta$.
    
        We now construct $\typemapfunc$ as follows:
        $\typemapfunc(\vx) = c_{\alpha}$ with $\alpha = \colorf(\barmu(\vx))$
        for every variable $\vx$ occurring in $\transcopies{D}$. It remains to argue that $\typemapfunc$ indeed satisfies all properties stated in Lemma~\ref{lemma:mtbschedule:types}. By construction of $\typemapset$ and $\typemapfunc$, we have $\barmu(\vx) = c(\type{\vx})$ for every variable $\vx$, with $c  = \typemapfunc(\vx)$.
        Towards the second property, notice that $\conn{\vx}{\vy}{D}$ implies $\typemapfunc(\vx) = c_{\colorf(\barmu(\vx))} = c_{\colorf(\barmu(\vy))} = \typemapfunc(\vy)$ by definition of $\colorf$ and $\typemapfunc$.
        For the last property, assume $\vx \neq \vy$ occurs in a \shortptrans{} $\tau \in \transcopies{D}$ and $\type{\vx} = \type{\vy}$. Since $\barmu$ is admissible for database $\db$, $\barmu(\vx) \neq \barmu(\vy)$.
        Then, by definition of $\colorf$ and $\typemapfunc$, we have $\typemapfunc(\vx) = c_{\colorf(\barmu(\vx))} \neq c_{\colorf(\barmu(\vy))} = \typemapfunc(\vy)$.
\end{proof}

From $D = (\tau_1, o_1, p_2, \tau_2), \ldots, (\tau_m, o_m, p_1, \tau_1)$ and $\typemapfunc$ as in Lemma~\ref{lemma:mtbschedule:types}
we can derive a sequence of quintuples $E = (\tau_1, o_1, c_{o_1}, p_1, c_{p_1}), \ldots, (\tau_m, o_m, c_{o_m}, p_m, c_{p_m})$ such that 
$c_{o_i} = \typemapfunc(\myvar{o_i})$ and $c_{p_i} = \typemapfunc(\myvar{p_i})$ for $i \in [1,m]$.
Intuitively, this sequence of quintuples can be used to reconstruct the original multiversion split schedule $\schedule$.
The next Lemma shows that we can decide robustness against \mvrc over a set of \ptranss{} admitting multi-tree bijectivity by searching for a specific sequence of quintuples over at most four type mappings.

\newcounter{resumecounter}
\begin{restatable}{lem}{lemmamtbchar}\label{lemma:mtbchar}
    Let $\workload\in \mtbtemplates$
     and let $\typemapset = \{c_1, c_2, c_3, c_4\}$ be a set consisting of four type mappings with disjoint ranges. Then, $\workload$ is not robust against \mvrc iff there is a sequence of quintuples $E = (\tau_1, o_1, c_{o_1}, p_1, c_{p_1}), \ldots, (\tau_m, o_m, c_{o_m}, p_m, c_{p_m})$ with $m \geq 2$ such that for each quintuple $(\tau_i, o_i, c_{o_i}, p_i, c_{p_i})$ in $E$:
    \begin{enumerate}
        \item \label{e1} $o_i$ and $p_i$ are operations in $\tau_i$, and $c_{o_i}, c_{p_i} \in \typemapset$;
        \item \label{e2} $\notconn{\vx_i}{\vy_i}{\tau_i}$ for each constraint $\vx_i \neq \vy_i$ in $\tau_i$;
        \item \label{e3} $c_{o_i} = c_{p_i}$ if 
        $\conn{o_i}{p_i}{\tau_i}$;
        \item \label{e4} $c_{o_i} \neq c_{p_i}$ if there is a constraint $\vx_i \neq \vy_i$ in $\tau_i$ with
        $\conn{\vx_i}{\myvar{o_i}}{\tau_i}$ and $\conn{\vy_i}{\myvar{p_i}}{\tau_i}$;
        \item \label{e6} if $i \neq 1$ and $c_{q_i} = c_{q_1}$ for some $q_i \in \{o_i, p_i\}$ and $q_1 \in \{o_1, p_1\}$,
        then there is no operation $o_i'$ in $\tau_i$ potentially ww-conflicting with an operation $o_1'$ in $\prefix{\tau_1}{o_1}$ with $\conn{\myvar{o_i'}}{\myvar{q_i}}{\tau_i}$ and $\conn{\myvar{o_1'}}{\myvar{q_1}}{\tau_1}$.
        \setcounter{resumecounter}{\value{enumi}}
    \end{enumerate}
    Furthermore, for each pair of adjacent quintuples $(\tau_i, o_i, c_{o_i}, p_i, c_{p_i})$ and $(\tau_j, o_j, c_{o_j}, p_j, c_{p_j})$ in $E$ with $j = i + 1$, or $i = m$ and $j = 1$:
    \begin{enumerate}
        \setcounter{enumi}{\value{resumecounter}}
        \item \label{e8} $o_i$ is potentially conflicting with $p_j$ and $c_{o_i} = c_{p_j}$;
        \item \label{e9} if $i = 1$ and $j = 2$, then $o_1$ is potentially rw-conflicting with $p_2$; and
        \item \label{e10} if $i = m$ and $j = 1$, then $o_1 <_{\tau_1} p_1$ or $o_m$ is potentially rw-conflicting with $p_1$.
    \end{enumerate}
\end{restatable}

\noindent The items have the following meaning: (\ref{e2}) $\tau_i$ is satisfiable; (\ref{e3}) connected operations are assigned the same type mapping; (\ref{e4}) 
variables connected through an inequality are assigned a different type mapping;
(\ref{e6}) $\typemapfunc$ only assigns the same type mapping
to $o_1$ or $p_1$ in $\tau_1$ and $o_i$ or $p_i$ in $\tau_i$ if it does not introduce a dirty write in the resulting multiversion split schedule (cf. Condition~(\ref{c:1}) in Definition~\ref{def:mvsplitschedule});
(\ref{e8}) each pair of variables in operations used for conflicts are assigned the same type mapping;
(\ref{e9}, \ref{e10}) the operations used for conflicts between $\tau_1$, $\tau_2$ and $\tau_m$ are restricted to satisfy respectively Condition~(\ref{c:3}) and (\ref{c:2}) in Definition~\ref{def:mvsplitschedule} in the resulting multiversion split schedule.

We first present an additional Lemma derived from Lemma~\ref{lemma:mtbconnected} that will be used in the proof of Lemma~\ref{lemma:mtbchar}.

\begin{lem}\label{lemma:connectedclosure}
    Let $D$ be a \pcqsequence. If $\conn{\vx}{\vy}{D}$ and $\conn{\vy}{\vz}{D}$ then $\conn{\vx}{\vz}{D}$ for every triple of variables $\vx, \vy, \vz$ occurring in $\transcopies{D}$.
\end{lem}

\begin{proof}
    According to Lemma~\ref{lemma:mtbconnected}, $\conn{\vx}{\vy}{D}$ and $\conn{\vy}{\vz}{D}$ imply respectively $\vx \varimplies{}{D} \vy$ and $\vy \varimplies{}{D} \vz$.
    By definition, $\vx \varimplies{}{D} \vz$ and hence $\conn{\vx}{\vz}{D}$.
\end{proof}

\begin{proof}[Proof of Lemma~\ref{lemma:mtbchar}]
    \emph{(if)}  Let $D = (\tau_1, o_1, p_2, \tau_2), \ldots, (\tau_{m}, o_m, p_1, \tau_1)$ be the \pcqsequence{} derived from $E$. Notice in particular that $D$ is indeed a \pcqsequence{} by (\ref{e1}) and (\ref{e8}).
    We construct a variable mapping $\barmu$ for $D$ admissible for a database $\db$ such that the \cqsequence{} $C = \barmu(D)$ satisfies the conditions in Definition~\ref{def:mvsplitschedule}, thereby proving that $\workload$ is not robust against \mvrc.
    
    Let $\typemapfunc: \variables \to \typemapset$ be the (partial) function assigning a type mapping in $\typemapset$ to each variable occurring in an operation in $E$:
    $$
    \typemapfunc(\vx) = 
    \begin{cases}
        c_{o_i} & \text{if $\myvar{o_i} = \vx$ for some $(\tau_i, o_i, c_{o_i}, p_i, c_{p_i}) \in E$,}\\
        c_{p_i} & \text{if $\myvar{p_i} = \vx$ for some $(\tau_i, o_i, c_{o_i}, p_i, c_{p_i}) \in E$.}
    \end{cases}
    $$
    This function $\typemapfunc$ is well defined: if there is a $(\tau_i, o_i, c_{o_i}, p_i, c_{p_i}) \in E$ with $\myvar{o_i} = \myvar{p_i} = \vx$, then $\conn{o_i}{p_i}{\tau_i}$ and hence $c_{o_i} = c_{p_i}$ by (\ref{e3}).
    Recall that we assume that templates in $E$ are variable-disjoint. 
    We argue that $\typemapfunc(\vx) = \typemapfunc(\vy)$ if $\conn{\vx}{\vy}{D}$ for each pair of variables $\vx$ and $\vy$ for which $\typemapfunc$ is defined.
    From Lemma~\ref{lemma:mtbconnected}, it follows that $\vx \varimplies{}{D} \vy$ whenever $\conn{\vx}{\vy}{D}$. Let $\tau_i$ and $\tau_j$ be the \shortptrans{} in which respectively $\vx$ and $\vy$ occur. The argument is now by induction on the definition of $\vx \varimplies{}{D} \vy$:
    \begin{itemize}
        \item If $i = j$ and $\vx \varimplies{}{\tau_i} \vy$, then $\typemapfunc(\vx) = \typemapfunc(\vy)$ is immediate by (\ref{e3});
        \item If $(\tau_i, o_i, p_j, \tau_j) \in D$ with $\myvar{o_i} = \vx$ and $\myvar{p_j} = \vy$ (respectively $(\tau_j, o_j, p_i, \tau_i) \in D$ with $\myvar{o_j} = \vy$ and $\myvar{p_i} = \vx$), then $\typemapfunc(\vx) = \typemapfunc(\vy)$ is immediate by (\ref{e8});
        \item Otherwise, if $\vx \varimplies{}{D} \vz$ and $\vz \varimplies{}{D} \vy$ for some variable $\vz$, then by induction $\typemapfunc(\vx) = \typemapfunc(\vz) = \typemapfunc(\vy)$.
    \end{itemize}

    \noindent By Lemma~\ref{lemma:connectedclosure}, $\conn{}{}{D}$ is an equivalence relation. For $\vx$ occurring in $\transcopies{D}$, denote by $[\vx]$ the equivalence class of $\vx$. Let $\typemapset'$ be obtained by extending $\typemapset$ with a type mapping $c_{[\vx]}$ for each equivalence class where no variable $\vy\in[\vx]$ is defined in $\typemapfunc[\typemapset]$. Furthermore, each of the $c_{[\vx]}$ are picked such that all type mappings in $\typemapset'$ have disjoint ranges.
    
    Next, we extend $\typemapfunc$ to a function $\typemapfunc[\typemapset']: \variables \to \typemapset'$ assigning a type mapping to each variable $\vx$ occurring in $\transcopies{D}$ 
    as follows:
    $$
    \typemapfunc[\typemapset'](\vx) =
    \begin{cases}
        \typemapfunc(\vx) & \text{if $\typemapfunc$ is defined for $\vx$,}\\
        \typemapfunc[\typemapset](\vy) & \text{if $\typemapfunc$ is defined for $\vy$ but not for $\vx$ and $\conn{\vx}{\vy}{D}$,}\\
        c_{[\vx]} & \text{otherwise.}
    \end{cases}
    $$
    Notice, furthermore, that in the second case $\vx$ might be connected in $D$ to multiple variables for which $\typemapfunc$ is defined, say $\vy_1$ and $\vy_2$.
    Then, by Lemma~\ref{lemma:connectedclosure}, $\conn{\vy_1}{\vy_2}{D}$ and hence $\typemapfunc(\vy_1) = \typemapfunc(\vy_2)$. We therefore conclude that $\typemapfunc[\typemapset'](\vx)$ is well defined.
    We argue that $\typemapfunc[\typemapset'](\vx) = \typemapfunc[\typemapset'](\vy)$ if $\conn{\vx}{\vy}{D}$ for each pair of variables $\vx$ and $\vy$.
    If $\typemapfunc$ is defined for both $\vx$ and $\vy$, then the result is immediate by $\typemapfunc(\vx) = \typemapfunc(\vy)$.
    If $\typemapfunc$ is defined for one of these two variables, say $\vx$, then $\typemapfunc[\typemapset'](\vx) = \typemapfunc[\typemapset](\vx) = \typemapfunc[\typemapset'](\vy)$ by construction of $\typemapfunc[\typemapset']$.
    If $\typemapfunc$ is not defined for both $\vx$ and $\vy$, then either there exists a variable $\vz$ for which $\typemapfunc$ is defined and $\vz$ is connected in $D$ to one of these two variables, say $\vx$, or no such variable $\vz$ exists.
    In the former case, $\conn{\vz}{\vy}{D}$ follows from Lemma~\ref{lemma:connectedclosure}, implying $\typemapfunc[\typemapset'](\vx) = \typemapfunc[\typemapset](\vz) = \typemapfunc[\typemapset'](\vy)$.
    In the latter case, $\typemapfunc[\typemapset'](\vx) = c_{[\vx]}=c_{[\vy]}=\typemapfunc[\typemapset'](\vy)$ by construction of $\typemapfunc[\typemapset']$.

    We now define the variable mapping $\barmu$ from $\typemapfunc[\typemapset']$ as
    $\barmu(\vx) = c(\type{\vx})$ for each variable $\vx$, where $c = \typemapfunc[\typemapset](\vx)$.
    Next, we construct the database $\db$. For each \shortptrans{} $\tau_i$ and corresponding variable mapping $\mu_i$ in $\barmu$, we add all tuples in $\mu_i(\tau_i)$ to the database $\db$.
    Furthermore, for each constraint $\vx = f(\vy)$ in $\Cset(\tau_i)$, we have $f^\db(\mu_i(\vx)) = \mu_i(\vy)$ in $\db$.
    This is well defined for each function $f^\db$.
    Towards a contradiction, assume we have $f^\db(\barmu(\vx_i)) = \barmu(\vy_i)$ witnessed by a \shortptrans{} $\tau_i$ and $f^\db(\barmu(\vx_j)) = \barmu(\vy_j)$ witnessed by a \shortptrans{} $\tau_j$, where $\barmu(\vx_i) = \barmu(\vx_j)$, but $\barmu(\vy_i) \neq \barmu(\vy_j)$.
    Since $\conn{\vx_i}{\vy_i}{D}$, $\conn{\vx_j}{\vy_j}{D}$ and $\barmu(\vx_i) = \barmu(\vx_j)$, we have $\typemapfunc[\typemapset'](\vy_i) = \typemapfunc[\typemapset'](\vx_i) = \typemapfunc[\typemapset'](\vx_j) = \typemapfunc[\typemapset'](\vy_j)$.
    Then, $\barmu(\vy_i) = \barmu(\vy_j)$, leading to a contradiction.
    In order to argue that $\barmu$ is indeed admissible for $\db$, it remains to show that for each constraint $\vx \neq \vy$ in a template $\tau_i$, we have $\barmu(\vx) \neq \barmu(\vy)$. Again towards a contradiction, assume $\barmu(\vx) = \barmu(\vy)$, and let $(\tau_i, o_i, c_{o_i}, p_i, c_{p_i})$ be the corresponding quintuple in $E$.
    By definition of $\barmu$, we have $\typemapfunc[\typemapset'](\vx) = \typemapfunc[\typemapset'](\vy)$.
    It follows from $\typemapfunc[\typemapset']$ that either $\conn{\vx}{\vy}{\tau_i}$;
    or $\conn{\vx}{\myvar{o_i}}{\tau_i}$ (respectively $\conn{\vy}{\myvar{o_i}}{\tau_i}$), $\conn{\vy}{\myvar{p_i}}{\tau_i}$ (respectively $\conn{\vx}{\myvar{p_i}}{\tau_i}$) and $c_{o_i} = c_{p_i}$. However, the former is contradicted by (\ref{e2}) and the latter by (\ref{e4}).
    We therefore conclude that $\barmu$ is admissible for $\db$, as it satisfies all constraints.    
    
    It remains to argue that the \cqsequence{} $C = \barmu(D)$ satisfies all conditions stated in Definition~\ref{def:mvsplitschedule}. The second and third condition are immediate by respectively (\ref{e10}) and (\ref{e9}).
    Towards a contradiction, assume the first condition holds. Then, there is an operation $b_1'$ in $\prefix{\barmu(\tau_1)}{\barmu{o_1}}$ ww-conflicting with an operation $b_i'$ in a transaction $\barmu{\tau_i}$.
    Let $b_1' = \barmu(o_1')$ with $\myvar{o_1'} = \vx_1$ and $b_i' = \barmu(o_i')$ with $\myvar{o_i'} = \vx$, and let $(\tau_1, o_1, c_{o_1}, p_1, c_{p_1})$ and $(\tau_i, o_i, c_{o_i}, p_i, c_{p_i})$ be the corresponding quintuples in $E$.
    Note that $o_1'$ is potentially ww-conflicting with $o_i'$ and $\barmu(\vx_1) = \barmu(\vx_i)$. Then, $\contextfunc[\contextset'](\vx_1) = \contextfunc[\contextset'](\vx_i)$.
    By construction of $\contextfunc[\contextset']$, this can only hold if $\conn{\vx_1}{\myvar{q_1}}{\tau_1}$, $\conn{\vx_i}{\myvar{q_i}}{\tau_i}$ and $c_{q_1} = c_{q_i}$ for some $q_1 \in {o_1, p_1}$ and $q_i \in {o_i, p_i}$, thereby contradicting (\ref{e6}).
    
    \medskip
    \noindent
    \emph{(only if)} If $\workload$ is not robust against \mvrc, then there exists a multiversion split schedule $\schedule$ based on a \cqsequence{} $C$ over a set of transactions $\transset$ consistent with $\workload$ and a database $\db$.
    Let $\barmu$ be the variable mapping for a \pcqsequence{} $D = (\tau_1, o_1, p_1, \tau_2), \ldots, (\tau_m, o_m, p_1, \tau_1)$ over $\workload$ with $\barmu(C) = D$, and let $\typemapset$ and $\typemapfunc$ be as in Lemma~\ref{lemma:mtbschedule:types}.
    
    From this sequence $D$ and function $\typemapfunc$, we derive the sequence of quintuples $E = (\tau_1, o_1, \typemapfunc(\myvar{o_1}), p_1, \typemapfunc(\myvar{p_1})), \ldots, (\tau_m, o_m, \typemapfunc(\myvar{o_m}), p_m, \typemapfunc(\myvar{p_m}))$.
    Let $\typemapfunc' = \{c_1, c_2, c_3, c_4\}$ be a set consisting of four type mappings with disjoint ranges. We adapt each quintuple in $E$ in order, thereby creating a sequence $E'$ satisfying the properties stated in Lemma~\ref{lemma:mtbchar}. First, we add $(\tau_1, o_1, c_1, p_1, c_k)$ to $E'$, where $c_k = c_1$ if $\typemapfunc(\myvar{o_1})= \typemapfunc(\myvar{p_1})$, and $c_k = c_2$ otherwise.
    For each of the remaining quintuples in $E$, let $(\tau_{i-1}, o_{i-1}, c_{o_{i-1}}, p_{i-1}, c_{p_{i-1}})$ be the quintuple previously added to $E'$. We then add $(\tau_i, o_i, c_{o_i},\allowbreak p_i, c_{p_i})$ to $E'$ where $c_{o_i} = c_{p_{i-1}}$ and
    $$
    c_{p_i}=
    \begin{cases}
        c_{o_i} & \text{if $\typemapfunc(\myvar{o_i})= \typemapfunc(\myvar{p_i})$,}\\
        c_1 & \text{if $\typemapfunc(\myvar{o_i})= \typemapfunc(\myvar{o_1})$,}\\
        c_2 & \text{if $\typemapfunc(\myvar{o_i})= \typemapfunc(\myvar{p_1})$ and $\typemapfunc(\myvar{o_1}) \neq \typemapfunc(\myvar{p_1})$,}\\
        c_3 & \text{if $\typemapfunc(\myvar{o_i})\neq \typemapfunc(\myvar{p_i})$ and $c_{p_i} \neq c_3$,}\\
        c_4 & \text{otherwise.}
    \end{cases}
    $$
    By construction, for every quintuple $(\tau_i, o_i, c_{o_i}, p_i, c_{p_i})$ in $E'$ we now have
    \begin{itemize}
        \item $c_{o_i}$ = $c_{p_i}$ iff $\typemapfunc(\myvar{o_i}) = \typemapfunc(\myvar{p_i})$; and
        \item $c_{q_i}$ = $c_{q_1}$ iff $\typemapfunc(\myvar{q_i}) = \typemapfunc(\myvar{q_1})$ for every $q_i \in \{o_i, p_i\}$ and $q_1 \in \{o_1, p_1\}$.
    \end{itemize}
    
    It remains to argue that $E'$ indeed satisfies all required properties.
    (\ref{e1}) is trivial by construction.
    (\ref{e2}) If $\vx_i \neq \vy_i$ is a constraint in $\tau_i$, then $\typemapfunc(\vx) \neq \typemapfunc(\vy)$ and $\notconn{\vx}{\vy}{D}$ according to Lemma~\ref{lemma:mtbschedule:types}.
    (\ref{e3}) If $\conn{o_i}{p_i}{\tau_i}$, then $\typemapfunc(\myvar{o_i}) = \typemapfunc(\myvar{p_i})$ by Lemma~\ref{lemma:mtbschedule:types}, and hence $c_{o_i} = c_{p_i}$.
    (\ref{e4}) Assume there is a constraint $\vx_i \neq \vy_i$ in a \shortptrans{} $\tau_i$ with $\conn{\vx_i}{\myvar{o_i}}{\tau_i}$ and $\conn{\vy_i}{\myvar{p_i}}{\tau_i}$.
    By Lemma~\ref{lemma:mtbschedule:types}, $\typemapfunc(\vx_i) = \typemapfunc(\myvar{o_i}) \neq \typemapfunc(\myvar{p_i}) = \typemapfunc(\vy_i)$, and therefore $c_{o_i} \neq c_{p_i}$.
     (\ref{e6}) Let $c_{q_i} = c_{q_1}$ for some $q_i \in \{o_i, p_i\}$ and $q_1 \in \{o_1, p_1\}$, with $i \neq 1$.
     Assume towards a contradiction that there is an operation $o_i'$ in $\tau_i$ potentially ww-conflicting with an operation $o_1'$ in $\prefix{\tau_1}{o_1}$ with $\conn{\myvar{o_i'}}{\myvar{q_i}}{\tau_i}$ and $\conn{\myvar{o_1'}}{\myvar{q_1}}{\tau_1}$.
     But then $\typemapfunc(\myvar{o_i'}) = \typemapfunc(\myvar{q_i}) = \typemapfunc(\myvar{q_1}) = \typemapfunc(\myvar{o_1'})$, implying that $\barmu(o_i')$ is ww-conflicting with $\barmu(o_1')$, contradicting the properties of $C$ stated in Definition~\ref{def:mvsplitschedule}.
     (\ref{e8}) is again trivial by construction.
     (\ref{e9}) By Definition~\ref{def:mvsplitschedule}, $\barmu(o_1)$ is rw-conflicting with $\barmu(p_2)$ in $C$. Therefore, $o_1$ is potentially rw-conflicting with $p_2$.
     (\ref{e10}) By Definition~\ref{def:mvsplitschedule}, $\barmu(o_1) <_{\barmu(\tau_1)} \barmu(p_1)$ or $\barmu(o_m)$ is rw-conflicting with $\barmu(p_1)$ in $C$. As a result, $o_1 <_{\tau_1} p_1$ or $o_m$ is rw-conflicting with $p_1$.
\end{proof}

The characterization for \RobustnessTempTwo{\mtbtemplates}{\mvrc} in Lemma~\ref{lemma:mtbchar} implies an \NLOGSPACE algorithm guessing the counterexample sequence $E$, thereby proving Theorem~\ref{theo:1dettemplates}. 
Indeed, the algorithm guesses the sequence of quintuples $E$, verifying all conditions for each newly guessed quintuple while only requiring logarithmic space.
Notice in particular that we only need to keep track of two other quintuples when verifying all conditions for the newly guessed quintuple, namely the first quintuple over $\tau_1$ and the quintuple immediately preceding the newly guessed one.
As usual, we can think of the encoding of \shortptranss{} and operations mentioned in each quintuple as pointers referring to the corresponding \shortptranss{} and operations on the input tape.
Furthermore, we do not encode the four type mappings explicitly as such a representation of a mapping might require polynomial space.
Since we are only interested in (dis)equality between type mappings, an encoding where these four type mappings are represented by four arbitrary strings of constant size suffices.

\begin{proof}[Proof of Theorem~\ref{theo:mtbtemplates}]
The \NLOGSPACE algorithm goes as follows: We start by guessing three initial quintuples, representing respectively the first, second and last quintuple of the possible sequence of quintuples as in Lemma~\ref{lemma:mtbchar}. Consistent with previously used notation, we refer to these quintuples by $E_1$, $E_2$, and $E_m$, with $E_i = (\tau_i, o_i, c_{o_i}, p_i, c_{p_i})$. Note that the indices we use here are not part of the algorithm. They are only used to distinguish between the different considered quintuples in the proof argument. 

We store all three quintuples using a logarithmic amount of space, by storing pointers to the respective \ptranss{} in $\workload$, the positions of operations in the respective \ptranss{}, and the number $1,2,3$ or $4$ for the type mappings.

At this point, we verify that Condition~(7) and (8) are true, that Conditions~(1-5) are true for all chosen \ptranss{} and operations, and that Condition~(6) is true for $\tau_1$ and $\tau_2$, and $\tau_2$ and $\tau_m$. We reject the guessed quintuples if any of the conditions is false.

If all previous checks are true, we proceed by inserting another step. Let $i=2$. We guess a new quintuple $E_{i+1}$ and verify that Condition (5) is true for $\tau_{i}$ and $\tau_{i+1}$ and that Conditions (1-6) are true for $\tau_{i+1}$ and reject the entire construction if one of these conditions failed. Notice that all Conditions, including Condition~(5) can be checked easily, particularly because quintuple $E_1$ is stored. To proceed, we discard quintuple $E_i$ and store $E_{i+1}$ instead, thus without increasing the  amount of space we use. 

If $\tau_{i+1}$ and $\tau_m$ (from quintuple $E_m$) have Condition~(6), the algorithm emits an accept. Indeed, then the sequence $E_1, \ldots, E_i, E_{i+1}, E_m$ of guessed quintuples has all the properties of Lemma~\ref{lemma:mtbchar}. Otherwise, the algorithm proceeds with another insertion step, for $i=i+1$.
\end{proof}

\section{Robustness for \shortPtranss{} over Acyclic Schemas}
\label{sec:acyclic}

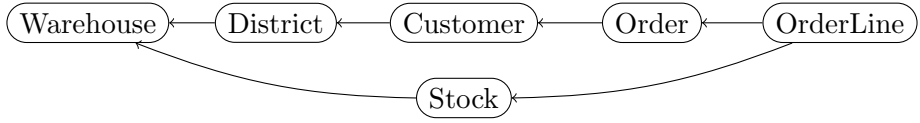
\begin{figure}[t]
    \centering
    \begin{tikzpicture}[scale=1]
    \node[draw, rounded rectangle,align=center] (W) at (-5,0) {\Warehouse};
    \node[draw, rounded rectangle,align=center] (D) at (-2.5,0) {\District};
    \node[draw, rounded rectangle,align=center] (C) at (0,0) {\Customer};
    \node[draw, rounded rectangle,align=center] (O) at (2.5,0) {\Order};
    \node[draw, rounded rectangle,align=center] (OL) at (5,0) {\OrderLine};
    \node[draw, rounded rectangle,align=center] (S) at (0,-1) {\Stock};
    \draw[->] (D) to (W) ;
    \draw[->] (C) to (D) ;
    \draw[->] (O) to (C) ;
    \draw[->] (OL) to (O) ;
    \draw[->] (S.west) to [bend left=10] (W) ;
    \draw[->] (OL) to [bend left=10] (S.east) ;
    \end{tikzpicture}
    \caption{{Acyclic schema graph for the TPC-C benchmark.}}
    \label{fig:tpcc:sg}
\end{figure}

We denote by \acyctemplates{} the class of all sets of \ptranss{} over acyclic schemas.
{As a concrete example, the schema graph for the TPC-C benchmark is given in Figure~\ref{fig:tpcc:sg}. Since this schema graph does not contain any cycles, the TPC-C benchmark is situated within \acyctemplates{}. Notice in particular how this acyclic schema graph corresponds to the hierarchical structure of many-to-one relationships inherent to the schema for this benchmark. For example, every orderline belongs to exactly one order, and every order is related to exactly one customer, but the opposite is never true (i.e., a customer can be related to multiple orders, each of which can be related to multiple orderlines). In general, the results presented in this section can be applied to all workloads over schemas with such a hierarchical structure.}

\begin{thm}\label{theo:asgexpspace}
    \RobustnessTempTwo{\acyctemplates}{\mvrc} is decidable in \EXPSPACE.
\end{thm}

We first provide some intuition for the proof. For a given acyclic schema graph $\sgname$,
$R \asgimplies{F}{\sgname} S$ denotes the directed path from node $R$ to node $S$ in $\sgname$ with $F$ the sequence of edge labels on the path.
The next lemma relates implication between variables to paths in $\sgname$.

\begin{restatable}{lem}{lemmavarimpliestosgpath}\label{lemma:varimpliestosgpath}
    Let $D$ be a \pcqsequence{} over a set of \ptranss{} $\workload\in \acyctemplates$. For every pair of variables $\vx, \vy$ occurring in $\transcopies{D}$, if $\vx \varimplies{F}{D} \vy$, 
    then $\type{\vx} \asgimplies{F}{SG} \type{\vy}$, with $\sgname$ the corresponding schema graph.
\end{restatable}

\begin{proof}
    Let $\tau_i$ and $\tau_j$ be the \shortptranss{} in $D$ in which $\vx$ and $\vy$ occur, respectively.
    The proof is by induction on the definition of $\vx \varimplies{F}{D} \vy$.
    
    \emph{(Implication within the same \shortptrans)} If 
    $i=j$ and $\vx \varimplies{F}{\tau_i} \vy$, then either $F = \varepsilon$ and $\vx = \vy$, or there is a variable $\vz$ such that $\vy = f(\vz)$ is a constraint in $\Cset(\tau_i)$, $\vx \varimplies{F'}{\tau_i} \vz$ and $F = F' \cdot f$.
    In the former case, $\type{\vx} = \type{\vy}$, so $\type{\vx} \asgimplies{\varepsilon}{\sgname} \type{\vy}$ is immediate. In the latter case, it follows by induction that $\type{\vx} \asgimplies{F'}{\sgname} \type{\vz}$.
    Since $\dom{f} = \type{\vz}$ and $\range{f} = \type{\vy}$, it follows by definition that $\type{\vz} \asgimplies{f}{\sgname} \type{\vy}$ and furthermore $\type{\vx} \asgimplies{F}{\sgname} \type{\vz}$ holds.
    
    \emph{(Implication between \shortptranss)} If $F = \varepsilon$ and $(\tau_i, o_i, p_j, \tau_j)$ (respectively $(\tau_j, o_j, p_i, \tau_i)$) is a \pconflictquadruple{} in $D$, with $\myvar{o_i} = \vx$ and $\myvar{p_j} = \vy$ (respectively $\myvar{p_i} = \vx$ and $\myvar{o_j} = \vy$), then $\type{\vx} = \type{\vy}$ by definition of potentially conflicting operations. So, $\type{\vx} \asgimplies{\varepsilon}{\sgname} \type{\vy}$ is again immediate.
    
    \emph{(Inductive case)} If $\vx \varimplies{F_1}{D} \vz$ and $\vz \varimplies{F_2}{D} \vy$ for some variable $\vz$ with $F = F_1 \cdot F_2$,
    then $\type{\vx} \asgimplies{F_1}{\sgname} \type{\vz}$ and $\type{\vz} \asgimplies{F_2}{\sgname} \type{\vy}$ follow by induction. We conclude that $\type{\vx} \asgimplies{F}{\sgname} \type{\vy}$.
\end{proof}

Notice that an assignment of a tuple to a variable $\vx$ determines the tuples assigned to all variables $\vy$ with $\vx \varimplies{F}{D} \vy$ for some sequence of function names $F$. From Lemma~\ref{lemma:varimpliestosgpath} it follows that each such implied tuple is witnessed by a path in the corresponding schema graph $\sgname$. Therefore, the maximal number of different tuples implied by $\vx$ corresponds to the number of paths in $\sgname$ starting in $\type{\vx}$, which is finite when $\sgname$ is acyclic.
Because there can be multiple paths between nodes in the schema graph, it is no longer the case as in the previous section that variables of the same type connected in $D$ must be assigned the same value. 
So, instead of using type mappings, we introduce \emph{tuple-contexts}
to represent the sets of all tuples implied by the assignment of a given variable.
Formally, a \emph{tuple-context for a type} $R \in \schRel$ is a function from paths with source $R$ in  $\sg{\schRel, \schFunc}$ to tuples in $\objects$ of the appropriate type. That is, for each tuple-context $c$ for type $R$ and for each path $R \asgimplies{F}{\sgname} S$ in $\sgname$, $\type{c(R \asgimplies{F}{\sgname} S)} = S$.

Similar to Lemma~\ref{lemma:mtbschedule:types}, we show 
that we can represent a counterexample schedule based on $D$ 
by assigning a tuple-context to each variable in $\transcopies{D}$, taking special care when assigning contexts to variables connected in $D$ to make sure that they are properly related to each other. For this, 
we introduce a (partial) function $\contextfunc: \variables \to \contextset$ mapping (a subset of) variables in $\transcopies{D}$ to tuple-contexts in $\contextset$ (for $\contextset$ a set of tuple-contexts)and refer to it as a \emph{(partial) context assignment for $D$ over $\contextset$}. In a sequence of lemmas, we show that $\contextfunc$ can always be expanded into a total function and an approach based on enumeration of quintuples analogous to Lemma~\ref{lemma:mtbchar} suffices to decide robustness. A major difference with the previous section is that there is no longer a constant bound on the number of tuple-contexts that are needed and consistency between tuple-contexts in connected variables needs to be maintained.

We call node $S$ a \emph{descendant} of node $R$ and $R$ an \emph{ancestor} of $S$ in an acyclic schema graph $\sgname$. We write $R \asgimplies{\varepsilon}{\sgname}S$, with $\varepsilon$ denoting the empty labeling, for the case $R=S$. This means that a node is a descendant and ancestor of itself. When $F$ is not relevant, we simply write $R \asgimplies{}{\sgname} S$. 

Let $c_R$ and $c_{S}$ be two tuple-contexts for types $R$ and $S$, respectively, such that $S$ is a descendant of $R$ in $\sgname$, witnessed by the path $R \asgimplies{F}{\sgname} S$ in $\sgname$.
We then say that \emph{$c_{S}$ is a tuple-subcontext of $c_R$ witnessed by $F$} if $c_{S}(S \asgimplies{F'}{\sgname} S') = c_{R}(R \asgimplies{F \cdot F'}{\sgname} S')$ for every path $S \asgimplies{F'}{\sgname} S'$ in $\sgname$.
It should be noted that $R \asgimplies{F \cdot F'}{\sgname} S'$ is indeed a valid path in $\sgname$, as it concatenates the paths $R \asgimplies{F}{\sgname} S$ and $S \asgimplies{F'}{\sgname} S'$. For a given tuple-context $c$ for a type $R$ in the schema graph $\sgname$, we will often write $c(F)$ as a shorthand notation for $c(R \asgimplies{F}{\sgname} S)$.

Similar to Lemma~\ref{lemma:mtbschedule:types} for sets of \ptranss{} admitting multi-tree bijectivity,
Lemma~\ref{lemma:acycschedule:contexts} shows that we can represent a counterexample schedule based on a \pcqsequence{} $D$ over an acyclic schema by assigning a tuple-context to each variable in $\transcopies{D}$, taking special care when assigning contexts to variables connected in $D$ to make sure that they are properly related to each other.
For a set of tuple-contexts $\contextset$, we refer to a (partial) function $\contextfunc: \variables \to \contextset$ mapping (a subset of) variables in $\transcopies{D}$ to tuple-contexts in $\contextset$ as a \emph{(partial) context assignment for $D$ over $\contextset$}. We furthermore say that $\contextfunc$ is a \emph{total context assignment for $D$ over $\contextset$} if $\contextfunc$ is defined for every variable in $\transcopies{D}$.

{Two variables $\vx$ and $\vy$ occurring in $\transcopies{D}$ are \emph{equivalent in $D$}, denoted $\vx \varequiv{D} \vy$ if
\begin{itemize}
    \item $\vx = \vy$;
    \item there exists a pair of variables $\vz$ and $\vw$ and a sequence of function names $F$ with $\vz \varequiv{D} \vw$, $\vz \varimplies{F}{D} \vx$ and $\vw \varimplies{F}{D} \vy$; or
    \item there exists a variable $\vz$ with $\vx \varequiv{D} \vz$ and $\vz \varequiv{D} \vy$.
\end{itemize}
Similarly, two variables $\vx$ and $\vy$ occurring in a \ptrans{} $\tau$ are \emph{equivalent in $\tau$}, denoted $\vx \varequiv{\tau} \vy$ if
\begin{itemize}
    \item $\vx = \vy$;
    \item there exists a pair of variables $\vz$ and $\vw$ in $\tau$ and a sequence of function names $F$ with $\vz \varequiv{\tau} \vw$, $\vz \varimplies{F}{\tau} \vx$ and $\vw \varimplies{F}{\tau} \vy$; or
    \item there exists a variable $\vz$ with $\vx \varequiv{\tau} \vz$ and $\vy \varequiv{\tau} \vz$.
\end{itemize}
}

\noindent Intuitively, every variable mapping admissible for a given database will assign the same tuple to equivalent variables (see Lemma~\ref{lemma:barmuvarequal}).
Due to these equivalent variables, the assignment of a tuple to a variable $\vx$ for a given database might imply the tuple assigned to a variable $\vy$, even if $\vx \varimplies{}{D} \vy$ does not hold.
We capture this observation by introducing \emph{variable determination}, which is stronger than the previously defined variable implication.
Formally, a variable $\vx$ \emph{determines} a variable $\vy$ in $D$ witnessed by a sequence of function names $F$, denoted $\vx \vardet{F}{D} \vy$ if:
\begin{itemize}
    \item $\vx \varimplies{F}{D} \vy$;
    \item $F = \varepsilon$ and $\vx \varequiv{D} \vy$; or
    \item there exists a variable $\vz$ with $\vx \vardet{F_1}{D} \vz$, $\vz \vardet{F_2}{D} \vy$ and $F = F_1 \cdot F_2$.
\end{itemize}
For two variables $\vx$ and $\vy$ in a \shortptrans{} $\tau \in \transcopies{D}$ we furthermore say that $\vx$ \emph{determines} $\vy$ in $\tau$ witnessed by a sequence of function names $F$, denoted $\vx \vardet{F}{\tau} \vy$ if:
\begin{itemize}
    \item $\vx \varimplies{F}{\tau} \vy$;
    \item $F = \varepsilon$ and $\vx \varequiv{\tau} \vy$; or
    \item there exists a variable $\vz$ with $\vx \vardet{F_1}{\tau} \vz$, $\vz \vardet{F_2}{\tau} \vy$ and $F = F_1 \cdot F_2$.
\end{itemize}

\begin{lem}\label{lemma:equalbarmu}
    For a multiversion split schedule $\schedule$ based on a \cqsequence{} $C$ over a set of transactions $\transset$ consistent with a set of \ptranss{} $\workload$ and a database $\db$, let $\barmu$ be the variable mapping for a \pcqsequence{} $D$ over $\workload$ with $\barmu(D) = C$.
    Then, for every combination of variables $\vw, \vx, \vy, \vz$ occurring in $\transcopies{D}$, if $\vz \varimplies{F}{D} \vx$, $\vw \varimplies{F}{D} \vy$ and $\barmu(\vz) = \barmu(\vw)$, then $\barmu(\vx) = \barmu(\vy)$.
\end{lem}

\begin{proof}
    \newcommand{\fsuffix}{\textit{suffix}_F}
    By definition of $\vz \varimplies{}{D} \vx$, there is a sequence of variables $\vx_1, \vx_2\, \ldots, \vx_n$ with $\vx_1 = \vz$ and $\vx_n = \vx$ such that for each pair of adjacent variables $\vx_i$ and $\vx_{i+1}$:
    \begin{itemize}
        \item[$(\dagger)$] $\vx_i$ and $\vx_{i+1}$ both occur in the same \shortptrans{} $\tau \in \transcopies{D}$ and $\vx_{i+1} = f(\vx_i) \in \Cset(\tau)$ for some function name $f$; or 
        \item[$(\ddagger)$] $\type{\vx_i} = \type{\vx_{i+1}}$ and there is a \pconflictquadruple{} $(\tau_j, o_j, p_k, \tau_k)$ in $D$
        with either $\myvar{o_j} = \vx_i$ and $\myvar{p_k} = \vx_{i+1}$ or $\myvar{p_k} = \vx_i$ and $\myvar{o_j} = \vx_{i+1}$.
    \end{itemize}
    Furthermore, the sequence $F$ corresponds to the function names used in applications of $(\dagger)$. Analogously, $\vw \varimplies{}{D} \vy$, implies a sequence of variables $\vy_1, \vy_2\, \ldots, \vy_n$ with $\vy_1 = \vw$ and $\vy_m = \vy$ with the same properties. Notice that the lengths of these two sequences of variables, namely $n$ and $m$, are not necessarily equal to each other and to the length of $F$ due to possible applications of $(\ddagger)$.
    For a variable $\vx_i$ in the sequence $\vx_1, \vx_2, \ldots, \vx_n$, we denote the sequence of function names derived from applications of $(\dagger)$ in the subsequence $\vx_i, \ldots, \vx_n$ by $\fsuffix(\vx_i)$. Notice that $\fsuffix(\vx_i)$ is indeed always a suffix of $F$, and that $\fsuffix(\vx_1) = \fsuffix(\vy_1)=F$ and $\fsuffix(\vx_n)=\fsuffix(\vy_n) = \varepsilon$.
    
    We argue by induction that for every $i\in[1,n]$ and $j\in [1,m]$, if $\fsuffix(\vx_i)=\fsuffix(\vy_j)$ then $\barmu(\vx_i) = \barmu(\vy_j)$. This then implies $\barmu(\vx) = \barmu(\vx_n) = \barmu(\vy_m) = \barmu(\vy)$.
    \emph{(base case)} Note that $\barmu(\vx_1) = \barmu(\vz) = \barmu(\vw) = \barmu(\vy_1)$ and $\fsuffix(\vx_1) = F = \fsuffix(\vy_1)$.
    \emph{(inductive case)} 
     Let $\fsuffix(\vx_{i+1}) = \fsuffix(\vy_{j+1})$. Then, 
     we distinguish the following cases:
    \begin{itemize}
        \item $\fsuffix(\vx_{i}) = \fsuffix(\vx_{i+1})$: This means that $(\ddagger)$ applies to $\vx_i$ and $\vx_{i+1}$, and there is a \pconflictquadruple{} $(\tau_k, o_k, p_\ell, \tau_\ell)$ in $D$
        with either $\myvar{o_k} = \vx_i$ and $\myvar{p_\ell} = \vx_{i+1}$ or $\myvar{p_\ell} = \vx_i$ and $\myvar{o_k} = \vx_{i+1}$. By definition of $\barmu$, we have $\barmu(\vx_{i+1}) = \barmu(\vx_{i})$ and by induction that $\barmu(\vx_{i+1}) = \barmu(\vy_{j+1})$ implying that 
        $\barmu(\vx_{i+1}) = \barmu(\vx_{j+1})$.

        \item $\fsuffix(\vy_j) = \fsuffix(\vy_{j+1})$: similar as previous argument; 

        \item $\fsuffix(\vx_{i}) \neq \fsuffix(\vx_{i+1})$ and $\fsuffix(\vy_j) \neq \fsuffix(\vy_{j+1})$: Then, $(\dagger)$ applies to both $\vx_i$ and $\vx_{i+1}$, and $\vy_j$ and $\vy_{j+1}$. Furthermore, $\fsuffix(\vx_i) = \fsuffix(\vy_{j}) = f\cdot F'$ for some $f$ and $F'$. By induction, $\barmu(\vx_i)=\barmu(\vy_j)$. Then, $\vx_{i+1} = f(\vx_i)$ is a constraint in some \shortptrans{} $\tau_k \in \transcopies{D}$ and $\vy_{j+1} = f(\vy_j)$ is a constraint in some \shortptrans{} $\tau_\ell\in \transcopies{D}$. 
        Since $\barmu$ is admissible for database $\db$ and $\barmu(\vx_i) = \barmu(\vy_j)$, it follows that $\barmu(\vx_{i+1}) = f^\db(\barmu(\vx_i)) = f^\db(\barmu(\vy_j)) = \barmu(\vy_{j+1})$. \qedhere
    \end{itemize}
\end{proof}

The next Lemma shows that every variable mapping admissible for a given database will assign the same tuple to equivalent variables.

\begin{lem}\label{lemma:barmuvarequal}
    For a multiversion split schedule $\schedule$ based on a \cqsequence{} $C$ over a set of transactions $\transset$ consistent with a set of \ptranss{} $\workload$ and a database $\db$, let $\barmu$ be the variable mapping for a \pcqsequence{} $D$ over $\workload$ with $\barmu(D) = C$.
    Then, for every pair of variables $\vx$ and $\vy$ occurring in \shortptranss{} $\tau_i$ and $\tau_j$ in $\transcopies{D}$ respectively, if $\vx \varequiv{D} \vy$, then $\barmu(\vx) = \barmu(\vy)$.
\end{lem}

\begin{proof}
    The proof is by induction on the definition of $\vx \varequiv{D} \vy$. \emph{(base case)} If $\vx = \vy$, then the result is immediate.
    \emph{(inductive cases)}
    If there are two variables $\vz$ and $\vw$ and a sequence of function names $F$ such that $\vz \varequiv{D} \vw$, $\vz \varimplies{F}{D} \vx$ and $\vw \varimplies{F}{D} \vy$, then by induction we have $\barmu(\vz) = \barmu(\vw)$.
    The proof that $\barmu(\vx) = \barmu(\vy)$ is now immediate by application of Lemma~\ref{lemma:equalbarmu}.
    If instead there is a variable $\vz$ with $\vx \varequiv{D} \vz$ and $\vy \varequiv{D} \vz$, then we can argue by induction that $\barmu(\vx) = \barmu(\vz)$ and $\barmu(\vy) = \barmu(\vz)$, and hence $\barmu(\vx) = \barmu(\vy)$.
\end{proof}

\begin{defi}\label{def:respectingconstraints}
    Let $D$ be a \pcqsequence{}, $\contextset$ a set of tuple-contexts and $\contextfunc$ a partial context assignment for $D$ over $\contextset$.
    We say that $\contextfunc$ \emph{respects the constraints of $D$} if,
    for every two (not necessarily different) variables $\vx$ and $\vy$ occurring in $D$ that $\contextfunc$ is defined for,
    the following conditions are true, where $c_\vx = \contextfunc(\vx)$ and $c_\vy = \contextfunc(\vy)$:
    \begin{enumerate}
        \item\label{d:rc:1} $c_\vx$ is a tuple-context for $\type{\vx}$;
        \item\label{d:rc:2} for every $\vx \vardet{F_1}{D} \vz$ and $\vy \vardet{F_2}{D} \vz$, $c_\vx(F_1) = c_\vy(F_2)$;
        \item\label{d:rc:3} for every $\vx \vardet{F_1}{D} \vw$ and $\vy \vardet{F_2}{D} \vz$ with $\vw \neq \vz$ a constraint in a \shortptrans{} $\tau$, $c_\vx(F_1) \neq c_\vy(F_2)$;
        \item\label{d:rc:4} for every $\vx \vardet{F_1}{D} \vw$ and $\vy \vardet{F_2}{D} \vz$, if $c_\vx(F_1) = c_\vy(F_2)$, then there is no constraint $\vw \neq \vz$ in a \shortptrans{} $\tau \in \transcopies{D}$;
        \item\label{d:rc:5} if $\vx \vardet{F}{D} \vy$, then $c_\vy$ is a tuple-subcontext of $c_\vx$ witnessed by $F$; and
        \item\label{d:rc:6} for every pair of tuple-subcontexts $c_\vx'$ and $c_\vy'$ of $c_\vx$ and $c_\vy$ witnessed by respectively $F_\vx$ and $F_\vy$, if $c_\vx(F_\vx) = c_\vy(F_\vy)$, then $c_\vx' = c_\vy'$.
    \end{enumerate}
\end{defi}

The next Lemma shows that we can represent a counterexample schedule based on a sequence of potentially conflicting quadruples D over an acyclic schema by assigning a tuple-context to each variable in Trans(D).

\begin{lem}\label{lemma:acycschedule:contexts}
    For a multiversion split schedule $\schedule$ based on a \cqsequence{} $C$ over a set of transactions $\transset$ consistent with a set of \ptranss{} $\workload \in \acyctemplates$ and a database $\db$, let $\barmu$ be the variable mapping for a \pcqsequence{} $D$ over $\workload$ with $\barmu(D) = C$.
    Then a set $\contextset$ of tuple-contexts and a total context assignment $\contextfunc$ for $D$ over $\contextset$
    exist with: 
    \begin{itemize}
        \item $\contextfunc$ respects the constraints of $D$; and
        \item $\barmu(\vx) = c_\vx(\epsilon)$ for every variable $\vx$, with $c_\vx = \contextfunc(\vx)$.
    \end{itemize}
\end{lem}

\begin{proof}
    We first assign a tuple-context to each tuple in database $\db$, based on the functions in $\db$. Let $(\schRel, \schFunc)$ be the schema over which $\workload$ is defined.
    Since the schema graph $\sg{\schRel, \schFunc}$ is acyclic, a total order $<_\sgname$ over $\schRel$ exists such that there is no path from type $R$ to type $S$ in $\sgname$ if $R <_\sgname S$.
    We now assign tuple-contexts to tuples based on the order implied by $<_\sgname$. That is, we first consider all tuples of the type that is ordered first by $<_\sgname$, then all tuples of the type that is ordered second, etc. If there are multiple tuples of the same type, the relative order in which we handle them is not important.
    For each tuple $\x$, we construct a tuple-context $c_\x$ with $c_\x(\varepsilon) = \x$, and for each path $F = f \cdot F'$ in $\sgname$ starting in $\type{\x}$, set $c_\x(F) = c_\y(F')$, with $\y = f^\db(\x)$.
    Notice that $c_\y$ is already defined for $\y$, as there is a path from $\type{\x}$ to $\type{\y}$ in $\sgname$ and, hence, $\type{\y} <_\sgname \type{\x}$.
    By construction, $c_\y$ is a tuple-subcontext of $c_\x$ witnessed by $f$.
    
    Next, we construct $\contextfunc$ as follows: $\contextfunc(\vx) = c_\x$ with $\barmu(\vx) = \x$ for every variable $\vx$ occurring in $\transcopies{D}$.
    We argue by induction on the definition of  $\vardet{F}{D}$ that
    \begin{center}
    \hfill $c_\x(F) = \barmu(\vy)$ for every $\vx \vardet{F}{D} \vy$ (with $c_\x = \contextfunc(\vx)$). \hfill $(\dagger)$
    \end{center}
    If $\vx \varimplies{F}{D} \vy$, then
    by construction of $\contextfunc$ and since $\barmu$ is admissible for $\db$, we have $c_\x(F) = \barmu(\vy)$. If $F = \varepsilon$ and $\vx \varequiv{D} \vy$, then $c_\x(\varepsilon) = \barmu(\vx) = \barmu(\vy)$ by Lemma~\ref{lemma:barmuvarequal}.
    Otherwise, if there exists a variable $\vz$ with $\vx \vardet{F_1}{D} \vz$, $\vz \vardet{F_2}{D} \vy$ and $F = F_1 \cdot F_2$,
    then by induction $c_\x(F_1) = \barmu(\vz) = \y$ and $c_\y(F_2) = \barmu(\vy)$, with $\contextfunc(\vz) = c_\y$. By construction of $c_\x$ and $c_\y$, the desired $c_\x(F) = \barmu(\vy)$ now follows.
    
    It remains to verify that $\contextfunc$ indeed satisfies all required properties. By construction, $\barmu(\vx) = c_\x(\varepsilon)$ with $c_\x = \contextfunc(\vx)$, so we only need to show that $\contextfunc$ respects the constraints of $D$ by verifying all properties in Definition~\ref{def:respectingconstraints}.
    To this end, let $\vx$ and $\vy$ be two variables occurring in $\transcopies{D}$, and let $c_\x = \contextfunc(\vx)$ and $c_\y = \contextfunc(\vy)$.
    \emph{(\ref{d:rc:1})} By construction, $c_\x$ is a tuple-context for $\type{\vx}$.
    \emph{(\ref{d:rc:2})} $c_\x(F_1) = \barmu(\vz) = c_\y(F_2)$ by $(\dagger)$.
    \emph{(\ref{d:rc:3})} If $\vw \neq \vz$ is a constraint, then $\barmu(\vw) \neq \barmu(\vz)$ as $\barmu$ is admissible for $\db$. By $(\dagger)$, it follows that $c_\x(F_1) = \barmu(\vw) \neq \barmu(\vz) = c_\y(F_2)$.
    \emph{(\ref{d:rc:4})} If $c_\x(F_1) = c_\y(F_2)$, then $\barmu(\vw) = c_\x(F_1) = c_\y(F_2) = \barmu(\vz)$ by $(\dagger)$. Since $\barmu$ is admissible for $\db$, there cannot be a constraint $\vw \neq \vz$.
    \emph{(\ref{d:rc:5})} If $\vx \vardet{F}{D} \vy$, then by construction of the tuple-contexts $c_\x$ and $c_\y$ it follows that $c_\y$ is a tuple-subcontext of $c_\x$ witnessed by $F$.
    \emph{(\ref{d:rc:6})} Let $c_{\x'}$ and $c_{\y'}$ be tuple-subcontexts of $c_{\x}$ and $c_{\y}$ witnessed by respectively $F_\vx$ and $F_\vy$.
    If $c_{\x}(F_\vx) = c_{\y}(F_\vy) = \z$ for some tuple $\z$, then by construction of $c_{\x}$ and $c_{\y}$ we have $c_{\x'} = c_{\y'} = c_\z$, with $c_\z$ the tuple-context assigned to this tuple $\z$.
\end{proof}

From $D = (\tau_1, o_1, p_2, \tau_2), \ldots, (\tau_m, o_m, p_1, \tau_1)$ and $\contextfunc$ as in Lemma~\ref{lemma:acycschedule:contexts}
we can derive a sequence of quintuples $E = (\tau_1, o_1, c_{o_1}, p_1, c_{p_1}), \ldots, (\tau_m, o_m, c_{o_m}, p_m, c_{p_m})$ such that $c_{o_i} = \contextfunc(\vx_i)$ (respectively $c_{p_i} = \contextfunc(\vy_i)$) for $i \in [1,m]$ with $o_i$ (respectively $p_i$) an operation over variable $\vx_i$ (respectively $\vy_i$). Intuitively, this sequence of quintuples can be used to reconstruct the original multiversion split schedule $\schedule$.
To this end, notice that we can derive the original \pcqsequence{} $D$ and a partial context assignment $\contextfunc'$ from $E$ that is defined for each variable $\vx_i$ occurring in either an operation $o_i$ or $p_i$ in $\tau_i$. We first show that we can extend this partial context assignment $\contextfunc'$ to a total context assignment respecting the constraints in $D$ (Lemma~\ref{lemma:extendcontextassignment}),
and then prove that such a total context assignment respecting the constraints in $D$ implies a variable assignment $\barmu$ such that the $C = \barmu(D)$ is a valid \cqsequence{} (Lemma~\ref{lemma:contextimpliesbarmu}).

\begin{lem}\label{lemma:extendcontextassignment}
    Let $D = (\tau_1, o_1, p_2, \tau_2), \ldots, (\tau_m, o_m, p_1, \tau_1)$ be a \pcqsequence{} over a set of \ptranss{} $\workload \in \acyctemplates$ and $\contextfunc$ a partial context assignment defined for every variable $\vx_i$ of $o_i$ and $\vy_i$ of $p_i$ in every $\tau_i$.
    If
    \begin{itemize}
        \item $\contextfunc$ respects the constraints of $D$; and
        \item for every pair of variables $\vx$ and $\vy$ in a \shortptrans{} $\tau_i$ with $\vx \varequiv{\tau_i} \vy$, there is no constraint $\vx \neq \vy$ in $\tau_i$;
    \end{itemize}
    then we can extend $\contextfunc$ to a total context assignment $\contextfunc'$ for $D$ respecting the constraints of $D$.
\end{lem}

\begin{proof}
    \newcommand{\varequivclass}[1]{[#1]}
    By definition of equivalence, $\varequiv{D}$ partitions all variables occurring in $\transcopies{D}$ in \emph{equivalence classes}. That is, two variables $\vx$ and $\vy$ are in the same equivalence class iff $\vx \varequiv{D} \vy$. For a given variable $\vx$, we denote the equivalence class $\vx$ belongs to by $\varequivclass{\vx}$.
    Note that for any pair of variables $\vx$ and $\vy$ occurring  in $\transcopies{D}$, if $\vx \vardet{F}{D} \vy$, then $\vx' \vardet{F}{D} \vy'$ for any pair of variables $\vx' \in \varequivclass{\vx}$ and $\vy' \in \varequivclass{\vy}$.
    By slight abuse of notation, we use $\vx \vardet{F}{D} \varequivclass{\vy}$ and $\varequivclass{\vx} \vardet{F}{D} \vy$ to denote that
    $\vx \vardet{F}{D} \vy'$ for every $\vy' \in \varequivclass{\vy}$ and $\vx' \vardet{F}{D} \vy$ for every $\vx' \in \varequivclass{\vx}$, respectively.
    
    Let $(\schRel, \schFunc)$ be the schema over which $\workload$ is defined.
    Since the schema graph $\sg{\schRel, \schFunc}$ is acyclic, a total order $<_\sgname$ over $\schRel$ exists such that there is no path from type $R$ to type $S$ in $\sgname$ if $R <_\sgname S$.
    We now define $\contextfunc'$ for variables in $\transcopies{D}$ according to the order implied by $<_\sgname$. If there are multiple variables of the same type, the relative order in which we handle them is not important.

    The proof is as follows. Assume $\contextfunc$ respects the constraints of $D$ and is at least defined for every variable $\vx_i$ of $o_i$ and $\vy_i$ of $p_i$ in every $\tau_i$.
    We extend $\contextfunc$ towards $\contextfunc'$ by defining $\contextfunc'$ for the whole equivalence class $\varequivclass{\vx}$ of the first (according to $<_\sgname$) variable $\vx$ for which $\contextfunc$ is not defined.
    The precise construction is by case. In the first case, the tuple-context that should be assigned to variables in $\varequivclass{\vx}$ is already implied, as it is the tuple-subcontext of an existing tuple-context. In the second case, we construct a fresh tuple-context, including existing tuple-contexts as tuple-subcontexts where we need to make sure that $\contextfunc'$ respects the constraints in $D$.
    In each case, we then argue that $\contextfunc'$ still respects the constraints in $D$.
    By repeating this argument, we can extend the context assignment to a total context assignment defined for all variables occurring in $\transcopies{D}$.

    \medskip
    \noindent
    \emph{(Case 1)} If a variable $\vy$ exists with $\contextfunc$ defined for $\vy$ and $\vy \vardet{F}{D} \varequivclass{\vx}$, then $\contextfunc'(\vx') = c_{\vx'}$ for every variable $\vx' \in \varequivclass{\vx}$, with $c_{\vx'}$ the tuple-subcontext of $c_\vy = \contextfunc(\vy)$ witnessed by $F$.
    Notice that this is well defined, even if there are multiple such $\vy$, as they all agree on $c_{\vx'}$ by Definition~\ref{def:respectingconstraints}~\emph{(\ref{d:rc:2},~\ref{d:rc:6})}.
    Also note that the special case where $\contextfunc$ is already defined for at least one variable $\vx' \in \varequivclass{\vx}$ is covered by this case as well, as $\vx' \vardet{\varepsilon}{D} \varequivclass{\vx}$ follows from $\vx' \in \varequivclass{\vx}$. In this special case, the tuple-subcontext of $\contextfunc(\vx')$ witnessed by $\varepsilon$ (i.e., $\contextfunc(\vx')$ itself) will be assigned to each variable in $\varequivclass{\vx}$.
    
    We show that $\contextfunc'$ indeed respects the constraints in $D$ according to the properties stated in Definition~\ref{def:respectingconstraints}.
    To this end, let $\vx'$ and $\vy'$ be two variables, with $c_{\vx'} = \contextfunc'(\vx')$ and $c_{\vy'} = \contextfunc'(\vy')$.
    \emph{(\ref{d:rc:1})} By construction, $c_{\vx'}$ is a tuple-context for $\type{\vx'}$.
    \emph{(\ref{d:rc:2}-\ref{d:rc:4})} Note that if $\vx'' \vardet{F'}{D} \vz''$ with $\vx'' \in \varequivclass{\vx}$, then $\vy \vardet{F \cdot F'}{D} \vz''$
    and $c_{\vx''} = \contextfunc'(\vx'')$ is the tuple-subcontext of $c_\vy = \contextfunc(\vy)$ witnessed by $F'$, implying that $c_{\vx''}(F') = c_{\vy}(F \cdot F')$. If $\vx'$ and/or $\vy'$ are in $\varequivclass{\vx}$, then we can apply this substitution and use the fact that $\contextfunc$ respects the constraints in $\tau$ to conclude that the desired properties hold for $\contextfunc'$.
    \emph{(\ref{d:rc:5})} Assume $\vx' \vardet{F'}{D} \vy'$. If both $\vx' \in \varequivclass{\vx}$ and $\vy' \in \varequivclass{\vx}$, then $F' = \varepsilon$ as otherwise the schema graph is not acyclic.
    Since $c_{\vy'} = c_{\vx'}$, it follows that $c_{\vy'}$ is a tuple-subcontext of $c_{\vx'}$ witnessed by $\varepsilon$.
    If $\vx' \in \varequivclass{\vx}$ and $\vy' \not\in \varequivclass{\vx}$, then $\vy \vardet{F \cdot F'}{\tau} \vy'$ and $c_{\vy'}$ is a tuple-subcontext of $c_{\vy}$ witnessed by $F \cdot F'$ as $\contextfunc$ respects the constraints of $\tau$.
    Since $c_{\vx'}$ is the tuple-subcontext of $c_\vy$ witnessed by $F$, it follows that $c_{\vy'}$ is a tuple-subcontext of $c_{\vx'}$ witnessed by $F'$.
    If $\vx' \not\in \varequivclass{\vx}$ and $\vy' \in \varequivclass{\vx}$, then $\vy \vardet{F}{\tau} \vy'$.
    Since $\contextfunc$ respects the constraints in $D$, we apply Definition~\ref{def:respectingconstraints}~\emph{(\ref{d:rc:2},~\ref{d:rc:5},~\ref{d:rc:6})} to conclude that $c_{\vy'}$ is a tuple-subcontext of $c_{\vx'}$ witnessed by $F'$.
    \emph{(\ref{d:rc:6})} Assume $c_{\vx'}(F_{\vx'}) = c_{\vy'}(F_{\vy'})$, and let $c_{\vx''}$ and $c_{\vy''}$ be the tuple-subcontexts of respectively $c_{\vx'}$ witnessed by $F_{\vx'}$ and $c_{\vy'}$ witnessed by $F_{\vy'}$.
    We argue that $c_{\vx''} = c_{\vy''}$.
    Note that, if $\vx' \in \varequivclass{\vx}$, then $c_{\vx''}$ is the tuple-subcontext of $c_\vy$ witnessed by $F \cdot F_{\vx'}$. The reasoning for $\vy' \in \varequivclass{\vx}$ is analogous. Since $\contextfunc$ respects the constraints in $D$, it follows that $c_{\vx''} = c_{\vy''}$.
    
    \medskip
    \noindent
    \emph{(Case 2)} Otherwise, we construct a fresh tuple-context $c_\vx$ and define $\contextfunc'(\vx') = c_\vx$ for every variable $\vx' \in \varequivclass{\vx}$. This tuple-context $c_\vx$ is constructed as follows: $c_\vx(\varepsilon) = \x_\vx$, with $\x_\vx$ a fresh tuple of the appropriate type.
    For every path $F = f \cdot F'$ in $\sgname$ starting in $\type{\vx}$, if there is a variable $\vy$ with $\varequivclass{\vx} \vardet{f}{D} \vy$, then $c_\vx(F) = c_\vy(F')$, with $c_\vy = \contextfunc(\vy)$. In other words, $c_\vy$ is the tuple-subcontext of $c_\vx$ witnessed by $f$.
    Note that due to the order $<_\sgname$, $\contextfunc(\vy)$ has to be defined already.
    Also note that this is well defined, even if multiple such $\vy$ exist. In that case, all these $\vy$ are equivalent to each other by definition of $\varequiv{D}$, and by construction of $\contextfunc$ they are assigned the same tuple-context. If instead no such variable $\vy$ exists, we define $c_\vx(F) = \x_F$, with $\x_F$ a fresh tuple of the appropriate type.
    
    We show that $\contextfunc'$ indeed respects the constraints in $D$ according to the properties stated in Definition~\ref{def:respectingconstraints}.
    To this end, let $\vx'$ and $\vy'$ be two variables occurring in $\transcopies{D}$, with $c_{\vx'} = \contextfunc'(\vx')$ and $c_{\vy'} = \contextfunc'(\vy')$.
    \emph{(\ref{d:rc:1})} By construction, $c_{\vx'}$ is a tuple-context for $\type{\vx'}$.
    \emph{(\ref{d:rc:2})} Assume $\vx' \vardet{F_1}{D} \vz$ and $\vy' \vardet{F_2}{D} \vz$ for some variable $\vz$.
    We argue that there exists a pair of variables $\vx''$ and $\vy''$ and two sequences of function names $F_1'$ and $F_2'$ such that $c_{\vx'}(F_1) = c_{\vx''}(F_1') = c_{\vy''}(F_2') = c_{\vy'}(F_2)$, where $c_{\vx''} = \contextfunc'(\vx'')$ and $c_{\vy''} = \contextfunc'(\vy'')$.
    If $\vx' \in \varequivclass{\vx}$,
    then either $F_1 = f \cdot F_1'$ or $F_1 = \varepsilon$. In the former case there is a variable $\vx''$ with $\vx'' \vardet{f}{D} \vz$ such that $c_{\vx'}(F_1) = c_{\vx''}(F_1')$, where $c_{\vx''} = \contextfunc(\vx'')$.
    In the later case, $\vz \in \varequivclass{\vx}$, and we simply take $\vx'' = \vx'$ and $F_1' = F_1$.
    If $\vx' \not\in \varequivclass{\vx}$, we take $\vx'' = \vx'$ and $F_1' = F_1$.
    For $\vy' \in \varequivclass{\vx}$ and $\vy' \in \not\varequivclass{\vx}$, the reasoning is analogous.
    It remains to argue that $c_{\vx''}(F_1') = c_{\vy''}(F_2')$. By choice of $\vx''$ and $\vy''$, either $\vx'' \not\in \varequivclass{\vx}$ and $\vy'' \not\in \varequivclass{\vx}$; or $\vz \in \varequivclass{\vx}$.
    In te former case, $c_{\vx''}(F_1') = c_{\vy''}(F_2')$ follows by the fact that $\contextfunc$ respects the constraints of $D$. In the latter case, both $\vx'' \in \varequivclass{\vx}$ and $\vy'' \in \varequivclass{\vx}$, as otherwise (Case 1) would apply to $\varequivclass{\vx}$ instead. Then, $c_{\vx''}(F_1') = c_{\vy''}(F_2') = c_{\vx}(\varepsilon) = \x_\vx$.
    \emph{(\ref{d:rc:3},~\ref{d:rc:4})} The reasoning is analogous to the previous property. Note in particular that by construction of the new $c_\vx$ we have $c_{\vx'}(F_1) = c_{\vy'}(F_2)$ if $\vw \varequiv{D} \vz$.
    Since $\vw \varequiv{D} \vz$ implies that there is no constraint $\vw \neq \vz$ by the assumptions on $\contextfunc$ and on the disequality constraints in each \shortptrans{} $\tau \in \transcopies{D}$, this does not lead to contradictions.
    \emph{(\ref{d:rc:5})} If $\vx' \vardet{F'} \vy'$, then $\vy' \in \varequivclass{\vx}$ only if $\vx' \in \varequivclass{\vx}$, as otherwise (Case 1) would apply to $\varequivclass{\vx}$ instead. We argue by case that $c_{\vy'}$ is a tuple-subcontext of $c_{\vx'}$ witnessed by $F'$.
    If $\vx' \not\in \varequivclass{\vx}$ and $\vy' \not\in \varequivclass{\vx}$, the result is immediate by the fact that $\contextfunc$ respects the constraints of $D$.
    If $\vx' \in \varequivclass{\vx}$ and $\vy' \not\in \varequivclass{\vx}$, then $c_{\vx'} = c_\vx$ and a variable $\vz$ exists such that $F' = f \cdot F''$, $\vx' \vardet{f}{D} \vz$, $\vz \vardet{F''}{D} \vy'$,
    and, by construction of $c_\vx$, $c_{\vy'}$ is a tuple-subcontext of $\contextfunc(\vz)$ witnessed by $F''$. It now follows that $c_{\vy'}$ is a tuple-subcontext of $c_\vx$ witnessed by $F'$.
    Lastly, If both $\vx' \in \varequivclass{\vx}$ and $\vy' \in \varequivclass{\vx}$, then $F' = \varepsilon$, as otherwise the schema graph is not acyclic. The result is immediate, as $c_{\vy'} = c_{\vx'} = c_{\vx}$ is by definition a tuple-subcontext of itself witnessed by $\varepsilon$.
    \emph{(\ref{d:rc:6})} Assume $c_{\vx''}(F_1) = c_{\vy''}(F_1)$ for some pair of tuple-contexts $c_{\vx''}$ and $c_{\vy''}$ that are tuple-subcontexts of respectively $c_{\vx'}$ witnessed by $F_1$ and $c_{\vy'}$ witnessed by $F_2$. We argue that $c_{\vx''} = c_{\vy''}$.
    If both $c_{\vx'}$ and $c_{\vy'}$ are different from $c_\vx$, the result is immediate as $\contextfunc$ respects the constraints of $D$.
    Otherwise, since the construction of $c_\vx$, either copies existing tuple-contexts as tuple-subcontexts, or introduces fresh variables. the result holds if $c_{\vx'}$ and/or $c_{\vy'}$ are equal to $c_\vx$.
\end{proof}

Given a total context assignment respecting the constraints in $D$ as in Lemma~\ref{lemma:extendcontextassignment}, we show in the next Lemma that such a total context assignment implies a variable assignment $\barmu$ such that the $C = \barmu(D)$ is a valid \cqsequence{}.

\begin{lem}\label{lemma:contextimpliesbarmu}
    Let $D = (\tau_1, o_1, p_2, \tau_2), \ldots, (\tau_m, o_m, p_1, \tau_1)$ be a \pcqsequence{} over a set of \ptranss{} $\workload$ and $\contextfunc$ a total context assignment for $D$ respecting the constraints of $D$.
    The variable mapping $\barmu$ obtained by defining $\barmu(\vx) = c_\vx(\epsilon)$ for every variable $\vx$ in $\transcopies{D}$ with $c_\vx = \contextfunc(\vx)$ then is 
    a valid variable mapping admissible for some database $\db$.
\end{lem}

\begin{proof}
    We first argue that $\barmu$ is valid by showing for each \conflictquadruple{} $(\tau_i, o_i, p_j, \tau_j)$ in $D$ that $\barmu(\vx) = \barmu(\vy)$ with $\vx = \myvar{o_i}$ and $\vy = \myvar{p_j}$.
    By definition, $\vx \varimplies{\varepsilon}{D} \vy$, and hence $\vx \vardet{\varepsilon}{D} \vy$. Let $c_\vx = \contextfunc(\vx)$ and $c_\vy = \contextfunc(\vy)$.
    Since $\contextfunc$ respects the constraints of $D$, $c_\vy$ is a tuple-subcontext of $c_\vx$ witnessed by $\varepsilon$.
    By definition of tuple-subcontexts, $c_\vx(\varepsilon) = c_\vy(\varepsilon)$, and, as a result, $\barmu(\vx) = c_\vx(\varepsilon) = c_\vy(\varepsilon) = \barmu(\vy)$.
    
    Next, we construct a database $\db$ and show that $\barmu$ is admissible for $\db$. To this end, we add the tuple $\barmu(\vx)$ to $\db$ for each variable $\vx$ occurring in $\transcopies{D}$. 
    For each 
    functional constraint $\vy = f(\vx)$ in a \ptrans{} in $\transcopies{D}$, we define $\barmu(\vy) = f^\db(\barmu(\vx))$ for the corresponding function $f^\db$ in $\db$. 
    Note that this is well defined.
    Towards a contradiction, assume that we have both $\barmu(\vy) = f^\db(\barmu(\vx))$ and $\barmu(\vw) = f^\db(\barmu(\vz))$, with $\barmu(\vx) = \barmu(\vz)$ but $\barmu(\vy) \neq \barmu(\vw)$.
    Let $c_\vx = \contextfunc(\vx)$, $c_\vy = \contextfunc(\vy)$, $c_\vz = \contextfunc(\vz)$ and $c_\vw = \contextfunc(\vw)$.
    By construction of $\barmu$, we have $c_\vx(\varepsilon) = \barmu(\vx) = \barmu(\vz) = c_\vz(\varepsilon)$ and $c_\vy(\varepsilon) = \barmu(\vy) \neq \barmu(\vw) = c_\vw(\varepsilon)$.
    By Definition~\ref{def:respectingconstraints}~\emph{(\ref{d:rc:6})}, it now follows that $c_\vx = c_\vz$, since $c_\vx$ (respectively $c_\vz$) is a tuple-subcontext of itself witnessed by $\varepsilon$ and $c_\vx(\varepsilon) = c_\vz(\varepsilon)$.
    Since we defined $\barmu(\vy) = f^\db(\barmu(\vx))$, there is a constraint $\vy = f(\vx)$ in some \shortptrans{} in $\transcopies{D}$, and hence $\vx \vardet{f}{D} \vy$. Analogously, $\vz \vardet{f}{D} \vw$.
    By Definition~\ref{def:respectingconstraints}~\emph{(\ref{d:rc:5})}, $c_\vy$ and $c_\vw$ are tuple-subcontexts of respectively $c_\vx$ and $c_\vz$ witnessed by $f$.
    As $c_\vx = c_\vz$, it immediately follows that $c_\vy = c_\vw$, and in particular $c_\vy(\varepsilon) = c_\vw(\varepsilon)$, leading to the desired contradiction.
    
    To conclude the proof, we show that $\barmu$ is indeed admissible for $\db$.
    By construction of $\db$ based on $\barmu$, $\barmu(\vy) = f^\db(\barmu(\vx))$ is immediate for each constraint $\vy = f(\vx)$ in a \shortptrans{} $\tau \in \transcopies{D}$.
    We still need to argue that $\barmu(\vx) \neq \barmu(\vy)$ for each constraint $\vx \neq \vy$ in a \shortptrans{} $\tau \in \transcopies{D}$.
    Let $c_\vx = \contextfunc(\vx)$ and $c_\vy = \contextfunc(\vy)$.
    By construction of $\barmu$ we have $\barmu(\vx) = c_\vx(\varepsilon)$ and $\barmu(\vy) = c_\vy(\varepsilon)$.
    Note that $\vx \vardet{\varepsilon}{D} \vx$ and $\vy \vardet{\varepsilon}{D} \vy$.
    Therefore, we can apply Definition~\ref{def:respectingconstraints}~\emph{(\ref{d:rc:3})} to conclude that $c_\vx(\varepsilon) \neq c_\vy(\varepsilon)$, and hence $\barmu(\vx) \neq \barmu(\vy)$.
\end{proof}

In order to decide robustness against \mvrc, one can now construct a sequence of quintuples $E$ and derive the \pcqsequence{} $D$ and partial context assignment $\contextfunc$ from it.
If $\contextfunc$ respects the constraints in $D$, then it follows from Lemma~\ref{lemma:extendcontextassignment} and Lemma~\ref{lemma:contextimpliesbarmu} that we can construct a variable assignment $\barmu$ such that $C = \barmu(D)$ is a valid \cqsequence{}.
However, in this construction of $E$, care should be taken to guarantee that $\contextfunc$ indeed respects the constraints in $D$, and that the resulting multiversion split schedule based on $C$ indeed satisfies all properties in Definition~\ref{def:mvsplitschedule}.

In the algorithm that we are about to propose, we search for such a sequence of quintuples $E$, but without fixating all the tuples in each context.
For this, we generalize our definition of tuple-contexts to allow variables: A \emph{context for a type $R$} is a function from paths with source $R$ in  $\sg{\schRel, \schFunc}$ to variables in $\variables$ and 
tuples in $\objects$ of the appropriate type. The purpose of variables is to encode equalities and disequalities within each context, without being explicit about the precise tuples. That is, if two paths ending in the same node in $\sgname$ are mapped on the same variable, then they will represent the same tuple; if they are mapped on different variables, then they represent a different tuple. We remark that a same variable occurring in different contexts can still represent different tuples.
Analogous to tuple-subcontexts, for two types $R$ and $S$ with $R \asgimplies{F}{\sgname} S$, we say that a context $c_{S}$ for type $S$ is a \emph{subcontext} of a context $c_R$ for type $R$ witnessed by $F$ if:
\begin{itemize}
    \item for every path $S \asgimplies{F'}{\sgname} S'$ in $\sgname$, if $c_{R}(F \cdot F')$ is a tuple, then $c_{S}(F') = c_{R}(F \cdot F')$; otherwise, $c_{S}(F')$ is a variable; and
    \item for every pair of paths $S \asgimplies{F_1}{\sgname} S_1$ and $S \asgimplies{F_2}{\sgname} S_2$ in $\sgname$ with $c_{R}(F \cdot F_1)$ and $c_{R}(F \cdot F_2)$ variables, $c_{S}(F_1) = c_{S}(F_2)$ iff $c_{R}(F \cdot F_1) = c_{R}(F \cdot F_2)$.
\end{itemize}
We call a context a \emph{variable-context} if all paths are mapped on variables.

For a \ptrans{} $\tau$, tuple-context $c_p$ for $p$ and $c_o$ for $o$ in $\tau$, we consider the set $\contextdomain[SG,\tau,p,c_p, o,c_o]$ of all different (not-necessarily tuple-) contexts $c$ (up to isomorphisms over the variables in $c$) that can be obtained, starting from a variable-context $c'$, by performing substitutions of subcontexts of $c'$ with subcontexts of $c_p$ and/or $c_o$.
More formally, these substitutions are of the form:
For a path $R_c \asgimplies{F}{\sgname} S$ (here $R_c$ is the type that $c$ is for) and $R_p \asgimplies{F'}{\sgname} S$ (with $R_p$ the type that $c_p$ is for) then $c(F \cdot F'') = c_p(F' \cdot F'')$ for every path $S \asgimplies{F''}{\sgname} S'$ in $\sgname$ and with $c(F \cdot F'') = c'(F \cdot F'')$ otherwise. (The substitution rule can be applied for $c_p$ as well as for $c_o$.)

\begin{lem}\label{lemma:asgcharacterization}
    Let $\workload$ be a set of \ptranss{} over an acyclic schema.
    Then, $\workload$ is not robust against \mvrc if, and only if, there is a sequence of quintuples $E = (\tau_1, o_1, c_{o_1}, p_1, c_{p_1}),\allowbreak \ldots,\allowbreak (\tau_m, o_m, c_{o_m}, p_m, c_{p_m})$
    with $m \geq 2$ such that for each quintuple $(\tau_i, o_i, c_{o_i}, p_i, c_{p_i})$ in $E$, with $q_i$ and $r_i$ two (not necessarily different) operations in $\{o_i, p_i\}$,
    \newcounter{resumecountertwo}
    \begin{enumerate}
        \item\label{asgc:1} if $i = 1$, then $c_{o_1}$ and $c_{p_1}$ are tuple-contexts for $\type{\myvar{o_1}}$ and $\type{\myvar{p_1}}$.
        Furthermore, for every pair of tuple-subcontexts $c_{o_1}'$ and $c_{p_1}'$ of $c_{o_1}$ and $c_{p_1}$ witnessed by respectively $F$ and $F'$, if $c_{o_1}(F) = c_{p_1}(F')$, then $c_{o_1}' = c_{p_1}'$;
        \item\label{asgc:2} if $i \neq 1$, then $c_{o_i}, c_{p_i} \in \contextdomain[\sgname,\tau_1,p_1,c_{p_1}, o_1,c_{o_1}]$ are contexts for $\type{\myvar{o_i}}$ and $\type{\myvar{p_i}}$;
        \item\label{asgc:3} for every pair of variables $\vw_i$ and $\vz_i$ in $\tau_i$ with $\vw_i \varequiv{\tau_i} \vz_i$, there is no constraint $\vw_i \neq \vz_i$ in $\tau_i$;
        \item\label{asgc:4} for every $\myvar{q_i} \vardet{F_1}{\tau_i} \vz_i$ and $\myvar{r_i} \vardet{F_2}{\tau_i} \vz_i$, {the subcontext of $c_{q_i}$ witnessed by $F_1$ is equal (up to isomorphisms over variables) to the subcontext of $c_{r_i}$ witnessed by $F_2$.}
        \item\label{asgc:5} for every $\myvar{q_i} \vardet{F_1}{\tau_i} \vw_i$ and $\myvar{r_i} \vardet{F_2}{\tau_i} \vz_i$ with $\vw_i \neq \vz_i$ a constraint in $\tau_i$,
        $c_{q_i}(F_1) \neq c_{r_i}(F_2)$ or $c_{q_i}(F_1)$ and $c_{r_i}(F_2)$ are both variables;
        \item\label{asgc:6} for every $\myvar{q_i} \varimplies{F_1}{\tau_i} \vw_i$ and $\myvar{r_i} \varimplies{F_2}{\tau_i} \vz_i$, if $c_{q_i}(F_1)$ and $c_{r_i}(F_2)$ are the same tuple, then there is no constraint $\vw_i \neq \vz_i$ in $\tau_i$;
        \item\label{asgc:7} if $\myvar{q_i} \vardet{F}{\tau_i} \myvar{r_i}$, then $c_{r_i}$ is a subcontext of $c_{q_i}$ witnessed by $F$; and
        \item\label{asgc:8} If $i \neq 1$ and $c_{q_i}(F) = c_{q_1}(F')$ is a tuple for some $q_1 \in \{o_1, p_1\}$ and some sequence of function names $F$ and $F'$,
        then there is no operation $o_i' \in \tau_i$ potentially ww-conflicting with an operation $o_1' \in \prefix{\tau_1}{o_1}$ with
        $\myvar{q_i} \vardet{F}{\tau_i} \myvar{o_i'}$ and $\myvar{q_1} \vardet{F'}{\tau_1} \myvar{o_1'}$.
        \setcounter{resumecountertwo}{\value{enumi}}
    \end{enumerate}
    Furthermore, for each pair of adjacent quintuples $(\tau_i, o_i, c_{o_i}, p_i, c_{p_i})$ and $(\tau_j, o_j, c_{o_j}, p_j, c_{p_j})$ in $E$ with $j = i + 1$, or $i = m$ and $j = 1$:
    \begin{enumerate}
        \setcounter{enumi}{\value{resumecountertwo}}
        \item\label{asgc:9} $o_i$ is potentially conflicting with $p_j$ and $c_{o_i} = c_{p_j}$;
        \item\label{asgc:10} if $i = 1$ and $j = 2$, then $o_1$ is potentially rw-conflicting with $p_2$; and
        \item\label{asgc:11} if $i = m$ and $j = 1$, then $o_1 <_{\tau_1} p_1$ or $o_m$ is potentially rw-conflicting with $p_1$.
    \end{enumerate}
\end{lem}

\begin{proof}
    \emph{(if)} Let $D = (\tau_1, o_1, p_2, \tau_2), \ldots, (\tau_m, o_m, p_1, \tau_1)$ be the \pcqsequence{} derived from $E$.
    Note that each $(\tau_i, o_i, p_j, \tau_j) \in D$ is indeed a valid \pcqsequence{}, as $o_i$ is potentially conflicting with $p_j$ by \emph{(\ref{asgc:9})}.
    We show in Lemma~\ref{lemma:asgcharifhelper} that a partial context assignment $\contextfunc$ over a set of tuple-contexts $\contextset$ exists such that
    \begin{itemize}
        \item for every pair of operations $o_i$ and $p_i$ occurring in $D$, $\contextfunc$ is defined for  $\myvar{o_i}$ and $\myvar{p_i}$;
        \item $\contextfunc$ respects the constraints in $D$; and
        \item for every \shortptrans{} $\tau_i$ in $D$ with $i \neq 1$ and for every $q_i \in \{o_i, p_i\}$ and $q_1 \in \{o_1, p_1\}$, let $c_{q_i} = \contextfunc(\myvar{q_i})$ and  $c_{q_1} = \contextfunc(\myvar{q_1})$.
        If $c_{q_i}(F) = c_{q_1}(F')$ for some sequence of function names $F$ and $F'$,
        then there is no operation $o_i' \in \tau_i$ potentially ww-conflicting with an operation $o_1' \in \prefix{\tau_1}{o_1}$ with
        $\myvar{q_i} \vardet{F}{\tau_i} \myvar{o_i'}$ and $\myvar{q_1} \vardet{F'}{\tau_1} \myvar{o_1'}$.
    \end{itemize}
    Because of \emph{(\ref{asgc:3})}, we can now apply Lemma~\ref{lemma:extendcontextassignment} extending $\contextfunc$ to a total context assignment defined for all variables occurring in $\transcopies{D}$, without losing the property that $\contextfunc$ respects all constraints in $D$.
    Let $\barmu$ be the variable mapping obtained by defining $\barmu(\vx) = c_\vx(\varepsilon)$ for every variable $\vx$ in $\transcopies{D}$ with $c_\vx = \contextfunc(\vx)$.
    By Lemma~\ref{lemma:contextimpliesbarmu}, $\barmu$ is a valid variable mapping and a database $\db$ exists such that $\barmu$ is admissible for $\db$.
    
    We now prove that the \cqsequence{} $C = \barmu(D)$ satisfies the conditions stated in Definition~\ref{def:mvsplitschedule} to show that $\workload$ is indeed not robust against \mvrc.
    Condition~(\ref{c:2}) and Condition~(\ref{c:3}) are immediate by respectively \emph{(\ref{asgc:11})} and \emph{(\ref{asgc:10})}.
    Towards a contradiction, assume Condition~(\ref{c:1}) does not hold.
    Then, there is an operation $o_i'$ in a \shortptrans{} $\tau_i$ potentially ww-conflicting with an operation $o_1' \in \prefix{\tau_1}{o_1}$,
    and $\barmu(\myvar{o_i'}) = \barmu(\myvar{o_1'})$. Let $c_{o_i'} = \contextfunc(\myvar{o_i'})$ and $c_{o_1'} = \contextfunc(\myvar{o_1'})$.
    By construction of $\barmu$, we have $c_{o_i'}(\varepsilon) = c_{o_1'}(\varepsilon)$.
    By construction of the total context assignment in Lemma~\ref{lemma:extendcontextassignment}, this is only the case if for some $q_i \in \{o_i, p_i\}$ and $q_1 \in \{o_1, p_1\}$
    it holds that $\myvar{q_i} \vardet{F}{\tau_i} \myvar{o_i'}$ with $c_{q_i}(F) = c_{o_i'}(\varepsilon)$ and $\myvar{q_1} \vardet{F}{\tau_1} \myvar{o_1'}$ with $c_{q_1}(F') = c_{o_1'}(\varepsilon)$.
    Consequently, $c_{q_i}(F) = c_{q_1}(F')$, contradicting Lemma~\ref{lemma:asgcharifhelper}.

    \medskip
    \emph{(only if)} Since $\workload$ is not robust against \mvrc, a multiversion split schedule $\schedule$ exists
    based on a \cqsequence{} $C$ over a set of transactions $\transset$ consistent with a set of \ptranss{} $\workload \in \acyctemplates$ and a database $\db$.
    Let $\barmu$ be the variable mapping for a \pcqsequence{} $D = (\tau_1, o_1, p_2, \tau_2), \ldots, (\tau_m, o_m, p_1, \tau_1)$ with $\barmu(D) = C$.
    By Lemma~\ref{lemma:acycschedule:contexts} a set $\contextset$ of tuple-contexts and a total context assignment $\contextfunc$ for $D$ over $\contextset$
    exist with:
    \begin{itemize}
        \item $\contextfunc$ respects the constraints of $D$; and
        \item $\barmu(\vx) = c_\vx(\epsilon)$ for every variable $\vx$, with $c_\vx = \contextfunc(\vx)$.
    \end{itemize}
    
    Let $\lambda: \objects \to \objects \cup \variables$ with $\lambda(\x) = \x$ if $\x$ occurs in $\contextfunc(\myvar{o_1})$ or $\contextfunc(\myvar{p_1})$;
    and $\lambda(\x) \in \variables$ otherwise, such that $\lambda(\x) \neq \lambda(\y)$ if $\x \neq \y$.
    From $D$ and $\contextfunc$ we derive the sequence of quintuples $E = (\tau_1, o_1, c_{o_1}, p_1, c_{p_1}), \ldots, (\tau_m, o_m, c_{o_m}, p_m, c_{p_m})$ with $c_{o_i} = \lambda \circ c'_{o_i}$ and $c_{p_i} = \lambda \circ c'_{p_i}$ for each $o_i$ and $p_i$,
    where $c'_{o_i} = \contextfunc(\myvar{o_i})$ and $c'_{p_i} = \contextfunc(\myvar{p_i})$.
    Intuitively, we modify the tuple-contexts for each $\myvar{o_i}$ and $\myvar{p_i}$ as defined by $\contextfunc$ into contexts over tuples and variables by replacing all tuples that do not occur in $\contextfunc(\myvar{o_1})$ and $\contextfunc(\myvar{p_1})$ with unique variables.
    Note that by construction $c_{o_1} = \contextfunc(\myvar{o_1})$ and $c_{p_1} = \contextfunc(\myvar{p_1})$.
    
    Next, we show that $E$ satisfies all properties.
    In the argumentation below, we denote $\contextfunc(\myvar{o_i})$ by $c'_{o_i}$ and $\contextfunc(\myvar{p_i})$ by $c'_{p_i}$ for each $o_i$ and $p_i$ in $E$.
    \emph{(\ref{asgc:1})} By construction, $c_{o_1} = c'_{o_1}$ is a tuple-context for $\type{\myvar{o_1}}$, and $c_{p_1} = c'_{p_1}$ is a tuple-context for $\type{\myvar{p_1}}$.
    By Condition~(\ref{d:rc:6}) of Definition~\ref{def:respectingconstraints}, we have $c_{o_1}' = c_{p_1}'$ if $c_{o_1}(F) = c_{p_1}(F')$
    for every pair of tuple-subcontexts $c_{o_1}'$ and $c_{p_1}'$ of $c_{o_1}$ and $c_{p_1}$ witnessed by respectively $F$ and $F'$.
    \emph{(\ref{asgc:2})}~This property follows by construction of $c_{o_i}$ and $c_{p_i}$ based on $\lambda$ and $\contextfunc$. For completeness sake, one should note that for each group of contexts up to isomorphisms over variables, $\contextdomain[\sgname,\tau_1,p_1,c_{p_1}, o_1,c_{o_1}]$
    contains only one context.
    W.l.o.g.\ we can implicitly assume that each $c_{o_i}$ and $c_{p_i}$ in $E$ is replaced by the same context in $\contextdomain[\sgname,\tau_1,p_1,c_{p_1}, o_1,c_{o_1}]$ up to isomorphisms,
    as we will never directly test for equality or disequality between two variables of different contexts.
    \emph{(\ref{asgc:3})} Assume $\vw_i \varequiv{\tau_i} \vz_i$, then $\vw_i \vardet{\varepsilon}{D} \vz_i$.
    Since $\contextfunc$ respects the constraints in $D$, $c_{\vz_i}$ is a tuple-subcontext of $c_{\vw_i}$ witnessed by $\varepsilon$, with $c_{\vz_i} = \contextfunc(\vz_i)$ and $c_{\vw_i} = \contextfunc(\vw_i)$.
    Then, $c_{\vz_i} = c_{\vw_i}$, and in particular $c_{\vz_i}(\varepsilon) = c_{\vw_i}(\varepsilon)$.
    By Lemma~\ref{lemma:acycschedule:contexts}, $\barmu(\vz_i) = c_{\vz_i}(\varepsilon) = c_{\vw_i}(\varepsilon) = \barmu(\vw_i)$. Since $\barmu$ is admissible for $\db$, the constraint $\vw_i \neq \vz_i$ cannot exist.
    \emph{(\ref{asgc:4})}~Since $\contextfunc$ respects the constraints in $D$, {it follows from Conditions~(\ref{d:rc:2})~and~(\ref{d:rc:6}) in Definition~\ref{def:respectingconstraints} that the tuple-subcontext of $c'_{q_i}$ witnessed by $F_1$ is equal to the tuple-subcontext of $c'_{r_i}$ witnessed by $F_2$.
    As a result, the subcontext of $c_{q_i} = \lambda \circ c'_{r_i}$ witnessed by $F_1$ is equal (up to isomorphisms over variables) to the subcontext of $c_{r_i} = \lambda \circ c'_{r_i}$ witnessed by $F_2$.}
    \emph{(\ref{asgc:5})} Analogous to the previous case, we can conclude that $c_{q_i}(F_1) = \lambda \circ c'_{q_i}(F_1)$ and $c_{r_i}(F_2) = \lambda \circ c'_{r_i}(F_2)$ are either two different tuples, or both a variable.
    \emph{(\ref{asgc:6})} If $c_{q_i}(F_1) = \lambda \circ c'_{q_i}(F_1) = \lambda \circ c'_{r_i}(F_2) = c_{r_i}(F_2)$ is a tuple, then $c'_{q_i}(F_1) = c'_{r_i}(F_2)$. Since $\contextfunc$ respects the constraints in $D$, it follows that there is no constraint $\vw_i \neq vz_i$.
    \emph{(\ref{asgc:7})} If $\myvar{q_i} \vardet{F}{\tau_i} \myvar{r_i}$, then $c'_{r_i}$ is a tuple-subcontext of $c'_{q_i}$, as $\contextfunc$ respects the constraints of $D$. By construction of $c_{q_i}$ and $c_{r_i}$ based on respectively $c'_{q_i}$ and $c'_{r_i}$, it immediately follows that $c_{r_i}$ is a subcontext of $c_{q_i}$.
    The case for $\vy_i \vardet{F}{\tau_i} \vx_i$ is analogous.
    \emph{(\ref{asgc:8})}~Assume towards a contradiction that $c_{q_i}(F) = c_{q_1}(F')$ is a tuple for some $q_1 \in \{o_1, p_1\}$ and some sequence of function names $F$ and $F'$,
    and there is an operation $o_i' \in \tau_i$ potentially ww-conflicting with an operation $o_1' \in \prefix{\tau_1}{o_1}$ with
    $\myvar{q_i} \vardet{F}{\tau_i} \myvar{o_i'}$ and $\myvar{q_1} \vardet{F'}{\tau_1} \myvar{o_1'}$.
    Since $c_{q_i}(F) = c_{q_1}(F')$ is a tuple, $c'_{q_i}(F) = c_{q_i}(F) = c_{q_1}(F') = c'_{q_1}(F')$.
    By definition of $\barmu$ and since $\contextfunc$ respects the constraints in $D$, we conclude that $\barmu(o_i')$ in $\barmu(\tau_i)$ is ww-conflicting with $\barmu(o_1')$ in $\prefix{\barmu(\tau_1)}{\barmu(o_1)}$, thereby contradicting Condition~(\ref{c:1}) of Definition~\ref{def:mvsplitschedule}.
    \emph{(\ref{asgc:9})} Since $E$ is based on $C$, the operation $o_i$ is potentially conflicting with $p_j$. Furthermore, since $\myvar{o_i} \varequiv{D} \myvar{p_j}$ and since $\contextfunc$ respects the constraints of $D$, $c'_{o_i} = c'_{p_j}$, and hence $c_{o_i} = c_{p_j}$.
    \emph{(\ref{asgc:10})} Immediate by Condition~(\ref{c:3}) of Definition~\ref{def:mvsplitschedule}.
    \emph{(\ref{asgc:11})} Immediate by Condition~(\ref{c:2}) of Definition~\ref{def:mvsplitschedule}.
\end{proof}

Central to the correctness of the if-direction of Lemma~\ref{lemma:asgcharacterization} is the observation that for any sequence $E$ satisfying the stated conditions, we can assign tuples to variables in each context, thereby replacing contexts with tuple-contexts, in such a way that the resulting context assignment over tuple-contexts respects the constraints of the corresponding \pcqsequence{} $D$. This observation is formalized in the following Lemma:

\begin{lem}\label{lemma:asgcharifhelper}
    Let $E = (\tau_1, o_1, c_{o_1}, p_1, c_{p_1}), \ldots, (\tau_m, o_m, c_{o_m}, p_m, c_{p_m})$ be a sequence of quintuples satisfying the conditions stated in Lemma~\ref{lemma:asgcharacterization}, and let $D = (\tau_1, o_1, p_2, \tau_2), \ldots, (\tau_m, o_m,\allowbreak p_1, \tau_1)$ be the \pcqsequence{} derived from $E$.
    Then a partial context assignment $\contextfunc$ over a set of tuple-contexts $\contextset$ exists such that
    \begin{itemize}
        \item for every pair of operations $o_i$ and $p_i$ occurring in $D$, $\contextfunc$ is defined for  $\myvar{o_i}$ and $\myvar{p_i}$;
        \item $\contextfunc$ respects the constraints in $D$; and
        \item for every \shortptrans{} $\tau_i$ in $D$ with $i \neq 1$ and for every $q_i \in \{o_i, p_i\}$ and $q_1 \in \{o_1, p_1\}$, let $c_{q_i} = \contextfunc(\myvar{q_i})$ and  $c_{q_1} = \contextfunc(\myvar{q_1})$.
        If $c_{q_i}(F) = c_{q_1}(F')$ for some sequence of function names $F$ and $F'$,
        then there is no operation $o_i' \in \tau_i$ potentially ww-conflicting with an operation $o_1' \in \prefix{\tau_1}{o_1}$ with
        $\myvar{q_i} \vardet{F}{\tau_i} \myvar{o_i'}$ and $\myvar{q_1} \vardet{F'}{\tau_1} \myvar{o_1'}$.
    \end{itemize}
\end{lem}

\begin{proof}
    The general proof idea is as follows. We iteratively extend $\contextfunc$ by deriving $\contextfunc(\myvar{o_i})$ and $\contextfunc(\myvar{p_i})$ from contexts $c_{o_i}$ and $c_{p_i}$ for each quintuple $(\tau_i, o_i, c_{o_i}, p_i, c_{p_i})$ in the order that they appear in $E$.
    Afterwards, we argue that $\contextfunc$ respects the constraints in $D$
    and for every $q_i \in \{o_i, p_i\}$ with $i \neq 1$ and $q_1 \in \{o_1, p_1\}$ with $c_{q_i} = \contextfunc(\myvar{q_i})$ and  $c_{q_1} = \contextfunc(\myvar{q_1})$, and for every pair of sequences of function names $F$ and $F'$ with $c_{q_i}(F) = c_{q_1}(F')$,
    there is no operation $o_i' \in \tau_i$ potentially ww-conflicting with an operation $o_1' \in \prefix{\tau_1}{o_1}$ with
    $\myvar{q_i} \vardet{F}{\tau_i} \myvar{o_i'}$ and $\myvar{q_1} \vardet{F'}{\tau_1} \myvar{o_1'}$.
    
    Let $(\tau_1, o_1, c_{o_1}, p_1, c_{p_1})$ be the first quintuple in the sequence $E$. We initiate $\contextfunc$ by defining $\contextfunc(\myvar{o_1}) = c_{o_1}$ and $\contextfunc(\myvar{p_1}) = c_{p_1}$.
    Next, we iteratively extend $\contextfunc$ by considering the remaining quintuples in $E$ in order.
    To this end, let $(\tau_{i-1}, o_{i-1}, c_{o_{i-1}}, p_{i-1}, c_{p_{i-1}})$ be the last considered quintuple and $(\tau_i, o_i, c_{o_i}, p_i, c_{p_i})$ the next quintuple in $E$.
    We define $\contextfunc(\myvar{p_i}) = c'_{p_i} = \contextfunc(\myvar{o_{i-1}})$ and $\contextfunc(\myvar{o_i}) = c'_{o_i} = \lambda_i \circ c_{o_i}$,
    where $\lambda_i: \objects \cup \variables \to \objects$ is a function mapping tuples and variables occurring in $c_{o_i}$ to tuples such that\footnote{Note that $\lambda_i$ is defined over variables in the context $c_{o_i}$. These variables are unrelated to the variables occurring in $\transcopies{D}$. We therefore denote these variables by $\vq$ instead of the usual $\vw, \vx, \vy, \vz$ to avoid confusion.}
    \begin{itemize}
        \item $\lambda_i(\x) = \x$ for each tuple $\x$ occurring in $c_{o_i}$;
        \item $\lambda_i(\vq) = \y$ for each variable $\vq$ occurring in $c_{o_i}$ for which there is a variable $\vz$ in $\tau_i$ and sequences of function names $F$, $F'$ and $F''$
        with $\myvar{p_i} \vardet{F}{\tau_i} \vz$, $\myvar{o_i} \vardet{F'}{\tau_i} \vz$, $c'_{p_i}(F \cdot F'') = \y$ and $c_{o_i}(F' \cdot F'') = \vq$; and
        \item $\lambda_i(\vq) = \x_{i,\vq}$, for the remaining variables $\vq$ in $c_{o_i}$ where $\x_{i,\vq}$ is a fresh tuple.
    \end{itemize}
    Note that this $\lambda_i$ is well defined. In particular, the second rule intuitively states that the tuple-subcontext of the resulting $c'_{o_i}$ witnessed by $F'$ is equal to the tuple-subcontext of $c'_{p_i}$ witnessed by $F$, given that there is a variable $\vz$ with $\myvar{o_i} \vardet{F'}{\tau_i} \vz$ and $\myvar{p_i} \vardet{F}{\tau_i} \vz$.
    This substitution is well defined since in this case the subcontext of $c_{o_i}$ witnessed by $F'$ is equal (up to isomorphisms over variables) to the subcontext of $c_{p_i}$ witnessed by $F$ according to Condition~\ref{asgc:4} of Lemma~\ref{lemma:asgcharacterization}.
    
    It remains to show that $\contextfunc$ indeed satisfies the conditions stated in Lemma~\ref{lemma:asgcharifhelper}. To this end, note that by definition of variable determination,
    if $\vx \vardet{F}{D} \vy$ with $\vx$ in a \shortptrans{} $\tau_i$ and $\vy$ in a \shortptrans{} $\tau_i$ in $\transcopies{D}$,
    then a sequence of variables $\vx_{k_1}, \vy_{k_1}, \ldots, \vx_{k_m}, \vy_{k_m}$ exists such that $(\dagger)$:
    \begin{itemize}
        \item $\vx_{k_1} = \vx$ and $\vy_{k_m} = \vy$;
        \item each pair of (not necessarily different) variables $\vx_{k_i}, \vy_{k_i}$ occur in the same \shortptrans{} $\tau_{k_i}$ in $\transcopies{D}$ and $\vx_{k_i}, \vardet{F_i}{\tau_{k_i}} \vy_{k_i}$ with $F = F_1 \cdot \ldots, F_m$;
        \item in the implied sequence of \shortptranss{} $\tau_{k_1}, \ldots, \tau_{k_m}$, these $\tau_{k_i}, \ldots, \tau_{k_{i+1}}$ are neighbouring in $E$ (where we assume that $\tau_1$ is neighbouring to $\tau_n$ in $E$); and
        \item for each pair of variables $\vy_{k_i}, \vx_{k_{i+1}}$, there is a sequence of function names $F'$ such that $\myvar{o_{k_i}} \vardet{F'}{\tau_{k_i}} \vy_{k_i}$ and $\myvar{p_{k_{i+1}}} \vardet{F'}{\tau_{k_{i+1}}} \vx_{k_{i+1}}$
        (i.e., equivalence of $\vy_{k_i}$ and $\vx_{k_{i+1}}$ in $D$ is implied by equivalence of $\myvar{o_{k_i}}$ and $\myvar{p_{k_{i+1}}}$).
    \end{itemize}
    In other words, $\vx \vardet{F}{D} \vy$ can be broken down into a sequence of $\vx_{k_i} \vardet{F_i}{\tau_{k_i}} \vy_{k_i}$ through a sequence of neighbouring \shortptranss{}, where equivalence between each $\vy_{k_i}$ and $\vx_{k_{i+1}}$ is implied by the variables in the potentially conflicting operations $o_{k_i}$ and $p_{k_{i+1}}$.
    For ease of exposition, we implicitly assumed that $\tau_{k_1}, \ldots, \tau_{k_m}$ agrees with the order in $E$. If the order is opposite to the order in $E$ instead, the above still holds, but the occurrences of $o_{k_i}$ and $p_{k_{i+1}}$ should be replaced with $p_{k_i}$ and $o_{k_{i+1}}$.
    
    We argue by construction of $\contextfunc$ that for every pair of variables $\vx$ and $\vy$ for which $\contextfunc$ is defined with $c_\vx = \contextfunc(\vx)$ and $c_\vy = \contextfunc(\vy)$,
    if $c_\vx(F) = c_\vy(F')$ for some sequence of function names $F$ and $F'$, then $c'_\vx = c'_\vy$, where $c'_\vx$ and $c'_\vy$ are the tuple-subcontexts of $c_\vx$ witnessed by $F$ and $c_\vy$ witnessed by $F'$, respectively $(\ddagger)$.
    If $\vx = \myvar{o_1}$ and $\vy =\myvar{p_1}$ (or the other way around), the result is immediate by Lemma~\ref{lemma:asgcharacterization}~(\ref{asgc:1}).
    Otherwise, let $\vx$ be the variable in a \shortptrans{} $\tau_i$ and $\vy$ the variable in a \shortptrans{} $\tau_j$ such that $j \leq i$ (i.e., $\tau_i$ does not occur before $\tau_j$ in $E$).
    W.l.o.g., we assume that $\vx = \myvar{o_i}$ with $i \neq 1$ (the case where $\vx = \myvar{p_i}$ is analogous, as $\contextfunc(\myvar{p_i}) = \contextfunc(\myvar{o_{i-1}})$ by construction).
    By construction of each $c'_{o_i}$ based on $c'_{p_i}$ and $\lambda_i$, if $c'_{o_i}(F_i) = c'_{p_i}(F_i')$, then the whole tuple-subcontext of $c'_{o_i}$ witnessed by $F_i$ is copied over from the tuple-subcontext of $c'_{p_i}$ witnessed by $F_i'$. Indeed $\lambda_i$ introduces fresh tuples whenever the tuple for $c'_{o_i}(F_i)$ is not implied by $c'_{p_i}$.
    
    The desired properties now follow from $(\dagger)$ and $(\ddagger)$ as well as the conditions in Lemma~\ref{lemma:asgcharacterization}.
    In particular, $\contextfunc$ respecting the constraints of $D$ can now be derived from Conditions~(\ref{asgc:1},~\ref{asgc:2},~\ref{asgc:4}-\ref{asgc:7}) in Lemma~\ref{lemma:asgcharacterization},
    and the last condition of Lemma~\ref{lemma:asgcharifhelper} follows from Condition~(\ref{asgc:8}) in Lemma~\ref{lemma:asgcharacterization}. 
\end{proof}

\begin{proof}[Proof of Theorem~\ref{theo:asgexpspace}]
A \NEXPSPACE algorithm proving the correctness of Theorem~\ref{theo:asgexpspace} is now immediate by Lemma~\ref{lemma:asgcharacterization}, as we can iteratively guess and verify quintuples in $E$ while only keeping track of the very first quintuple and the previous quintuple. Since in an acyclic schema graph the number of paths starting in a given type is at most exponential in the total number of types, each context is defined over at most an exponential number of paths. However, to formally argue that these contexts can be encoded in exponential space, we still need to show that each tuple or variable used in a context can be encoded in at most exponential space. Since the only tuples used are those mentioned in $c_{o_1}$ and $c_{p_1}$, and since we can reuse the same variables over all contexts, both the maximal number of tuples and the maximal number of variables needed are exponential in the total number of types.
\end{proof}

\subsection{Lowering complexity}

Next, we consider restrictions that lower the complexity. To this end, 
we say that two variables $\vx$ and $\vy$ occurring in a \ptrans{} $\tau$ are \emph{equivalent in $\tau$}, denoted $\vx \varequiv{\tau} \vy$ if
\begin{itemize}
    \item $\vx = \vy$;
    \item there exists a pair of variables $\vz$ and $\vw$ in $\tau$ and a sequence of function names $F$ with $\vz \varequiv{\tau} \vw$, $\vz \varimplies{F}{\tau} \vx$ and $\vw \varimplies{F}{\tau} \vy$; or
    \item there exists a variable $\vz$ with $\vx \varequiv{\tau} \vz$ and $\vy \varequiv{\tau} \vz$.
\end{itemize}
Then, a \ptrans{} $\tau$ is \emph{\restricted} if for every combination of variables $\vx, \vy, \vw, \vz$ in $\tau$ with $\vx \varimplies{}{\tau} \vw$ and $\vy \varimplies{}{\tau} \vz$,
either $\vw \varequiv{\tau} \vz$, $\vw \varimplies{}{\tau} \vz$ or $\vz \varimplies{}{\tau} \vw$.
We denote by $\acycrestemplates$ the class of all sets of \restricted \ptranss{} over acyclic schemas.

\begin{thm}\label{theo:asgpspace:exptime}
\begin{enumerate}
        \item \RobustnessTempTwo{\acycrestemplates}{\mvrc} is decidable in \EXPTIME.
        \item \RobustnessTempTwo{\acyctemplates}{\mvrc} is decidable in \PSPACE when the number of paths between any two nodes in the schema graph is bounded by a constant $k$.
\end{enumerate}
\end{thm}
\noindent Regarding (1), all templates in TPC-C with the exception of NewOrder are restricted. Regarding (2), when the schema graph is a multi-tree then $k=1$ and for TPC-C $k=2$ (recall that in general there can be an exponential number of paths), leading to a more practical algorithm for robustness in those cases.

The \PSPACE result in Theorem~\ref{theo:asgpspace:exptime} for workloads over a schema $(\schRel, \schFunc)$ where the number of paths between any two nodes in the schema graph is bounded by a constant $k$ is immediate by the nondeterministic algorithm based on Lemma~\ref{lemma:asgcharacterization} presented for Theorem~\ref{theo:asgexpspace}. Indeed, in this case, the total number of paths starting in a given type is at most $k . |\schRel|$ and therefore each context is defined over at most a polynomial number of paths, instead of an exponential number of paths for the general case in Theorem~\ref{theo:asgexpspace}.

The \EXPTIME result in Theorem~\ref{theo:asgpspace:exptime} for workloads in \acycrestemplates follows from a deterministic algorithm based on Lemma~\ref{lemma:asgcharacterization}. In the remainder of this section, we first present the algorithm, and then discuss its complexity.

\paragraph*{A deterministic algorithm}

Towards a deterministic algorithm, assume the first quintuple $(\tau_1, o_1, c_{o_1}, p_1, c_{p_1})$ of $E$ is fixed. We now translate the problem of deciding whether we can extend $E$ such that it satisfies all properties to a graph problem over a graph $G(\tau_1, o_1, c_{o_1}, p_1, c_{p_1})$. This graph is constructed as follows:
\begin{itemize}
    \item each quintuple $(\tau_i, o_i, c_{o_i}, p_i, c_{p_i})$ satisfying Conditions~(\ref{asgc:2}-\ref{asgc:8}) of Lemma~\ref{lemma:asgcharacterization} is added as a node to $G(\tau_1, o_1, c_{o_1}, p_1, c_{p_1})$; and
    \item there is an edge from a node $(\tau_i, o_i, c_{o_i}, p_i, c_{p_i})$ to a node $(\tau_j, o_j, c_{o_j}, p_j, c_{p_j})$ if $o_i$ is potentially conflicting with $p_j$ and $c_{o_i} = c_{p_j}$ (c.f.\ Condition~(\ref{asgc:9}) of Lemma~\ref{lemma:asgcharacterization}).
\end{itemize}
By construction, it is now easy to see that there is a sequence $E$ satisfying Lemma~\ref{lemma:asgcharacterization} if there is a path from a quintuple $(\tau_2, o_2, c_{o_2}, p_2, c_{p_2})$ to a quintuple $(\tau_m, o_m, c_{o_m}, p_m, c_{p_m})$ in $G(\tau_1, o_1, c_{o_1}, p_1, c_{p_1})$ (where we allow a zero-length path with $2 = m$), such that $(\dagger)$
\begin{itemize}
    \item $c_{o_1} = c_{p_2}$ and $c_{o_m} = c_{p_1}$ (c.f.\ Condition~(\ref{asgc:9}) of Lemma~\ref{lemma:asgcharacterization});
    \item $o_1$ is potentially rw-conflicting with $p_2$ (c.f.\ Condition~(\ref{asgc:10}) of Lemma~\ref{lemma:asgcharacterization}); and
    \item $o_1 <_{\tau_1} p_1$ or $o_m$ is potentially rw-conflicting with $p_1$ (c.f.\ Condition~(\ref{asgc:11}) of Lemma~\ref{lemma:asgcharacterization}).
\end{itemize}

Given a set of \ptranss{} $\workload$ over a schema $(\schRel, \schFunc)$,
the algorithm iterates over all possible quintuples $(\tau_1, o_1, c_{o_1}, p_1, c_{p_1})$ satisfying Condition~(\ref{asgc:1},~\ref{asgc:3}-\ref{asgc:7}) of Lemma~\ref{lemma:asgcharacterization},
where we consider all possible tuple-contexts $c_{o_1}$ and $c_{p_1}$ up to isomorphisms.
For each such quintuple, the graph $G(\tau_1, o_1, c_{o_1}, p_1, c_{p_1})$ is constructed. Let $TC$ be the reflexive-transitive closure of $G$.
If there is a pair of quintuples $(\tau_2, o_2, c_{o_2}, p_2, c_{p_2})$ and $(\tau_m, o_m, c_{o_m}, p_m, c_{p_m})$ in $TC$ satisfying $(\dagger)$, the algorithm emits a reject, indicating that $\workload$ is not robust against $\mvrc$.
Otherwise, it proceeds with a new choice for $(\tau_1, o_1, c_{o_1}, p_1, c_{p_1})$. If, the algorithm didn't reject after considering all such quintuples, it accepts, indicating that $\workload$ is indeed robust against \mvrc.
The correctness of this algorithm is immediate by Lemma~\ref{lemma:asgcharacterization}.

\paragraph*{Complexity analysis}

We show the complexity of the presented algorithm.
For this, first, notice that we have defined contexts $c$ based on a type $S$ (with $S$ not necessarily a root of \sgname{}). For encoding purposes it makes sense to encode these as contexts for a root type $R$ in combination with the intended type $S$. The context as defined in the previous section can then be derived by taking the left-most subtree with root $S$. Notice that this is purely an encoding choice that will simplify the analysis.

\newcommand{\pathr}[1][R,S]{\text{Paths}_\text{SG}(#1)}
\newcommand{\pathrcount}{{|\pathr|}}
\newcommand{\roots}{\textsf{roots}}

For a schema graph $\sg{\schRel, \schFunc}$ the total number of non-isomorphic tuple-contexts can be expressed using Bell's number $B(n)$, denoting the number of partitions for a set of size $n$, and the set $\pathr = \{ (R,F,S) \mid R \asgimplies{F}{\sgname} S \}$ expressing the different paths from one node $R$ to another node $S$ in $\sgname$. 
Concretely, 
\[
    \ssize{\tcontextdomain[SG]} \le \sum_{R\in \roots(\sgname)}\prod_{S\in \schRel} B(\pathrcount) = B^*
\]
where $\tcontextdomain[SG]$ denotes the set containing all different tuple-contexts (up to isomorphisms).

\newcommand{\pathcount}{{|\text{Paths}_\text{SG}(R,S)|}}
\newcommand{\indicator}[1][]{\boldsymbol{1}_{#1}}

Now let $c_1$ and $c_2$ be two fixed contexts for types that are descendants of roots $R_1$ and root $R_2$, respectively, in \sgname, and let $c$ be a context for a type descending from root $R$.
To express a bound on the number of substitutions in $c$ from (parts of) $c_1$ and $c_2$, we need some additional terminology: Let $\pathr[R,*] = \bigcup_{S\in \schRel}\pathr$. We say that a path $R \asgimplies{F_1}{\sgname} S$ is a \emph{prefix} of a path $R \asgimplies{F}{\sgname} S'$ in \sgname{} if there is a (possibly empty) sequence of function names $F_2$ with $F = F_1 \cdot F_2$.
The number of substitutions in $c$ from (parts of) $c_1$ and $c_2$ is now bounded by 
\begin{equation*}
    \begin{split}
         \sum_{\text{Part} \subseteq \pathr[R,*]}\indicator[PFP]\prod_{R \asgimplies{F}{\sgname} S \in \text{Part}} \left(|\pathr[R_1, S]| + |\pathr[R_2,S]|\right)
         &\le T^\ell.(2P)^\ell \\
         &\le (2TP)^\ell
    \end{split}
\end{equation*}

In the above expression, $\indicator[PFP]$ is an indicator variable that equals $1$ if no path in Part is a prefix of another path in Part and that equals $0$ otherwise. Further:
$P$ denotes the maximum number of different paths between a particular root and a particular node in \sgname, $T$ denotes the maximum number of different paths from a particular root to nodes in SG, and $\ell$ denotes the maximal size of a set in which no path is a prefix of another path in the set. The latter is trivially bounded by $T$.

A special cases exists if all templates $\tau$ in $\workload$ are \restricted.
In that case, the size of sets $\text{Part}$ is bounded by $2$, hence $\ell \le 2$.

With the above bounds, the complexity of the presented algorithm is rather straightforward.
The iteration over all possible quintuples $(\tau_1, o_1, c_{o_1}, p_1, c_{p_1})$ requires at most $\ssize{\workload} . t^2.\left(B^*\right)^2$ iterations, with $t$ denoting the maximal number of operations in a \ptrans{} of $\workload$.
The remainder of the computation is dominated by the transitive-closure computation. Since the constructed graph $G(\tau_1, o_1, c_{o_1}, p_1, c_{p_1})$ has at most
\[\ssize{\workload}.t^2.\left(B^*.(2TP)^\ell\right)^2\]
nodes, the transitive closure computation requires
\[\left(\ssize{\workload}.t^2.\left(B^*.(2TP)^\ell\right)^2\right)^3\]
steps. Putting these numbers together, we obtain:
\[\mathcal{O}(\ssize{\workload}^4.t^8.\left(B^*\right)^8.(2TP)^{6\ell}).
\]
Since $\ell$ is bounded by a constant if all \shortptrans{} are \restricted, and since $B^*$, $T$ and $P$ can be exponential in the size of the input, the presented algorithm indeed decides \RobustnessTempTwo{\acycrestemplates}{\mvrc} in \EXPTIME.

\section{Related Work}
\label{sec:relwork}

\paragraph*{Transaction Programs} Previous work on static robustness testing~\cite{DBLP:journals/tods/FeketeLOOS05,DBLP:conf/aiccsa/AlomariF15} for transaction programs is based on the following key insight: when a \emph{schedule} is not serializable, then the dependency graph constructed from that schedule contains a cycle satisfying a condition specific to the isolation level at hand (\emph{dangerous structure} for \snapshot and the presence of a \emph{counterflow edge} for \MVRC). That insight is extended to a workload of \emph{transaction programs} through the construction of a so-called static dependency graph where each program is represented by a node, and there is a conflict edge from one program to another if there can be a schedule that gives rise to that conflict.  The absence of a cycle satisfying the condition specific to that isolation level then guarantees robustness while the presence of a cycle does not necessarily imply non-robustness. 

Other work studies robustness within a framework for uniformly specifying different isolation levels in a declarative way~\cite{DBLP:conf/concur/Cerone0G15,DBLP:conf/concur/0002G16,Cerone:2018:ASI:3184466.3152396}. A key assumption here is \emph{atomic visibility} requiring that either all or none of the updates of each transaction are visible to other transactions.
{These approaches aim at higher isolation levels and cannot be used for \MVRC, as \MVRC does not admit \emph{atomic visibility}.}

\paragraph*{Transaction Templates}
The static robustness approach based on transaction templates~\cite{fullversion} differs in two ways. First, it makes more underlying assumptions explicit within the formalism of transaction templates (whereas previous work departs from the static dependency graph that should be constructed in some way by the dba). Second, it allows for a decision procedure that is sound and complete for robustness testing
 against \MVRC, allowing to detect larger subsets of transactions to be robust~\cite{fullversion}. 

 The formalization of transactions and conflict serializability in \cite{fullversion} and this paper is based on \cite{DBLP:conf/pods/Fekete05}, generalized to operations over attributes of tuples and extended with $\myUP$-operations that combine $\myR$- and $\myW$-operations into one atomic operation. These definitions are closely related to the formalization presented by Adya et al.~\cite{DBLP:conf/icde/AdyaLO00}, but we assume a total rather than a partial order over the operations in a schedule. 
 There are also a few restrictions to the model: there needs to be a fixed set of read-only attributes that cannot be updated and which are used to select tuples for update. The most typical example of this are primary key values passed to transaction templates as parameters. 
 The inability to update primary keys is not an important restriction in many workloads, where keys, once assigned, never get changed, for regulatory or data integrity reasons. 

 In \cite{fullversion}, a \PTIME decision procedure is obtained for robustness against RC for templates without functional constraints and the present paper improves that result to \NLOGSPACE. 
 In addition,
 an experimental study was performed showing how an approach based on robustness and making transactions robust through promotion can improve transaction throughput.

 \paragraph*{Transactions}
\new{The work by} Fekete~\cite{DBLP:conf/pods/Fekete05} is the first work that provides a necessary and sufficient condition for deciding robustness against \snapshot for a workload of concrete transactions (not transaction programs). 
That work provides a characterization for acceptable allocations when every transaction runs under either 
\snapshot
or strict two-phase locking (S2PL).
The allocation then is acceptable when every possible execution respecting the allocated isolation levels is serializable. As a side result, this work indirectly provides a necessary and sufficient condition for robustness against \snapshot, since robustness against \snapshot holds iff the allocation where each transaction is allocated to \snapshot is acceptable. 
Ketsman et al.~\cite{DBLP:conf/pods/Ketsman0NV20} provide full characterizations for robustness against \readcom and \readun under lock-based semantics. In addition, it is shown that the corresponding decision problems are complete for \coNP and \LOGSPACE, respectively, which should be contrasted with the polynomial time characterization obtained in \cite{fullversion} for robustness against {\it multiversion} read committed.

\section{Conclusion}
\label{sec:concl}

This paper falls within a more general research line 
investigating how transaction throughput can be improved through an approach based on robustness testing that can be readily applied  without making any changes to the underlying database system. As argued in Section~\ref{sec:example}, incorporating functional constraints can detect larger sets of templates to be robust and requires less $\myR$-operations to be promoted to $\myUP$-operations. In future work, we plan to look at lower bounds,  restrictions that lower complexity, and consider other referential integrity constraints to further enlarge the modelling power of transaction templates.
\new{For example, the current formalism allows to express dependencies between tuples, but it is not suited to express the fact that a transaction accesses \emph{all} tuples depending on a specific tuple, such as all \emph{Order} tuples depending on a specific \emph{Customer} tuple.}

\section*{Acknowledgment}
\noindent This work is partly funded by FWO-grant G019921N.

%%
%% Bibliography
%%

%% Please use bibtex, 

\bibliographystyle{alphaurl}
\bibliography{references}

\appendix

\clearpage

\section{SmallBank Benchmark SQL Code}
\label{sec:app:smallbank}
{
\footnotesize
\begin{minipage}[t]{\textwidth/2-2ex}
\begin{verbatim}
Balance(N):
    SELECT CustomerId INTO :x
      FROM Account
     WHERE Name=:N;
        
    SELECT Balance INTO :a 
      FROM Savings
     WHERE CustomerId=:x;
     
    SELECT Balance + :a 
      FROM Checking
     WHERE CustomerId=:x;
    COMMIT; 

Amalgamate(N1,N2):
    SELECT CustomerId INTO :x1
      FROM Account
     WHERE Name=:N1;
    
    SELECT CustomerId INTO :x2
      FROM Account
     WHERE Name=:N2;
    
    UPDATE Savings AS new
       SET Balance = 0
      FROM Savings AS old
     WHERE new.CustomerId=:x1
           AND old.CustomerId
           = new.CustomerId
    RETURNING old.Balance INTO :a;
    
    UPDATE Checking AS new
       SET Balance = 0
      FROM Checking AS old
     WHERE new.CustomerId=:x1
           AND old.CustomerId
           = new.CustomerId
    RETURNING old.Balance INTO :b;
    
    UPDATE Checking
       SET Balance = Balance + :a + :b
     WHERE CustomerId=:x2;

DepositChecking(N,V):
    SELECT CustomerId INTO :x
      FROM Account
     WHERE Name=:N;
     
    UPDATE Checking
       SET Balance = Balance + :V 
     WHERE CustomerId=:x;
    COMMIT;
\end{verbatim}
\end{minipage}   
\begin{minipage}[t]{\textwidth/2}
\begin{verbatim}
TransactSavings(N,V):
    SELECT CustomerId INTO :x
      FROM Account
     WHERE Name=:N;
        
    UPDATE Savings
       SET Balance = Balance + :V 
     WHERE CustomerId=:x;
    COMMIT;

WriteCheck(N,V):
    SELECT CustomerId INTO :x
      FROM Account
     WHERE Name=:N;
        
    SELECT Balance INTO :a 
      FROM Savings
     WHERE CustomerId=:x;
     
    SELECT Balance INTO :b 
      FROM Checking
     WHERE CustomerId=:x;
     
    IF (:a + :b) < :V THEN 
        UPDATE Checking
           SET Balance = Balance - (:V+1) 
         WHERE CustomerId=:x;
    ELSE
        UPDATE Checking
           SET Balance = Balance - :V
         WHERE CustomerId=:x;
    END IF;
    COMMIT;

GoPremium(N):
    UPDATE Account
       SET IsPremium = TRUE
     WHERE Name=:N
    RETURNING CustomerId INTO :x;
    
    SELECT InterestRate INTO :a 
      FROM Savings
     WHERE CustomerId=:x;
     
    :rate = computePremiumRate(:x,:a);
     
    UPDATE Savings
       SET InterestRate = :rate
     WHERE CustomerId=:x;
    COMMIT;
\end{verbatim}
\end{minipage}
}

\end{document}